\documentclass[12pt, a4paper]{article}

\usepackage[utf8]{inputenc}
\usepackage[T1]{fontenc}
\usepackage[english]{babel}

\usepackage{amsmath}
\usepackage{amssymb}
\usepackage{amsthm}
\usepackage{mathtools}

\usepackage[left=2cm, right=2cm, top=2cm, bottom=2cm]{geometry}
\usepackage{fancyhdr}

\usepackage{enumitem}
\usepackage{graphicx}
\usepackage{wrapfig}
\usepackage{xcolor}
\usepackage{thmtools}
\usepackage{array}
\usepackage{booktabs}
\usepackage{float}
\usepackage{chngcntr}
\usepackage{subcaption}
\usepackage{multirow}
\usepackage{adjustbox}
\usepackage{caption}

\renewcommand{\arraystretch}{1.0}
\counterwithin{figure}{section}
\counterwithin{table}{section}
\numberwithin{equation}{section}

\usepackage{natbib}
\setcitestyle{authoryear,open={(},close={)}}

\usepackage[hidelinks]{hyperref}
\hypersetup{
    colorlinks=true,
    linkcolor=black,
    citecolor=black,
    filecolor=black,
    urlcolor=blue!60!black,
    pdftitle={Regimes in the Order Flow: Duration-Aware and Multivariate Bayesian Online Changepoint Detection for High-Frequency Markets (research internship report)},
    pdfauthor={Ramzi Jebali}
}

\definecolor{defblue}{RGB}{30, 100, 180}
\definecolor{notegray}{RGB}{90, 90, 90}
\definecolor{bglight}{RGB}{248, 249, 250}
\definecolor{propgreen}{RGB}{250, 50, 50}

\declaretheoremstyle[
    headfont=\bfseries\color{defblue},
    notefont=\normalfont\itshape,
    bodyfont=\normalfont,
    headpunct={.},
    postheadspace={0.5em},
    spaceabove=5pt,
    spacebelow=5pt,
    shaded={bgcolor=bglight, margin=0.5em}
]{defstyle}

\declaretheoremstyle[
    headfont=\bfseries\color{notegray},
    notefont=\normalfont\itshape,
    bodyfont=\normalfont,
    headpunct={.},
    postheadspace={0.5em},
    spaceabove=5pt,
    spacebelow=5pt,
    shaded={bgcolor=bglight, margin=0.5em}
]{notestyle}

\declaretheoremstyle[
    headfont=\bfseries\color{propgreen},
    notefont=\normalfont\itshape,
    bodyfont=\normalfont,
    headpunct={.},
    postheadspace={0.5em},
    spaceabove=5pt,
    spacebelow=5pt,
    shaded={bgcolor=bglight, margin=0.5em}
]{propstyle}

\declaretheorem[style=defstyle,  name=Definition,  numberwithin=section]{definition}
\declaretheorem[style=notestyle, name=Notation,    numberwithin=section]{notation}
\declaretheorem[style=propstyle, name=Proposition, numberwithin=section]{proposition}
\declaretheorem[style=defstyle,  name=Definition,  numbered=no]{definition*}
\declaretheorem[style=propstyle, name=Proposition, numbered=no]{proposition*}

\newcommand{\R}{\mathbb{R}}

\newcommand{\E}{\mathbb{E}}
\newcommand{\V}{\mathbb{V}}

\newcommand{\diag}{\text{diag}}
\newcommand{\tr}{\text{tr}}

\newcommand{\argmax}{\text{argmax}}

\newcommand{\indep}{\perp \!\!\! \perp}

\begin{document}

\pagenumbering{arabic}
\setcounter{page}{1}
\pagestyle{empty}

\begin{titlepage}
\thispagestyle{empty}

\noindent
\begin{minipage}[t]{0.48\textwidth}
    \raggedright
    \includegraphics[height=2.4cm]{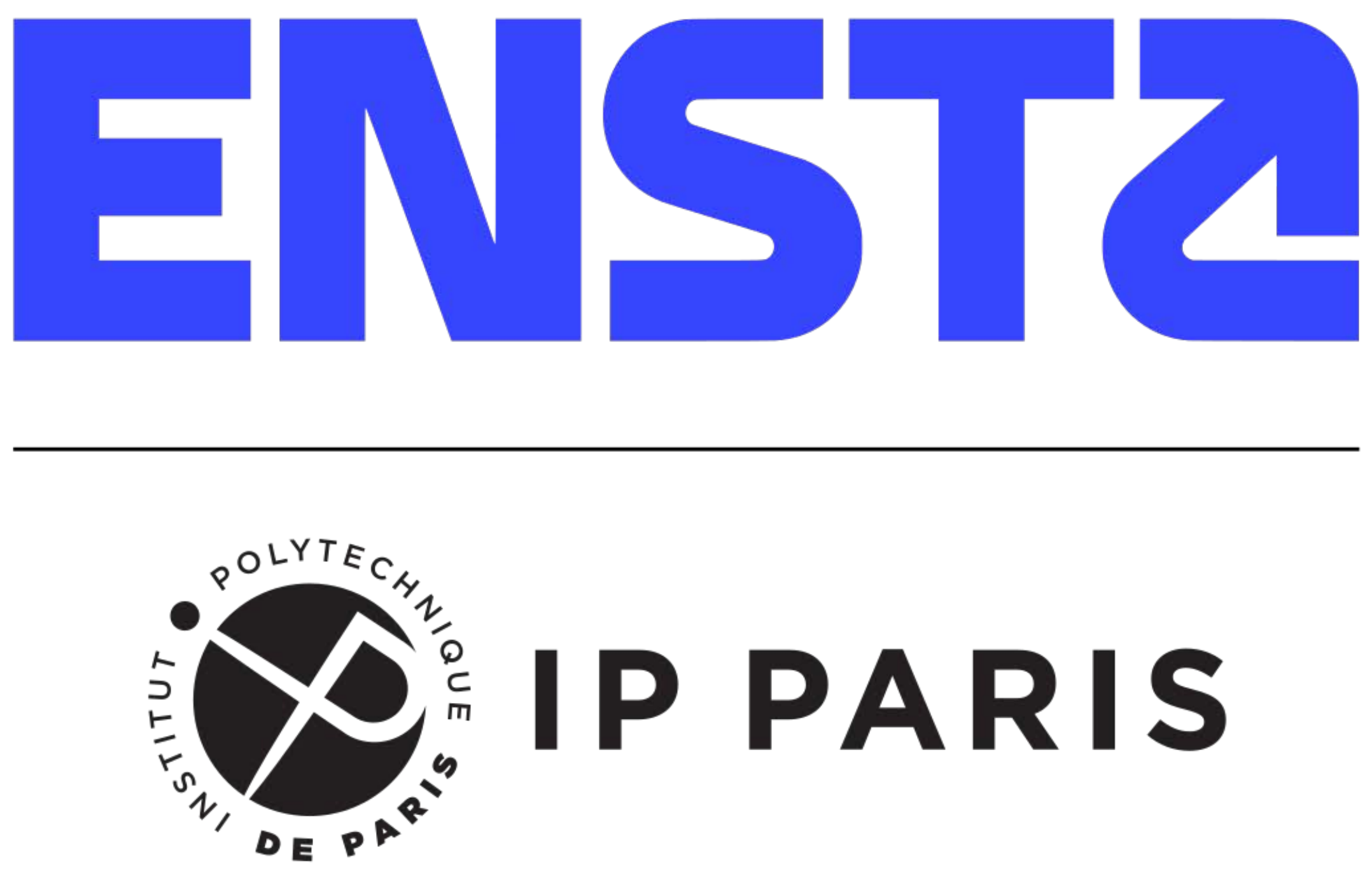}
\end{minipage}
\hfill
\begin{minipage}[t]{0.48\textwidth}
    \raggedleft
    \includegraphics[height=2.5cm]{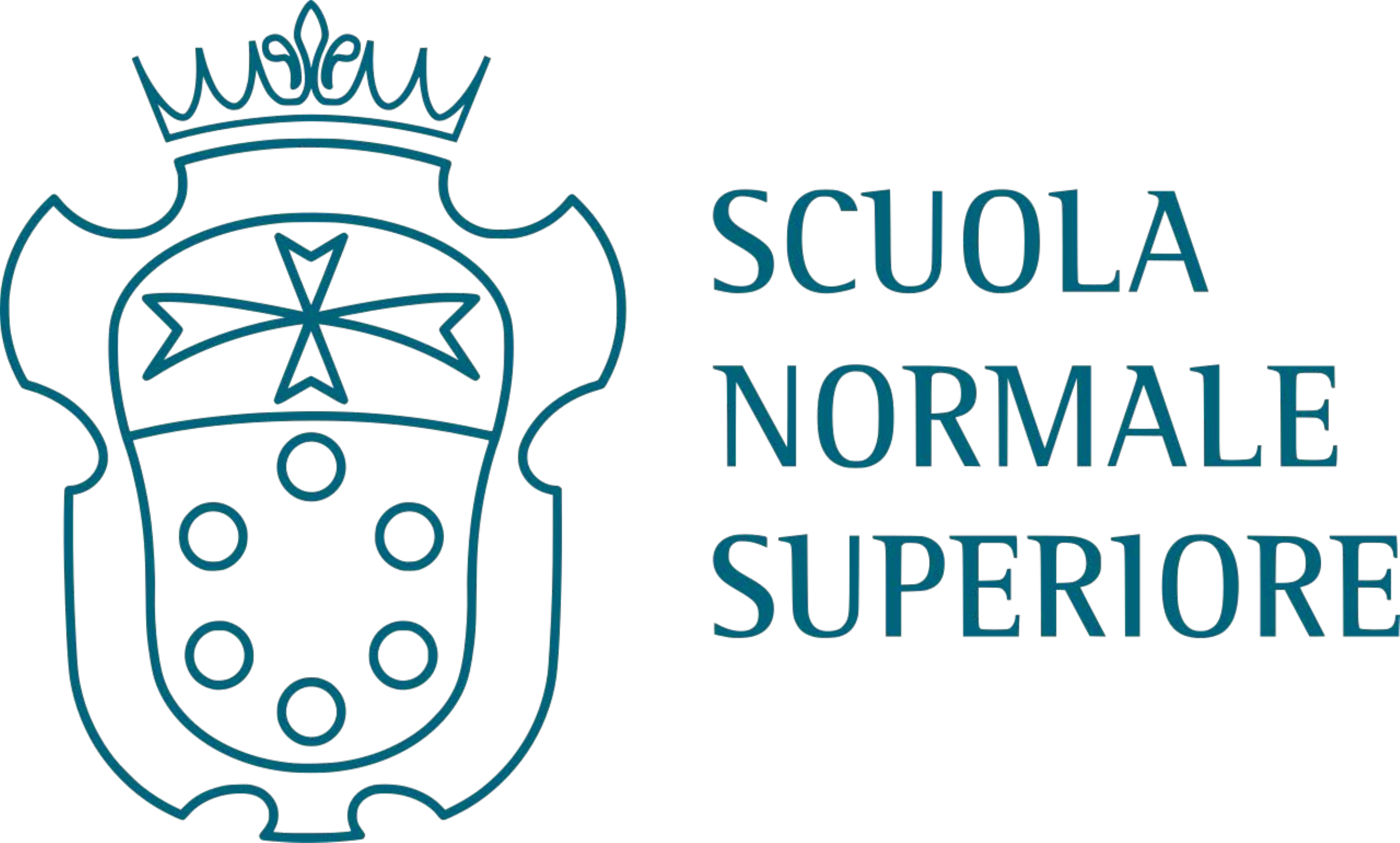}
\end{minipage}

\vspace{1cm}

\begin{center}
    {\large\scshape Research Project (PRe) --- Internship Report}\\[0.2cm]
    {\large ENSTA --- Institut Polytechnique de Paris}\\[0.2cm]

    \vspace{1cm}

    \rule{\textwidth}{1.2pt}\\[0.7cm]
    {\fontsize{22}{27}\selectfont\bfseries Regimes in the Order Flow}\\[0.7cm]
    {\fontsize{15}{19}\selectfont
     Duration-Aware and Multivariate Bayesian Online\\[2pt]
     Changepoint Detection for High-Frequency Markets}\\[0.4cm]
    \rule{\textwidth}{1.2pt}

    \vspace{1.1cm}

    {\large\textbf{Field of study:} Applied Mathematics --- Quantitative Finance}\\[0.15cm]

    \vspace{1.6cm}

    {\large\textbf{Ramzi Jebali}}\\[0.2cm]
    {\normalsize Master of Engineering, ENSTA, Class of 2027}

    \vspace{0.9cm}

    {\large\textbf{Host institution}}\\[0.25cm]
    Scuola Normale Superiore --- \textit{Quantitative Finance} research group\\
    Piazza dei Cavalieri, 7 --- 56126 Pisa, Italy

    \vspace{0.8cm}

    \begin{tabular}{@{}ll@{}}
      \textbf{Head of the research group} & Prof. Fabrizio \textsc{Lillo} \\[3pt]
      \textbf{Supervisor} & Dr Ioanna-Yvonni \textsc{Tsaknaki} \\[3pt]
      \textbf{Academic tutor} & Prof. Francesco \textsc{Russo} \\
    \end{tabular}

    \vfill

    {\large Internship carried out from May 18th to July 31st, 2026}
\end{center}
\end{titlepage}

\clearpage
\thispagestyle{empty}

\vspace*{-1.5cm}
\begin{center}
    {\fontsize{19}{21}\selectfont\bfseries Acknowledgements}
\end{center}
\vspace*{-0.1cm}

I would like to express my gratitude to all those who made this Research Project possible and who accompanied it throughout.

First, I am sincerely grateful to \textbf{Professor Francesco Russo}, academic tutor of this internship, without whom this experience would simply not have taken place. His availability, the clarity of his advice and his unfailing kindness deserve far more than the few lines I can devote to him here.

My thanks then go to \textbf{Professor Fabrizio Lillo}, professor at the Scuola Normale Superiore and head of the \textit{Quantitative Finance} research group, for welcoming me into his team, for the confidence he placed in me by entrusting me with an open-ended subject, and for the remarks that, on several occasions, redirected my work towards the right questions.

Finally, and most particularly, I thank \textbf{Dr Ioanna-Yvonni Tsaknaki}, researcher at the Scuola Normale Superiore and supervisor of this internship, for her guidance, her patience with my questions, the rigour she demanded at every stage, her incredible kindness. Her own work on online changepoint detection is the starting point of this report.

I also thank the PhD students and post-doctoral researchers of the group, whose seminars and informal discussions considerably broadened my view of the field, as well as the academic staff of ENSTA for organising the PRe.

The high-frequency data used in this report were obtained from LOBSTER (\url{https://php.lobsterdata.com/index.php}), a specialised limit-order-book data platform to which the Scuola Normale Superiore subscribes; access is restricted to subscribing institutions and the data are not publicly available. I thank the Scuola Normale Superiore for making them available for this project.

\begin{center}
    {\fontsize{19}{21}\selectfont\bfseries Abstract}
\end{center}
\vspace*{-0.1cm}

Financial markets alternate between periods of relative stability and instability, with structural breaks marking the transitions between these regimes. Identifying such breaks in real time is a central requirement for any trading or risk system operating at high frequency. This report studies Bayesian Online Changepoint Detection (BOCPD) and two extensions proposed in the literature, and applies them to the signed order flow of NASDAQ-listed equities. We first re-derive and implement the baseline algorithm of \cite{bocpd_2007}, whose conjugate exponential-family structure yields a closed-form predictive likelihood, and we assess it along three complementary criteria: a segmentation criterion (the covering metric), a point-forecast criterion (the mean squared error) and a fully probabilistic criterion (the cumulative predictive log-likelihood). We then replace the constant hazard rate of BOCPD --- and the geometric regime durations it implies --- by a run-length-dependent hazard derived from an explicit duration law within a hidden semi-Markov model, following the formulation of \cite{agudelo2020bayesian}. Calibrating Pareto and log-normal duration laws on NASDAQ order flow, we find that the duration-aware filter outperforms the constant-hazard baseline, the log-normal specification dominating both the geometric and the Pareto alternatives consistently across assets, months and calibration criteria; this is in agreement with the absence of a characteristic timescale in order flow. We finally implement the multivariate BOCPDMS framework of \cite{knoblauch2018bocpdms}, which combines a model universe with Bayesian vector autoregressions and online hyperparameter learning, and verify that its conjugacy extends to any known symmetric positive-definite noise matrix. Here the outcome is negative: on bivariate order flow the multivariate filter is outperformed by two independent univariate filters, whatever the noise specification. We trace this to the interaction between short regimes and heavy-tailed innovations, which turns the adaptivity of the model into a liability rather than an asset.

\vspace{0.3cm}

\clearpage
\pagestyle{fancy}
\markboth{Contents}{Contents}
\tableofcontents

\clearpage

\newpage
\section{Introduction}\label{sec:intro}

The modelling of order flow is central to the understanding of price formation in
electronic markets and to the design of execution algorithms. One of its most
robust empirical regularities is the persistence of the signed trade order flow:
buy trades tend to be followed by buy trades and sell trades by sell trades, over
horizons as long as several weeks. Documented independently by
\cite{lillo2004long} and \cite{bouchaud2004fluctuations}, this persistence takes
the form of an autocorrelation function decaying as a non-summable power law,
which is the formal signature of a long-memory process. It has been attributed
primarily to order splitting, the practice by which large investors execute a
single \emph{metaorder} through a long sequence of smaller child orders;
\cite{lillo2005theory} establish quantitatively the relationship between the
autocorrelation of order flow and the distribution of metaorder sizes.

From an econometric standpoint, a long-memory series can be generated,
at least approximately, by a regime-shift process in which each regime is
short-memory and the regime lengths are heterogeneous. This observation motivates
the description of order flow as a sequence of regimes, and suggests a natural
correspondence between a regime and the execution of a metaorder. It is the
approach followed by \cite{tsaknaki2025online}, who use Bayesian online
change-point detection to identify order-flow regimes on NASDAQ data and show that
regime information improves the prediction of both order flow and market impact.
The present report starts from that work.

Among the algorithms available for this task, we concentrate on those performing
detection \emph{online}, that is, using only the observations available up to the
current instant, and within a Bayesian framework, which is well suited to
quantifying the uncertainty attached to a change-point through a posterior
distribution rather than committing to a hard decision at each step. The reference
algorithm is that of \cite{bocpd_2007}, which recursively computes the posterior
distribution of the time elapsed since the last change-point --- the \emph{run
length} --- and remains tractable because conjugacy makes its predictive
likelihood available in closed form. Its constant hazard rate, however, carries an
implication that sits uneasily with the stylised fact recalled above: a constant
hazard forces regime durations to be geometric, a law possessing a characteristic
timescale, whereas order flow has none. The heterogeneity of regime lengths, which
is precisely what makes the regime-shift description of long memory work, is
therefore excluded by construction.

This report examines that tension. We re-derive and implement the baseline
algorithm in full, and equip it with
an evaluation framework measuring segmentation quality, point-forecast accuracy
and probabilistic calibration. We then replace the constant hazard rate by a
run-length-dependent one, derived from an explicit duration law within a hidden
semi-Markov model following the formulation of \cite{agudelo2020bayesian}, and
calibrate Pareto and log-normal duration laws on NASDAQ order flow. We finally
move from the univariate to the multivariate setting with the BOCPDMS framework of
\cite{knoblauch2018bocpdms}, which combines a model universe with Bayesian vector
autoregressions and online hyperparameter learning. The algorithms studied here
are due to their respective authors; the implementations used throughout, and the
experiments run with them, are ours.

The report is organised as follows. Section~\ref{sec:litreview} recalls the
Bayesian ingredients of BOCPD and derives the algorithm;
Section~\ref{sec:implementation} implements it on synthetic data and
Section~\ref{sec:metrics} introduces the evaluation metrics.
Section~\ref{sec:hsmm} moves from a constant to a run-length-dependent hazard
rate, Section~\ref{sec:microstructure} presents the microstructural facts that
motivate this choice, and Section~\ref{sec:realdata} reports the calibration and
out-of-sample results on NASDAQ data. Sections~\ref{sec:bocpdms}
to~\ref{sec:hyperlearning} present the BOCPDMS framework, its Bayesian vector
autoregressions and the online learning of its hyperparameters;
Sections~\ref{sec:synthexp} and~\ref{sec:realmulti} report the corresponding
experiments on synthetic and on real data. Section~\ref{sec:conclusion} concludes.
Proofs and additional material are collected in the appendices.

\newpage
\section{Background: Bayesian online changepoint detection}\label{sec:litreview}

Historically, Frequentist statistics provided fast, online algorithms to try to detect these shifts, but they struggled to quantify uncertainty rigorously. Conversely, Bayesian models handled uncertainty exceptionally well, but their heavy computational cost restricted them to offline analysis. The foundational paper by \cite{bocpd_2007} is ground-breaking because it brought the power of Bayesian probabilities into the real-time (online) world: it predicts the next data point using only the data observed up to the present moment, while providing a full probability distribution over the uncertainty of a changepoint.

\subsection{Foundations of Bayesian Inference and the Exponential Family}

In classical statistics (the Frequentist approach), parameters are estimated using Maximum Likelihood (ML), which yields a point estimate (a single numerical value). However, in the Bayesian framework, parameters are treated as random variables themselves. We maintain a probability distribution over the parameters, which reflects uncertainty. As more data points arrive, this distribution narrows, eventually peaking around the ML estimate.

To perform this Bayesian updating efficiently in real-time as new data arrives, we cannot afford to keep the entire historical dataset in the computer's active memory. We need a mathematical shortcut. This is why we strictly restrict our models to a specific class of probability distributions: the \textbf{Exponential Family}.

\vspace{-10pt}
\begin{definition}[Exponential family]\label{def:expfam}
Let $\pmb{x}$ be a $D$-dimensional vector (or a scalar if $D=1$). The exponential family of distributions over $\pmb{x}$, given natural parameters $\pmb{\eta}$, is defined as:
\begin{equation}
    p(\pmb{x}|\pmb{\eta}) = h(\pmb{x})g(\pmb{\eta})\exp\{\pmb{\eta}^T\pmb{u}(\pmb{x})\}
\end{equation}
where $h(\pmb{x})$ is the base measure, $\pmb{u}(\pmb{x})$ is the vector of sufficient statistics, and $g(\pmb{\eta})$ acts as the normalizing constant ensuring the probability integrates to $1$.
\end{definition}
\vspace{-20pt}
\begin{proposition}[Likelihood of i.i.d. data]\label{prop:iid_likelihood}
For a dataset $X = \{\pmb{x}_1, \dots, \pmb{x}_N\}$ where each observation is independent and identically distributed (i.i.d.), the joint likelihood under the exponential family is given by the product of individual probabilities:
\begin{equation*}
    p(X|\pmb{\eta}) = \prod_{n=1}^N p(\pmb{x}_n|\pmb{\eta}) = \Big(\prod_{n=1}^N h(\pmb{x}_n)\Big)g(\pmb{\eta})^N\exp\Big\{\pmb{\eta}^T\sum_{n=1}^N\pmb{u}(\pmb{x}_n)\Big\}
\end{equation*}
\end{proposition}
\vspace{-5pt}
\noindent The reason we use the exponential family: it guarantees the existence of a \textbf{conjugate prior}.
\vspace{-5pt}
\begin{definition}[Conjugacy]\label{def:conjugacy}
In Bayesian probability theory, a prior distribution is said to be conjugate to a likelihood function if the resulting posterior distribution belongs to the same parametric family as the prior.
\end{definition}
\vspace{-20pt}
\begin{proposition}[Conjugate prior for the exponential family]\label{prop:conj_prior}
For any member of the exponential family, there exists a conjugate prior distribution over the parameters $\pmb{\eta}$ of the form:
\begin{equation*}
    p(\pmb{\eta} \mid \pmb{\chi}, \nu) = f(\pmb{\chi}, \nu) g(\pmb{\eta})^{\nu} \exp\left\{ \nu \pmb{\eta}^T \pmb{\chi} \right\}
\end{equation*}
where $\pmb{\chi}$ and $\nu$ act as the hyperparameters defining our prior belief before observing the data.
\end{proposition}
\vspace{-5pt}
\noindent To demonstrate the analytical power of this conjugacy, we must first recall Bayes' theorem, which is the fundamental engine of our sequential learning.
\vspace{-5pt}
\begin{proposition}[Bayes' Theorem]\label{prop:bayes}
Given the observed data $X$, Bayes' theorem allows us to update our prior belief about the parameters $\pmb{\eta}$ to obtain the posterior distribution:
$$ p(\pmb{\eta} \mid X) = \frac{p(X \mid \pmb{\eta}) p(\pmb{\eta})}{p(X)} \propto p(X \mid \pmb{\eta}) p(\pmb{\eta}) \ \ \Longleftrightarrow \ \ \textbf{Posterior} \propto \textbf{Likelihood} \times \textbf{Prior}$$
The marginal likelihood (or evidence) $p(X)$ acts as a normalizing constant.
\end{proposition}

\noindent By applying this theorem to our exponential family likelihood and its conjugate prior, we can explicitly compute the updated posterior distribution:
\vspace{-5pt}
\begin{proposition}[Posterior Update Rule]\label{prop:posterior_update}
Given the likelihood of i.i.d. data $X$ and the conjugate prior, the posterior distribution over the natural parameters $\pmb{\eta}$ retains the exact same exponential form as the prior:
\vspace{-15pt}
\begin{align*}
    p(\pmb{\eta} \mid X, \pmb{\chi}, \nu) & \propto p(X \mid \pmb{\eta}) p(\pmb{\eta} \mid \pmb{\chi}, \nu)
    \propto g(\pmb{\eta})^{\nu+N} \exp\left\{ \pmb{\eta}^T \left( \pmb{\chi} + \sum_{n=1}^N \pmb{u}(\pmb{x}_n) \right) \right\}
\end{align*}
\vspace{-10pt}
Consequently, the new hyperparameters of the posterior are updated via simple additive rules:
\vspace{-3pt}
$$\nu_{\text{posterior}} = \nu_{\text{prior}} + N, \qquad
    \pmb{\chi}_{\text{posterior}} = \pmb{\chi}_{\text{prior}} + \sum_{n=1}^N \pmb{u}(\pmb{x}_n)$$
\end{proposition}
\vspace{-12pt}
\begin{definition}[Sufficient Statistics]\label{def:suffstat}
The quantities $N$ (number of observations) and $\sum_{n=1}^N \pmb{u}(\pmb{x}_n)$ are called the \textbf{sufficient statistics}. They contain all the necessary information from the data $X$ required to perform inference on $\pmb{\eta}$.
\end{definition}

\noindent The concept of sufficient statistics is the cornerstone of Bayesian online learning. It proves that to continuously update our beliefs as new data arrives, we do not need to look back at the entire historical dataset. Maintaining and updating these running counters is sufficient, ensuring that the computational memory footprint remains minimal.

\subsection{Bayesian Inference in the Univariate Gaussian Case}

\textit{We now apply the framework to the specific case of the univariate Gaussian distribution.}

\subsubsection*{Case of unknown mean $\mu$ and known variance $\sigma^2$}

Let data $\mathcal{D} = \{x_1, \dots, x_N\}$ be drawn independently from $\mathcal{N}(\mu,\sigma^2)$. Our aim is to infer the unknown mean $\mu$ sequentially. We establish the formal update rules (proof in Appendix~\ref{app:gauss_update}):
\vspace{-5pt}
\begin{proposition}[Gaussian Conjugate Update]
If the prior distribution of the mean is Gaussian $p(\mu) = \mathcal{N}(\mu|\mu_0,\sigma^2_0)$, and the likelihood is Gaussian with known variance $\sigma^2$, the posterior distribution remains Gaussian $p(\mu|\mathcal{D}) = \mathcal{N}(\mu|\mu_N,\sigma_N^2)$. The sufficient statistics are given by:
\vspace{-6pt}
$$ \frac{1}{\sigma_N^2} = \frac{N}{\sigma^2}+\frac{1}{\sigma_0^2} \hspace{1cm}|\hspace{1cm} \mu_N = \frac{N\sigma_0^2}{N\sigma_0^2+\sigma^2}\mu_{ML}+\frac{\sigma^2}{N\sigma_0^2+\sigma^2}\mu_0 $$
\vspace{-6pt}
\label{prop:gauss_conj_update}
\end{proposition}
\vspace{-5pt}
\noindent As noted by \cite{bishop2006pattern}, the posterior mean $\mu_N$ is a convex combination between the prior mean $\mu_0$ and the maximum likelihood solution $\mu_{ML}$. If $N=0$ (no data), $\mu_N = \mu_0$. As $N\rightarrow\infty$, $\mu_N \rightarrow \mu_{ML}$ and $\sigma_N^2 \rightarrow 0$, meaning our uncertainty about the mean vanishes.

When predicting a new unseen data point $x_{N+1}$, we rely on the \textbf{predictive posterior distribution}, obtained by marginalizing over the unknown parameter, which is $\mu$ here:
$$ p(x_{N+1}|x_{1:N}) = \int p(x_{N+1}|\mu,x_{1:N})p(\mu|x_{1:N})d\mu = \int p(x_{N+1}|\mu)p(\mu|x_{1:N})d\mu $$
The last equality follows directly from conditional independence: $x_{N+1}\indep x_{1:N} \mid \mu$.
\vspace{-5pt}
\begin{proposition}[Gaussian Predictive Posterior]
Given the posterior distribution of the mean $\mathcal{N}(\mu_N, \sigma_N^2)$ and the known observation variance $\sigma^2$, the posterior predictive distribution is:
\vspace{-6pt}
$$ p(x_{N+1}|\mathcal{D}) = \int \mathcal{N}(x_{N+1}|\mu,\sigma^2)\mathcal{N}(\mu|\mu_N,\sigma_N^2)d\mu \ = \ \pmb{\mathcal{N}(\mu_N,\sigma_N^2+\sigma^2)} $$
\vspace{-6pt}
\label{prop:gaussian_pred_post}
\end{proposition}
\vspace{-5pt}
\noindent\textit{This demonstrates that the total variance of our prediction is the sum of our uncertainty about the mean ($\sigma_N^2$) and the inherent noise of the data ($\sigma^2$).} The proof is given in Appendix~\ref{app:pred_post}.

\subsection{Bayesian Online Changepoint Detection}

\textit{We now introduce the foundational algorithm for detecting structural breaks. The core idea is to use the sequential Gaussian framework we just derived.}

In statistics, we assume that a data sequence is "generated" by an underlying probability distribution, which is defined by \textbf{generative parameters}. A \textbf{changepoint} is defined as an abrupt, structural shift in these generative parameters.

The work of \cite{bocpd_2007} is groundbreaking because it brought exact Bayesian inference into a real-time (online) setting.

\subsubsection*{BOCPD Foundational assumptions}
To achieve exact, real-time inference, the model relies on two strong assumptions:
\vspace{-5pt}
\begin{itemize}[itemsep=-2pt]
    \item \textbf{(a) Intra-regime i.i.d.:} Within any given regime $\rho$, the data points $x_t$ are independent and identically distributed (i.i.d.) given the generative parameters $\eta_\rho$: $x_t \sim P(x_t | \eta_\rho)$
    \item \textbf{(b) Inter-regime independence:} The generative parameters are completely i.i.d. from one regime to the next. When a changepoint occurs, the algorithm forgets its "history" and draws brand new parameters $\eta_{\rho+1}$ blindly from the base prior distribution $P(\eta)$.
\end{itemize}

\noindent We use the BOCPD algorithm to detect changes in the mean of the data, assuming the observations are normally distributed with a fixed variance $\sigma^2$ and a regime-switching mean $\mu$ governed by a Gaussian prior with mean $\mu_0$ and standard deviation $\sigma_0$.

\vspace{-5pt}
\subsubsection*{Mathematical breakdown of the algorithm}

\begin{definition}[Run Length]\label{def:runlength}
The duration of the current regime is tracked by a discrete state variable called the \textbf{run length}, denoted $r_t$. At time $t$, it evolves via a Bernoulli process:
\vspace{-5pt}
$$ p(r_t | r_{t-1}) = \begin{cases}
1/h & \text{if } r_t = 0 \text{ (a changepoint occurs)} \\
1 - 1/h & \text{if } r_t = r_{t-1} + 1 \text{ (the regime survives)}
\end{cases} $$
\vspace{-5pt}
$1/h$ is the \textbf{hazard rate}, the prior probability of a changepoint occurring at the next time step.
\end{definition}

The algorithm's objective is to compute the \textbf{posterior distribution of the run length}: $ P(r_t | x_{1:t}) = P(r_t, x_{1:t}) / P(x_{1:t})$. To compute it, we first compute recursively the \textbf{joint distribution} $P(r_t, x_{1:t})$ and then, we compute the \textbf{evidence} $P(x_{1:t}) = \sum_{r_{t}} P(r_t, x_{1:t})$.

\begin{proposition}[Recursive Message Passing]\label{prop:message_passing}
By applying Bayes' chain rule and the conditional independence of the run length, the joint distribution decomposes recursively into three distinct computational blocks:
\begin{align*}
    P(r_t, x_{1:t}) &= \sum_{r_{t-1}} P(r_t, x_t | r_{t-1}, x_{1:t-1}) P(r_{t-1}, x_{1:t-1}) \\
    &= \sum_{r_{t-1}} \underbrace{P(r_t | r_{t-1})}_{\text{Hazard Rate}} \underbrace{P(x_t | r_{t-1}, x_{t-1}^{(r)})}_{\text{UPM}} \underbrace{P(r_{t-1}, x_{1:t-1})}_{\text{Previous Message}}
\end{align*}
\end{proposition}

Usually, in Bayesian statistics, evaluating this predictive probability requires calculating a complex integral over all possible values of the unknown parameter. However, by exploiting the conjugacy property, we can get a \textbf{closed-form UPM} (proofs of the last two propositions in Appendix~\ref{app:proofs}).
\vspace{-5pt}
\begin{proposition}[Closed-form UPM for the Gaussian Case]\label{prop:upm_gaussian}
If the data within a regime is modeled as a Gaussian with an unknown mean and a known variance $\sigma^2$, and we use a conjugate Gaussian prior, the UPM evaluates to a closed-form Gaussian distribution:
$$p(x_t \mid r_{t-1}, x_{t-1}^{(r)}) = \pmb{\mathcal{N}(x_t \mid \mu_{r_{t-1}}, \sigma^2 + \sigma^2_{r_{t-1}})}$$
where $\mu_{r_{t-1}}$ and $\sigma^2_{r_{t-1}}$ are the updated sufficient statistics using strictly the $r_{t-1}$ observations of the current regime, using the Gaussian conjugate update (Proposition \ref{prop:gauss_conj_update}).
\end{proposition}
\vspace{-5pt}
This closed-form solution allows the algorithm to evaluate the likelihood of a new data point instantaneously, thus making true real-time detection feasible.

The algorithm yields the run length posterior distribution from which we can get the \textbf{posterior predictive distribution}, on which we could rely for prediction:
\vspace{-5pt}
$$P(x_{t+1} | x_{1:t}) = \sum_{r_t} P(x_{t+1} | r_t, x_t^{(r)}) P(r_t | x_{1:t})$$

\vspace{-5pt} This is Bayesian Model Averaging. Instead of making a hard decision about changepoint location, the algorithm maintains a full probability distribution over all possible run lengths. The final prediction $P(x_{t+1} | x_{1:t})$ is a weighted average of every run length's prediction, weighted by the posterior belief $P(r_t | x_{1:t})$ in that scenario. This way the model deals with uncertainty.

\section{Implementation of the BOCPD algorithm}\label{sec:implementation}

\textit{Now we translate the theoretical BOCPD framework into practice using it on synthetic data.}

\subsection{Generation of simulated data}

To rigorously test the BOCPD algorithm, we need to generate synthetic data whose ground truth is perfectly known. The design of this generative process must mirror the probabilistic assumptions that the algorithm itself postulates.

\subsubsection*{The importance of a coherent generative process}

As established, the algorithm models the run length $r_t$ as a Bernoulli process: at every time step $t$, there is a probability $1/h$ that a changepoint occurs. The duration of any regime is therefore a geometric random variable with parameter $1/h$.

\begin{definition}[Bernoulli Generative Process]\label{def:bernoulli_gen}
Given a time horizon $T$, a hazard rate $1/h_{\text{gen}}$, a prior distribution $\mathcal{N}(\mu_0, \sigma_0^2)$ for the regime means, and a known observation variance $\sigma^2$, the data $x_{1:T}$ is generated as follows. At each time step $t$, a Bernoulli trial with success probability $1/h_{\text{gen}}$ is performed:
\begin{itemize}[itemsep=-2pt]
    \item If the trial succeeds (with probability $1/h_{\text{gen}}$), a changepoint occurs at $t$. A new regime mean $\mu_{\text{new}}$ is drawn independently from $\mathcal{N}(\mu_0, \sigma_0^2)$.
    \item If the trial fails (with probability $1 - 1/h_{\text{gen}}$), the current regime continues with its existing mean.
\end{itemize}
Within any regime with mean $\mu_\rho$, each observation is drawn independently: $ x_t \sim \mathcal{N}(\mu_\rho, \sigma^2) $
\end{definition}

\subsection{Model Initialization}

Now that we have simulated data we can work with, we initialize the BOCPD algorithm with three core components, to try to infer the locations of the changepoints (marked by red lines in the figures):
\vspace{-5pt}
\begin{enumerate}[itemsep=-2pt]
    \item \textbf{The known observation variance ($\sigma^2$)}
    \item \textbf{The prior distribution parameters ($\mu_0, \sigma_0^2$)}
    \item \textbf{The hazard rate of the algorithm ($1/h_{\text{algo}}$).}
\end{enumerate}

\subsection{The Run Length Posterior}

At each time step $t$, the primary output of the BOCPD algorithm is the conditional distribution of the run length given the past observations: $P(r_t \mid x_{1:t})$.
\vspace{-5pt}
\begin{definition}[Run Length Posterior Matrix]
The Run Length Posterior Matrix is a two-dimensional grid of size $T \times T$. The horizontal axis represents the absolute time $t$, and the vertical axis represents the potential run length $r_t$. Because the run length can never exceed the current time ($r_t \le t$), this matrix naturally forms a lower-triangular structure, where all upper-triangular entries are strictly equal to zero.
\end{definition}
\vspace{-5pt}
To interpret this matrix visually, we plot it as a heatmap where the intensity of the color represents the probability mass. Reading it is quite intuitive:
\vspace{-5pt}
\begin{itemize}[itemsep=-4pt]
    \item If a regime is stable, its age increases by exactly one step at each time tick ($r_t = r_{t-1} + 1$). This survival translates into a clear, upward diagonal line of high probability.
    \item The moment the algorithm receives a data point that is highly improbable under the current regime, the probability shifts heavily toward the birth of a new regime ($r_t = 0$). The diagonal line abruptly drops ($r_t = 0$).
\end{itemize}
In an ideal world, the drops of these diagonal lines should coincide with the red lines that mark the changepoints of our generated dataset.

\subsection{The impact of the hazard rate}

We deliberately deceive the algorithm by setting a very high prior probability for changepoints ($1/h = 1/10$) and then run it with the true average frequency of our generative model ($h_{\text{algo}} = h_{\text{data}}$).

\begin{figure}[H]
    \centering
    \begin{subfigure}[b]{0.49\linewidth}
        \includegraphics[width=\linewidth]{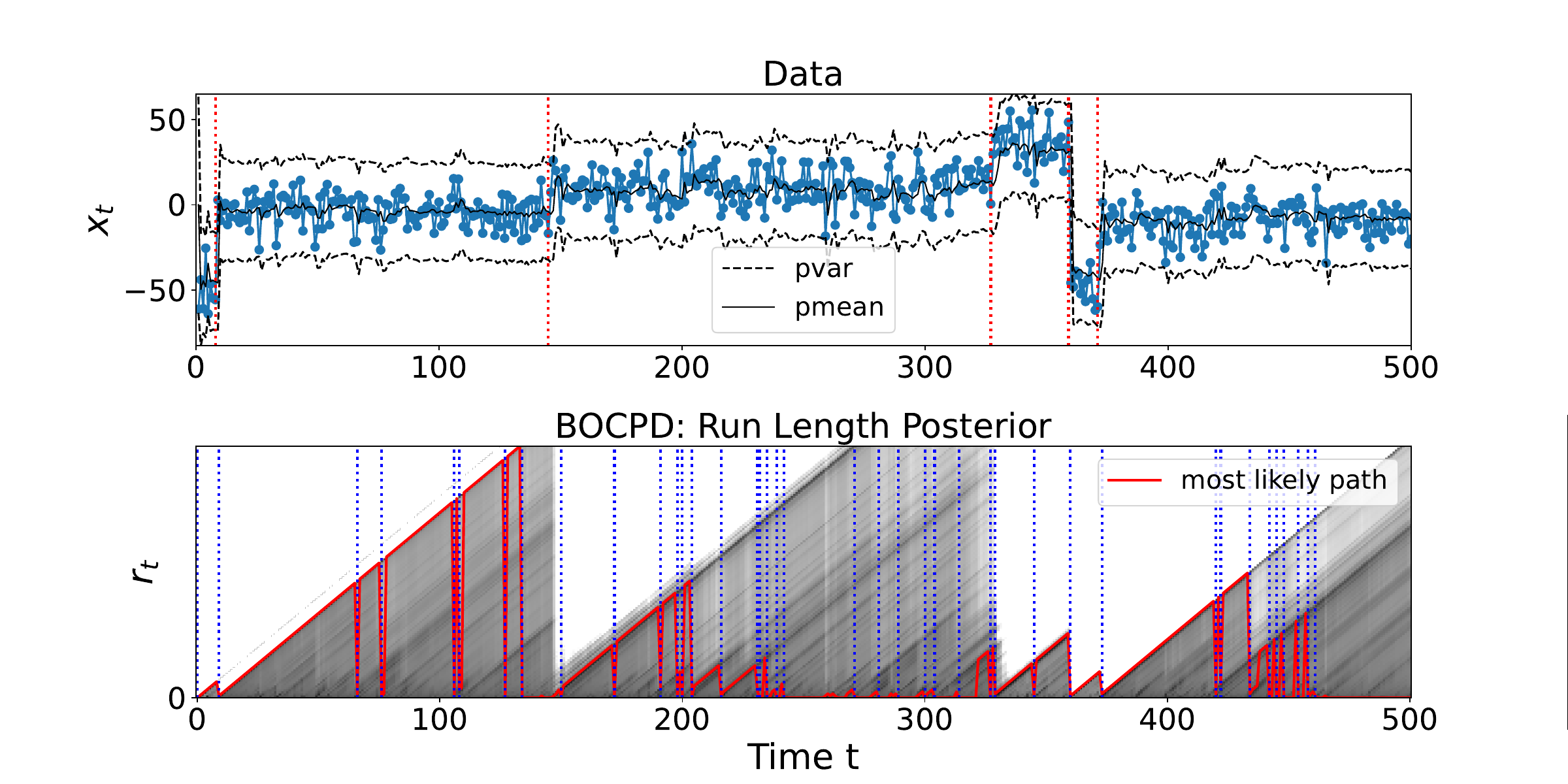}
        \caption{Hyper-reactive hazard rate ($h_{\text{algo}} = 10$).}
    \end{subfigure}
    \hfill
    \begin{subfigure}[b]{0.49\linewidth}
        \includegraphics[width=\linewidth]{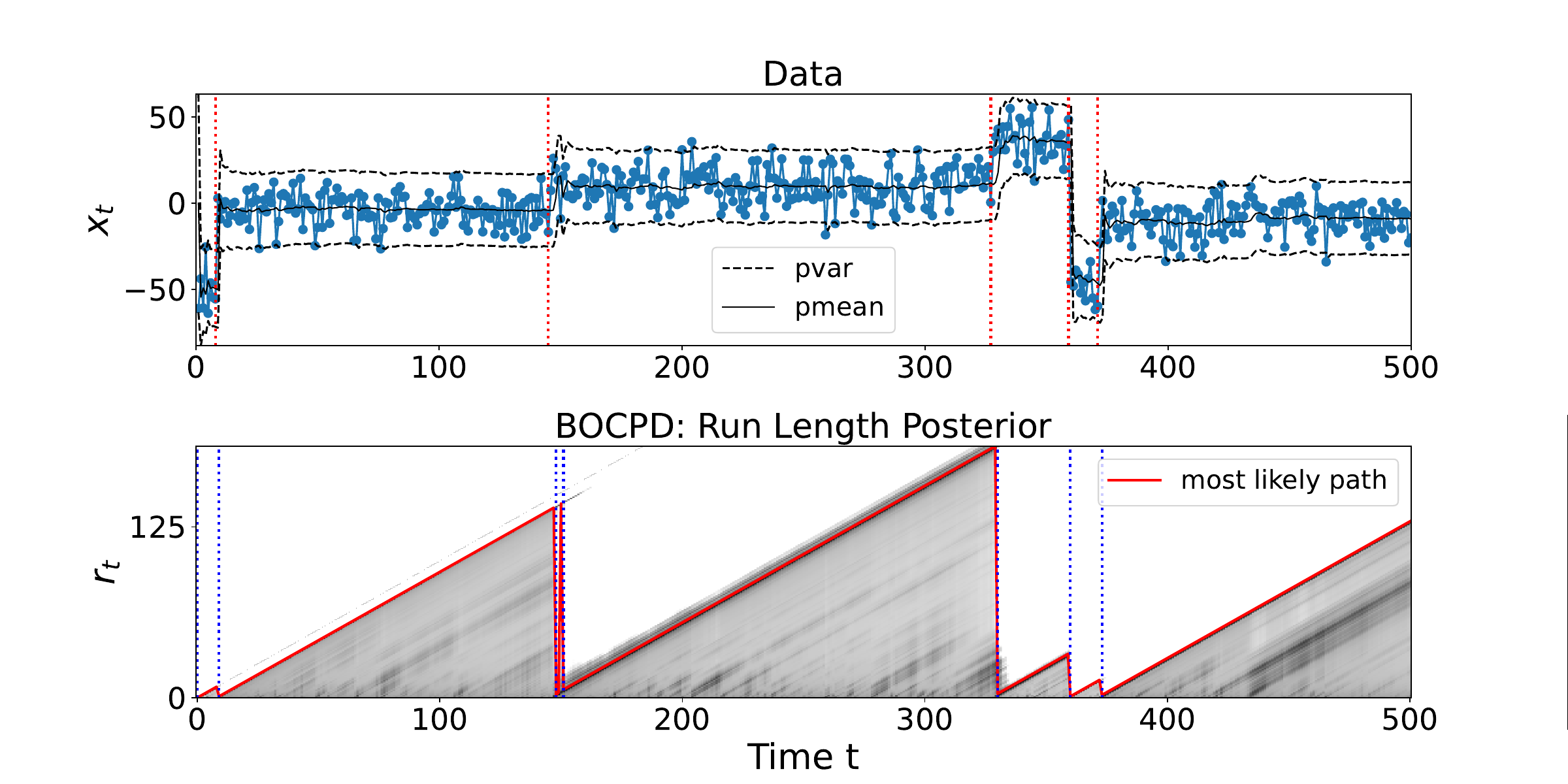}
        \caption{Optimal hazard rate ($h_{\text{algo}} = 100$).}
    \end{subfigure}
    \caption{Run Length Posterior matrix under a mis-specified and a well-specified hazard rate.}
    \label{fig:hazard_impact}
\end{figure}
\vspace{-3pt}
This highlights one practical difficulty of the BOCPD model: the choice of $h$. If it is too large, the algorithm may become hyper-sensitive to local noise, resulting in numerous false positive detections. With an adequate value, it strikes a good balance. However, it may still miss the subtlest jump as changes buried within the local noise variance are difficult to infer.

In our simulation, choosing an optimal $h$ is easy because we own the generative model. However, when we switch to real data, the true frequency of regimes is unknown.

\section{Performance Evaluation Metrics}\label{sec:metrics}

Now let us establish a rigorous mathematical framework to evaluate the detection performance.

The core problem with point-wise metrics such as the $\mathcal{F}_1$-Score is that they reduce the segmentation problem to a simple ``hit or miss'' question at individual points, discarding all information about how well the algorithm guesses the shape and duration of the regimes.

\subsection{Covering Metric}

Instead of measuring if the algorithm detected a given changepoint, it measures how well the algorithm's segmentation of the time axis agrees with the true segmentation.

A set of changepoints $\{\tau_1, \dots, \tau_K\}$ on a time series of length $T$ induces a \textbf{partition} of the time axis $[0, T]$ into $K+1$ contiguous intervals (the regimes). We denote the ground truth partition $\mathcal{S} = \{A_1, \dots, A_{K+1}\}$ and the algorithm's partition $\mathcal{S}' = \{B_1, \dots, B_{K'+1}\}$. We assess how well each true regime $A_i$ is ``covered'' by the best-matching detected regime $B_j$.

\begin{definition}[Jaccard Index]
For two intervals $A = [a_1, a_2)$ and $B = [b_1, b_2)$ on the real line, the Jaccard Index measures their degree of overlap relative to their combined extent:
\begin{equation}
    J(A, B) = \frac{|A \cap B|}{|A \cup B|} = \frac{\max(0, \; \min(a_2, b_2) - \max(a_1, b_1))}{|A| + |B| - |A \cap B|} \in [0, 1]
\end{equation}
The Jaccard Index equals $1$ if and only if $A = B$ (perfect overlap), and equals $0$ if the two intervals are completely disjoint.
\end{definition}

For each true regime $A_i \in \mathcal{S}$, we find the detected regime $B_j \in \mathcal{S}'$ that maximizes the Jaccard overlap. The Covering Metric then aggregates these best-match scores across all true regimes, weighting each by its relative length $|A_i|/T$ so that the reconstruction quality of long, important regimes contributes more than that of short ones.

\begin{definition}[Covering Metric]\label{def:covering_metric}
Let $\mathcal{S} = \{A_1, \dots, A_{K+1}\}$ be the ground truth partition and $\mathcal{S}' = \{B_1, \dots, B_{K'+1}\}$ be the partition induced by the algorithm's detected changepoints. The Covering Metric is defined as:
\vspace{-13pt}
\begin{equation}
    C(\mathcal{S}, \mathcal{S}') = \frac{1}{T}\sum_{i=1}^{K+1} |A_i| \cdot \max_{B_j \in \mathcal{S}'} J(A_i, B_j) \;\in [0, 1]
\end{equation}

\vspace{-5pt} \noindent $C = 1$ indicates a perfect reconstruction of the ground truth partition, $C = 0$ no overlap.
\end{definition}

\subsection{Mean Squared Error}

It is not all about \textbf{clustering} (i.e. segmenting correctly). We may also do predictions, so we introduce a metric that evaluates the \textbf{forecasting} capability.

The recursive inference of the BOCPD algorithm produces a one-step-ahead forecast at each time step. Recall that we can get from the algorithm the predictive posterior distribution:
$$P(x_{t+1} | x_{1:t}) = \sum_{r_t} P(x_{t+1} | r_t, x_t^{(r)}) \; P(r_t | x_{1:t})$$

\vspace{-5pt} Each UPM term $P(x_{t+1} | r_t, x_t^{(r)})$ is a Gaussian $\mathcal{N}(\mu_{r_t}, \sigma^2 + \sigma^2_{r_t})$, whose mean $\mu_{r_t}$ is the posterior mean of the regime that has been running for $r_t$ steps. The overall predictive mean is therefore the posterior-weighted average of these regime-specific means.

\begin{definition}[Predictive Mean]\label{def:predictive_mean}
The \textbf{predictive mean} (one-step-ahead forecast) at time $t$, denoted $\hat{\mu}_t$, is the expected value of the next observation under the full posterior predictive distribution:
\vspace{-12pt}
\begin{equation}
    \hat{\mu}_t = \sum_{r_t=0}^{t} \mu_{r_t} \; P(r_t \mid x_{1:t})
\end{equation}
where $\mu_{r_t}$ is the posterior mean of the Gaussian UPM for run length $r_t$, and $P(r_t \mid x_{1:t})$ is the run length posterior.
\end{definition}
\vspace{-20pt}
\begin{definition}[Mean Squared Error]
The Mean Squared Error (MSE) of the one-step-ahead forecasts over the time horizon $\{0, \dots, T\}$ is defined as:
\vspace{-8pt}
\begin{equation}
    \text{MSE} = \frac{1}{T} \sum_{t=1}^{T} \big(\hat{\mu}_{t-1} - x_{t}\big)^2
\end{equation}
\end{definition}
\vspace{-5pt}
\textit{We now have two complementary tools: the Covering Metric judges the quality of the segmentation, while the MSE judges the quality of the prediction. Together, they provide a satisfying picture of an algorithm's performance.}

\section{From Constant to Time-Varying Hazard Rate}\label{sec:hsmm}

\subsection{The fundamental limitation of a constant hazard rate}

The regimes in the data span different durations. Yet the baseline BOCPD model assumes a single, constant hazard rate $1/h$ that governs the prior probability of a changepoint at every time step, leading to an irreconcilable dilemma:
\vspace{-5pt}
\begin{itemize}[itemsep=-2pt]
    \item If $h$ is calibrated to detect short, rapid regime changes (small $h$), the algorithm becomes hyper-reactive. It will successfully catch micro-regimes, but it will also fragment the long stable periods with dozens of false positives.
    \item If $h$ is calibrated to respect the long stable periods (large $h$), the algorithm becomes too conservative. It will correctly preserve the long regimes, but it may completely miss the short-lived regimes and the subtle jumps.
\end{itemize}
\vspace{-5pt}
The optimum is therefore not a satisfying solution: it is merely the ``least bad'' compromise in a trade-off that has no ideal resolution within the constant-hazard framework.

To handle data where regime lengths are inherently heterogeneous, as they are in real financial markets, we need to abandon the assumption of a constant hazard rate and move towards a \textbf{time-varying hazard rate} $h_t$.

\subsection{The structural flaw: the geometric constraint}

The standard BOCPD algorithm models the run length $r_t$ as evolving according to a Bernoulli process with a constant hazard rate $H = 1/h$. At every time step $t$, regardless of how long the current regime has already survived, the algorithm assumes a fixed probability $1/h$ that a changepoint will occur at the next step. This assumption imposes a structural constraint on the distribution of regime durations.

\begin{proposition}[Geometric regime durations under constant hazard]\label{prop:geometric}
If the hazard rate is constant $H(r) = p$ for all $r \geq 0$, where $p = 1/h \in (0,1)$, then the duration $d$ of any regime follows a Geometric distribution with parameter $p$:
\begin{equation}
    P(d = k) = (1-p)^{k-1} \, p, \quad k = 1, 2, 3, \dots
\end{equation}
with expected duration $\E[d] = 1/p = h$ and variance $\V[d] = (1-p)/p^2$.
\end{proposition}

\subsection{From HMM to HSMM}

Moving from a constant hazard rate toward a \textbf{time-varying hazard rate} requires switching from the implicit Hidden Markov Model (HMM) underlying standard BOCPD to a fundamentally different generative framework, the Hidden semi-Markov Model (HSMM).

\subsubsection*{Hidden Markov Models (HMM)}

\begin{definition}[Hidden Markov Model]
A Hidden Markov Model (HMM) is a doubly stochastic process where an underlying sequence of hidden states $\{z_t\}_{t \ge 1}$ governs the distribution of a sequence of continuous observations $\{x_t\}_{t \ge 1}$. The system is fully characterized by the hidden state sequence $\{z_t\}$, which satisfies the first-order Markov property $P(z_{t+1} \mid z_t, \dots, z_1) = P(z_{t+1} \mid z_t)$, and by the emission distributions $P(x_t \mid z_t)$.
\end{definition}
\vspace{-5pt}
The generative process of an HMM: at each step, the hidden state transitions based solely on the current state, and an observation is emitted from the corresponding distribution.

A critical consequence of this memoryless state transition rule is that the time the system spends residing in any given state is implicitly constrained. If the system is in state $i$ with a self-transition probability $a_{ii} = P(z_{t+1} = i \mid z_t = i)$, the probability that it remains in state $i$ for exactly $k$ consecutive steps before transitioning out is $P(d = k \mid z = i) = a_{ii}^{k-1}(1 - a_{ii})$, which is a Geometric distribution with parameter $1 - a_{ii}$.

\subsubsection*{Hidden semi-Markov Models (HSMM)}

To escape this constraint, we must abandon the Markov property on state transitions while keeping the conditional independence of the observations.

\begin{definition}[Hidden semi-Markov Model]\label{def:hsmm}
A Hidden semi-Markov Model (HSMM) is a generalization of the HMM in which the duration of each state is governed by an explicit distribution $P(d \mid z)$, rather than being implicitly determined by self-transition probabilities. The generative process is defined as follows:
\vspace{-5pt}
\begin{enumerate}[itemsep=-2pt]
    \item \textbf{State entry:} When a new regime begins, a hidden state $z$ is selected.
    \item \textbf{Duration draw:} A total duration $d$ is drawn directly from an explicit, state-specific duration distribution: $d \sim P(d \mid z)$.
    \item \textbf{Emission:} For exactly $d$ consecutive time steps, observations are emitted from the active state distribution $P(x \mid z)$. During this period, the state $z$ is locked and cannot change.
    \item \textbf{Transition:} After the $d$ observations have been emitted, the regime terminates deterministically, and the process returns to Step 1 to instantiate a new state.
\end{enumerate}
\end{definition}
\vspace{-3pt}
The prefix ``semi'' denotes that the Markov property is relaxed on the time dimension. In a standard HMM, the decision to stay or leave a state is made independently at every single time step. In an HSMM, the total duration is drawn once at the regime's birth, and the regime is then deterministically locked for that duration (making it non-Markovian on time).

\begin{notation}[HSMM state variables]
At each time step $t$, the full latent state of an HSMM is described by a triplet:
\vspace{-5pt}
\begin{itemize}[itemsep=-2pt]
    \item $z_t \in \{1, \dots, K\}$: the active hidden state index. In the context of standard changepoint detection (and our $K=1$ reduction), $z_t$ tracks the identity of the current regime $\rho$, meaning every changepoint instantiates a new set of parameters rather than jumping between pre-defined discrete classes.
    \item $d_t \in \mathbb{N}^*$: the total duration drawn at the birth of the current regime.
    \item $r_t \in \{0, 1, \dots, d_t - 1\}$: the run length, tracking the number of observations already emitted since the current regime began.
\end{itemize}
\end{notation}

\subsection{The HSMM-consistent data generation process}

To rigorously test our inference algorithms later, we must generate synthetic data from a generative process that is strictly consistent with the HSMM assumptions.
\vspace{-5pt}
\begin{definition}[HSMM Generative Process for Synthetic Data]\label{def:hsmm_gen}
Given a time horizon $T$, a prior distribution $\mathcal{N}(\mu_0, \sigma_0^2)$ for regime means, a known observation variance $\sigma^2$, and a chosen duration distribution, the data $x_{1:T}$ is generated as follows:
\vspace{-5pt}
\begin{enumerate}[itemsep=-2pt]
    \item \textbf{Initialize:} Set $t_{\text{start}} = 0$.
    \item \textbf{Draw a regime duration $d$:} Sample $d \sim P(d)$, rounded to the nearest integer and floored at $1$.
    \item \textbf{Draw regime mean:} Sample a regime mean $\mu_\rho \sim \mathcal{N}(\mu_0, \sigma_0^2)$.
    \item \textbf{Generate observations:} For $t = t_{\text{start}}, \dots, t_{\text{start}} + d - 1$, sample $x_t \sim \mathcal{N}(\mu_\rho, \sigma^2)$.
    \item \textbf{Advance:} Set $t_{\text{start}} \leftarrow t_{\text{start}} + d$. If $t_{\text{start}} < T$, return to Step 2.
\end{enumerate}
\end{definition}
\vspace{-3pt}
In this generative model, the total duration $d_t$ of a regime is predetermined at its birth. The regime terminates when its age reaches $r_t = d_t - 1$. However, from the perspective of an online inference algorithm receiving data tick-by-tick, this total duration $d_t$ is obviously unknown.

\subsection{From 3D Inference to 2D}

The fundamental challenge of online HSMM inference is the dimensionality of the latent space. An algorithm would seemingly need to maintain a posterior distribution over a three-dimensional object: $P(r_t, d_t, z_t \mid x_{1:t})$.

Updating this 3D joint distribution naively is impractical, and even with dynamic programming it costs $\mathcal{O}(K^2 + D^2 K)$ per step for exponential family UPMs \citep{agudelo2020bayesian}. However, a crucial insight, implemented in their \texttt{C++} codebase, enables a dimensionality reduction to a tractable 2D matrix $P(r_t, z_t \mid x_{1:t})$. It is made possible by the fact that the observation $x_t$ carries no information about the total duration $d_t$ of the segment.

\vspace{-5pt}
\begin{proposition}[Duration-agnostic UPM]\label{prop:cond_indep}
If the predictive likelihood is agnostic to the total segment duration, i.e. $P(x_t \mid r_t, d_t, z_t) = P(x_t \mid r_t, z_t)$, then $d_t$ enters the recursion only through the hazard function and can be marginalized out.
\end{proposition}

\vspace{-19pt}
\subsubsection*{From deterministic termination to probabilistic hazard}

\vspace{-5pt}
A question arises: \textit{if we discard the explicit total duration $d_t$ from the posterior, how does the algorithm retain the memory of the specific duration distribution chosen for the regimes?}

The answer lies in the fundamental difference between data generation (nature) and online inference (the algorithm). In the generative model, the total duration $d_t$ is drawn at birth; the regime's termination is therefore a deterministic event.

However, during online inference, the total duration $d_t$ is unknown. Because the algorithm does not know when the regime ends, it cannot make a deterministic decision. Instead, it maintains and evaluates a termination probability at every single time step. The algorithm answers the following question: \textit{Given that the current segment has survived up to age $r$, what is the probability that it terminates at the very next step?} This conditional probability is exactly the \textbf{hazard function}.

This represents the theoretical bridge between the 3D HSMM and the 2D BOSD framework, which is reflected in the \texttt{C++} codebase by \cite{agudelo2020bayesian} who implemented the exact 3D inference which maintains a tensor to track the full $(d, r, z)$ space, but also another class that operates on the 2D matrix $P(r_t, z_t \mid x_{1:t})$. The duration prior is thus compiled into a 1D hazard vector $H(r)$.

\subsection{The $K=1$ Reduction: our implementation}

The full BOSD (Bayesian Online Segment Detection) framework, as introduced by \cite{agudelo2020bayesian}, tracks the joint posterior $P(r_t, z_t \mid x_{1:t})$ over both the run length $r_t$ and a discrete hidden meta-state $z_t \in \{1, \dots, K\}$, where each meta-state possesses its own duration distribution $P(d \mid z)$ and its own emission model. We deliberately restrict the model to a single meta-state ($K=1$), which allows us to isolate the upgrade of the time-varying hazard rate.

\vspace{2pt}
\noindent\textbf{\textit{What the $K=1$ reduction preserves:}} Under the $K=1$ reduction, the posterior simplifies to a 1D vector, $P(r_t, z_t \mid x_{1:t}) \xrightarrow{K=1} P(r_t \mid x_{1:t})$, which is exactly the same shape as the standard BOCPD posterior. The following components of the original \cite{bocpd_2007} framework are fully preserved:

\vspace{-7pt}
\begin{itemize}[itemsep=-2pt]
    \item \textbf{The conjugate prior structure:} When a changepoint occurs, a new regime mean is drawn from the prior $\mathcal{N}(\mu_0, \sigma_0^2)$. The number of possible regime means remains infinite.
    \item \textbf{The Underlying Predictive Model:} The closed-form Gaussian predictive likelihood $P(x_t \mid r_{t-1}, x_{t-1}^{(r)}) = \mathcal{N}(\mu_{r_{t-1}}, \sigma^2 + \sigma^2_{r_{t-1}})$ is used without modification.
    \item \textbf{The sufficient statistics}.
\end{itemize}

\noindent\textbf{\textit{What the $K=1$ reduction changes:}} The only structural change is the replacement of the constant hazard rate $1/h$ by the time-varying hazard vector $H(r)$. The UPM and the message-passing structure remain identical.

\subsection{Deriving the Time-Varying Hazard Rate}

\begin{definition}[Hazard Rate function]\label{def:hazard}
Let $d$ be a discrete random variable representing the total duration of a regime. The \textbf{hazard rate function} is defined as the conditional probability that the regime terminates at step $r+1$, given that it has survived up to and including step $r$:
\vspace{-5pt}
\begin{equation}\label{eq:hazard}
    H(r) \;=\; P(d = r+1 \mid d > r) \;=\; \frac{P(d = r+1)}{P(d > r)}
\end{equation}
\end{definition}
\noindent The hazard functions of the duration laws used in this work (Gaussian, Pareto, log-normal, Poisson) are derived in Appendix~\ref{app:durations}.

\subsection{Benchmarking}

By comparing it against BOCPD with a constant hazard rate on HSMM-generated synthetic data, this model allows us to check the contribution of duration-awareness. The Monte Carlo validation of the regime-duration laws recovered by BOCPD and BOSD ($K=1$) is reported in Appendices~\ref{app:valid_bocpd} and~\ref{app:valid_bosd}.

\subsubsection*{A new metric: the cumulative Predictive log-Likelihood}

We introduce a third metric which is introduced in the paper of \cite{bocpd_2007}, and implemented in the \cite{agudelo2020bayesian} \texttt{C++} codebase. Recall that the BOCPD recursion produces the \textbf{predictive posterior distribution}.

This quantity is referred to as the \textit{one-step-ahead predictive likelihood} by \cite{agudelo2020bayesian}, which can be a bit confusing as we constantly refer to it as the \textbf{predictive posterior distribution}. Accumulating it over the whole sequence yields the desired metric, that we name the \textbf{Predictive log-Likelihood}.

\begin{definition}[Predictive log-Likelihood]\label{def:pred_loglik}
Given an observed sequence $x_{1:T}$, the \textbf{Predictive log-Likelihood} is
\vspace{-13pt}
\begin{equation}\label{eq:pred_loglik}
    \mathcal{L}(x_{1:T}) \;=\; \sum_{t=1}^{T} \log P(x_t \mid x_{1:t-1})
\end{equation}
\vspace{-8pt}
where $P(x_1 \mid x_{1:0}) := P(x_1)$ denotes the prior predictive likelihood at $t=1$.
\end{definition}

\noindent It rewards a model for assigning high probability to the data actually observed, and it does so using the \textit{entire} predictive posterior distribution $P(x_t \mid x_{1:t-1})$ at each step, not merely its mean. Unlike the Covering Metric, $\mathcal{L}(x_{1:t})$ requires no changepoint-labeling heuristic whatsoever.

\newpage
\section{Market microstructure and memory of order flow}\label{sec:microstructure}

\textit{Upgrading the hazard rate from a constant to a duration-dependent function requires specifying a distributional family for regime durations. We will see why it must accommodate heavy tails.}

Modern financial markets are \textbf{order-driven}: rather than a small number of designated dealers quoting prices (older ``quote-driven'' paradigm), every participant may submit orders that collectively determine the price \citep{cartea2015algorithmic}. The data structure that organises this is the \textbf{limit order book} (LOB). Most of the world's financial exchanges now operate a LOB mechanism to facilitate trade \citep{gould2013limit}.

The LOB is best understood through the two elementary actions available to any market participant: submitting an order and cancelling an order.
\vspace{-5pt}
\begin{definition}[Order]
An \textbf{order} $x = (p_x, \omega_x, t_x)$ submitted at time $t_x$ with price $p_x$ and size $\omega_x$ is a commitment to trade up to $|\omega_x|$ units of the asset at a price no worse than $p_x$. By convention, $\omega_x > 0$ denotes a sell order and $\omega_x < 0$ a buy order.
\end{definition}
\vspace{-5pt}
Two resolution parameters govern the granularity of trading. The \textbf{tick size} $\pi$ is the smallest permissible price increment. The \textbf{lot size} $\sigma$ is the smallest tradable quantity. We distinguish two categories of orders \citep{gould2013limit, cartea2015algorithmic}:

\textbf{Limit orders (LOs)} are \textit{passive} orders. A trader submitting a limit order specifies a price and a quantity, and the order joins the book without triggering an immediate trade. For a buy limit order, the specified price is typically at or below the current best buy price. Limit order submissions are said to \textit{provide liquidity} to the market.

\textbf{Market orders (MOs)} are \textit{aggressive} orders. A trader submitting a market order wishes to trade immediately at the best available price. A buy market order matches against the cheapest active sell orders; a sell market order matches against the most expensive active buy orders. Market order submissions are said to \textit{consume liquidity}.

\vspace{3pt} The active orders in the book partition naturally into the set of active buy orders $\mathcal{B}(t)$ (the \textit{bid side}) and the set of active sell orders $\mathcal{A}(t)$ (the \textit{ask side}).
\vspace{-2pt}
\begin{definition}[Bid price, ask price, mid price, and spread]
\label{def:prices}
Let $\mathcal{L}(t)$ be a limit order book at time $t$. The \textbf{bid price} is the highest stated price among active buy orders, $b(t) := \max_{x \in \mathcal{B}(t)} p_x$; the \textbf{ask price} is the lowest stated price among active sell orders, $a(t) := \min_{x \in \mathcal{A}(t)} p_x$. The \textbf{bid-ask spread} is $s(t) := a(t) - b(t)$, and the \textbf{mid price} is $m(t) := \frac{1}{2}[a(t) + b(t)]$.
\end{definition}
\vspace{3pt}
\begin{wrapfigure}{l}{0.42\textwidth}
  \centering
  \includegraphics[width=\linewidth]{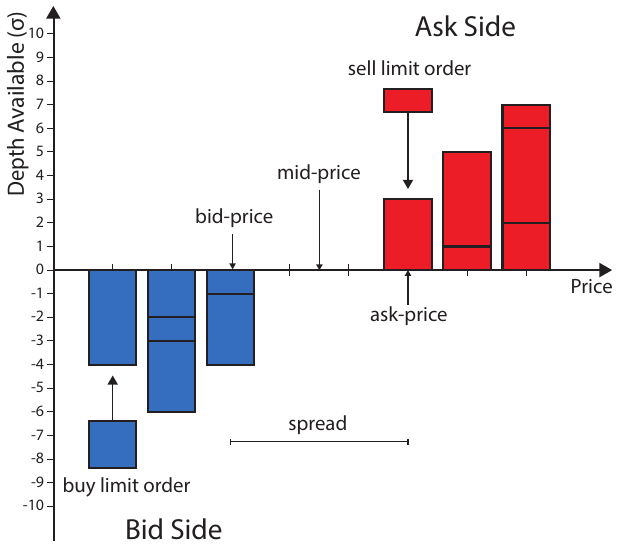}
  \caption{LOB: bid side, ask side, best prices, spread and depth profile.}
  \label{fig:lob}
\end{wrapfigure}

\noindent The bid price $b(t)$ is the best price at which one can immediately sell; the ask price $a(t)$ is the best price at which one can immediately buy. The spread $s(t)$ measures the cost of a round-trip transaction (buying and immediately reselling) and is a primary indicator of market liquidity.

Beyond the best prices, the \textbf{depth profile} describes the total quantity of active orders at each price level. When a large market order arrives whose size exceeds the depth available at the best price, it ``walks the book'': it consumes the entire depth at the best price level and then continues to match against orders at successively worse price levels, causing the best price to move. The interaction between aggressive orders and the available depth is the fundamental driver of price formation in electronic markets.

\subsection{Structure of order flow}

The \textbf{order flow} is the time series of market orders arriving at the exchange. Following the standard convention in the literature \citep{lillo2005theory, bouchaud2018trades, sato2023quantitative}, we encode the order flow as a sequence of signs:
\vspace{-5pt}
\begin{definition}[Order sign sequence]
\label{def:order_sign}
Let $\varepsilon(t) \in \{+1, -1\}$ denote the sign of the $t$-th market order, where $\varepsilon(t) = +1$ indicates a buy order and $\varepsilon(t) = -1$ indicates a sell order. The index $t \in \{1, 2, \ldots\}$ is the \textbf{tick time}: it is incremented by one unit at each transaction, regardless of the elapsed calendar time between transactions.
\end{definition}
\vspace{-5pt}
The choice of tick time rather than calendar time is deliberate and important. In calendar time, the order flow is highly irregular: there may be hundreds of transactions per second during peak activity and seconds of silence during quiet periods. In tick time, events are uniformly spaced by construction, which makes autocorrelation analysis meaningful, despite intraday seasonality patterns such as the well-documented U-shaped activity profile.
\vspace{-5pt}
\begin{definition}[Sign autocorrelation function]
\label{def:acf}
The autocorrelation function of the order flow is defined as $C(\tau) := \lim_{T \to \infty} \frac{1}{T} \sum_{t=1}^{T} \varepsilon(t)\, \varepsilon(t + \tau)$, where $\tau \geq 0$ is the time lag measured in ticks.
\end{definition}
\vspace{-5pt}
Because the order signs $\varepsilon(t)$ take values in $\{+1, -1\}$, their product $\varepsilon(t)\varepsilon(t+\tau)$ is strictly positive if both orders share the same direction (e.g., buy-buy), and negative otherwise. Consequently, the autocorrelation $C(\tau)$ essentially measures the excess probability that two transactions separated by $\tau$ ticks have the same sign. A strictly positive value, $C(\tau) > 0$, indicates persistence in the order flow.

\subsection{Long-range correlation in order flow}

In the early 2000s, two independent studies made a remarkable discovery. \cite{bouchaud2004fluctuations}, analysing the Paris Stock Exchange, and \cite{lillo2004long}, analysing the London Stock Exchange, both found that the autocorrelation function of the order sign sequence $\varepsilon(t)$ does not decay exponentially, but instead decays as a \textbf{power law}:
\begin{equation}\label{eq:lrc_empirical}
    C(\tau) \approx \frac{c_0}{\tau^{\gamma}}, \qquad \gamma \in (0, 1),
\end{equation}

The fact that $\gamma < 1$ has a profound mathematical consequence. The process has a non-integrable autocorrelation function: $\sum_{\tau=1}^{\infty} C(\tau) = \infty$. This divergence is the formal definition of a \textbf{long-memory} process. The correlation between distant events decays so slowly that arbitrarily distant past events still have a statistically significant influence on the present. This is different from the exponential decay of a short-memory process, where influence from the past fades over a characteristic timescale, which a long-memory process doesn't have.

Concretely, this means that buy orders tend to follow buy orders, and sell orders tend to follow sell orders, over time horizons as long as several weeks \citep{lillo2004long}. This persistence has been confirmed across many markets and asset classes: equities \citep{bouchaud2004fluctuations, lillo2004long}, foreign exchange, and cryptocurrency markets, making it a well-known fact of financial microstructure.

\vspace{3pt}
This property has consequences well beyond changepoint detection. A non-summable, power-law autocorrelation is exactly what characterises processes such as fractional Brownian motion with Hurst index $H > 1/2$ --- and such processes are not semimartingales, so that the classical Itô calculus does not apply and the theory of stochastic integration itself has to be rebuilt. This is the purpose of the \emph{calculus via regularization} introduced by \cite{russo_vallois_1993}, later extended to fractional Brownian motion through generalized covariations by \cite{gradinaru_russo_vallois_2003}.

\newpage
\section{Application to real market data}\label{sec:realdata}

We use high-frequency transaction data for two NASDAQ-listed stocks, Apple (AAPL) and Microsoft (MSFT), over the four months from August to November 2022, provided by the LOBSTER platform. For each stock, the raw record is the sequence of executed transactions, each characterised by its timestamp, its traded volume $V_k$ and its direction $\varepsilon_k \in \{+1, -1\}$, with $\varepsilon_k = +1$ for a buyer-initiated trade and $\varepsilon_k = -1$ for a seller-initiated one. Only transactions occurring during the continuous trading session are retained; opening and closing auctions are excluded.

The quantity we model is the \textbf{aggregated signed volume}. Rather than sampling in calendar time, which would make the series highly irregular because of the intraday variations of trading activity, we sample in \emph{volume clock}: transactions are grouped into non-overlapping buckets of $N$ consecutive executions, and the observation associated with bucket $t$ is the net signed volume traded within it,

\vspace{-19pt}
\begin{equation}
    x_t = \sum_{k=(t-1)N + 1}^{tN} \varepsilon_k \, V_k .
\end{equation}
A positive $x_t$ therefore indicates a bucket dominated by buyer-initiated volume, a negative one a bucket dominated by seller-initiated volume.

\vspace{3pt}\noindent\textit{NB:} to deal with the scale of the raw volumes without altering the mathematical topology of the posterior distribution, we apply a linear scaling factor of $10^{-3}$ to the series.

\vspace{-7pt}
\subsection{Calibration results}

We perform a hyper-parameter grid search over the 2022-08 dataset for each asset under two distinct selection criteria: minimizing the Mean Squared Error (MSE) of the one-step-ahead predictive posterior mean, and maximizing the cumulative Predictive log-Likelihood (Log-Lik).

The mean of the prior is fixed $\mu_0 = 0$. The variance is also fixed according to the global variance of a small sample of the dataset. So we perform the grid search over the variance of the normal prior $\sigma_0$ and the specific parameters of the algorithms.

\begin{table}[H]
    \centering \footnotesize
    \setlength{\tabcolsep}{10pt}
    \caption{Optimal parameters and resulting performance. \\ 
    MSFT ($N=542$): $2739$ points, mean $-0.062$, var $51.93$, model var $50$; \\
    AAPL ($N=887$): $2629$ points, mean $-5.777$, var $436.32$, model var $200$.}
    \label{tab:calib_all}
    \begin{tabular}{@{}lllccrr@{}}
        \toprule
        \textbf{Calibration} & \textbf{Asset} & \textbf{Model} & $\text{var}_0$ & \textbf{Parameters} & \textbf{MSE} & \textbf{Log-Lik} \\
        \midrule
        \multirow{9}{*}{\shortstack[l]{MSE-optimal \\ calibration}}
        & \multirow{3}{*}{MSFT}
            & BOCPD           & 30.0  & $H=1/6$                     & \textbf{0.8432} & $-9108.2$ \\
        &   & Pareto BOSD     & 20.0  & $\alpha=1.8,\ d_{\min}=3$   & 0.8552          & $-9103.8$ \\
        &   & Log-Normal BOSD & 30.0  & $s=1.0,\ \text{scale}=3$    & \textbf{0.8422} & $-9096.5$ \\
        \cmidrule(l){3-7}
        & \multirow{3}{*}{AAPL}
            & BOCPD           & 60.0  & $H=1/11$                    & 0.9213          & $-11679.4$ \\
        &   & Pareto BOSD     & 50.0  & $\alpha=1.05,\ d_{\min}=1$  & 0.9243          & $-11718.0$ \\
        &   & Log-Normal BOSD & 120.0 & $s=1.5,\ \text{scale}=1$    & \textbf{0.9075} & $-11474.0$ \\
        \midrule
        \multirow{9}{*}{\shortstack[l]{Log-Lik-optimal \\ calibration}}
        & \multirow{3}{*}{MSFT}
            & BOCPD           & 250.0 & $H=1/30$                    & 0.9577 & $-9007.4$ \\
        &   & Pareto BOSD     & 130.0 & $\alpha=1.05,\ d_{\min}=1$  & 0.9317 & $-9059.9$ \\
        &   & Log-Normal BOSD & 250.0 & $s=3.0,\ \text{scale}=2$    & 0.9116 & \textbf{$-8986.2$} \\
        \cmidrule(l){3-7}
        & \multirow{3}{*}{AAPL}
            & BOCPD           & 250.0 & $H=1/12$                    & 0.9814 & $-11323.3$ \\
        &   & Pareto BOSD     & 250.0 & $\alpha=1.05,\ d_{\min}=1$  & 0.9993 & $-11442.0$ \\
        &   & Log-Normal BOSD & 250.0 & $s=2.0,\ \text{scale}=1$    & 0.9209 & \textbf{$-11271.9$} \\
        \bottomrule
    \end{tabular}
\end{table}

\noindent\textbf{\textit{Calibration Summary:}} Log-Normal BOSD consistently achieves the best performance. This consistent dominance suggests that, among the duration laws considered, the log-normal's combination of a lighter initial hazard than the Pareto and heavier tail than the geometric distribution is best aligned with the persistence structure of these order-flow series.

\subsection{Out-of-Sample Comparison}

We now run the algorithms using their respective optimal calibrated parameters on every month for each asset, allowing us to perform out-of-sample comparisons. The corresponding run length posteriors are displayed in Appendix~\ref{app:rlpost}.

Since the true underlying change-points aren't known, the models can't be evaluated on their clustering capacity with the covering metric but we still have the mean squared error for forecasting capacity, and the predictive likelihood. For each dataset, we evaluate every model under both of its calibrated parameter sets obtained in the previous subsection.

\vspace{5pt}
\noindent\textbf{\textit{The Diebold-Mariano test:}} To assess whether the difference in forecasting accuracy between BOCPD and BOSD is statistically significant, we use the Diebold-Mariano test. For each time step, we compute the prediction error of both models and then look at the difference in their squared errors. If one model is truly better, this difference should be consistently positive or negative over time. This test checks whether the average of these loss differences is significantly different from zero by computing a test statistic and a p-value. A small p-value tells us the two models genuinely differ in accuracy, and the sign of the statistic tells us which one is better.

\begin{table}[H]
    \centering \small
    \setlength{\tabcolsep}{3.5pt}
    \renewcommand{\arraystretch}{0.95}
    \caption{Out-of-sample performance metrics, 2022-08 to 2022-11. For each asset, period and calibration criterion, the best value across the three models is in bold. $N_{cps}$: number of detected changepoints; $\overline{L}$: average regime length.}
    \label{tab:oos_all}
    \begin{tabular}{@{}ll rrrr rrrr@{}}
        \toprule
        & & \multicolumn{4}{c}{\textbf{MSE-calibrated}} & \multicolumn{4}{c}{\textbf{Log-Lik-calibrated}}\\
        \cmidrule(lr){3-6}\cmidrule(lr){7-10}
        \textbf{Period} & \textbf{Model} & $N_{cps}$ & $\overline{L}$ & MSE & Log-Lik & $N_{cps}$ & $\overline{L}$ & MSE & Log-Lik \\
        \midrule
        \multicolumn{10}{@{}l}{\textbf{MSFT} \textit{(N542)}}\\
        \multirow{3}{*}{2022-08}
          & BOCPD      & 86  & 31.5 & 0.8432 & $-9108.2$ & 63  & 42.8 & 0.9719 & $-9004.6$\\
          & Pareto     & 390 & 7.0  & 0.8552 & $-9103.8$ & 71  & 38.0 & \textbf{0.8675} & $-9070.1$\\
          & Log-Norm   & 66  & 40.9 & \textbf{0.8422} & \textbf{$-9096.5$} & 48 & 55.9 & 0.9475 & \textbf{$-8991.7$}\\
        \cmidrule(l){2-10}
        \multirow{3}{*}{2022-09}
          & BOCPD      & 178 & 15.1 & \textbf{0.8783} & $-9818.0$ & 99 & 27.0 & 1.1142 & $-9443.0$\\
          & Pareto     & 453 & 5.9  & 0.8922 & $-9962.4$ & 78 & 34.2 & \textbf{0.9162} & $-9803.6$\\
          & Log-Norm   & 118 & 22.7 & 0.8790 & \textbf{$-9811.6$} & 98 & 27.3 & 1.0505 & \textbf{$-9392.7$}\\
        \cmidrule(l){2-10}
        \multirow{3}{*}{2022-10}
          & BOCPD      & 153 & 20.3 & 0.8980 & $-10928.9$ & 100 & 30.9 & 1.0640 & $-10670.9$\\
          & Pareto     & 518 & 6.0  & 0.9141 & $-11042.9$ & 102 & 30.3 & \textbf{0.9183} & $-10926.0$\\
          & Log-Norm   & 117 & 26.4 & \textbf{0.8952} & \textbf{$-10918.4$} & 90 & 34.3 & 1.0210 & \textbf{$-10640.2$}\\
        \cmidrule(l){2-10}
        \multirow{3}{*}{2022-11}
          & BOCPD      & 89  & 22.0 & 0.8526 & $-6836.7$ & 56 & 34.8 & 0.9618 & $-6722.3$\\
          & Pareto     & 278 & 7.1  & 0.8507 & $-6837.5$ & 58 & 33.6 & \textbf{0.8632} & $-6797.5$\\
          & Log-Norm   & 82  & 23.9 & \textbf{0.8441} & \textbf{$-6814.5$} & 54 & 36.1 & 0.9375 & \textbf{$-6701.9$}\\
        \midrule
        \multicolumn{10}{@{}l}{\textbf{AAPL} \textit{(N887)}}\\
        \multirow{3}{*}{2022-08}
          & BOCPD      & 142 & 18.4 & 0.9250 & $-11522.9$ & 171 & 15.3 & 0.9649 & $-11388.0$\\
          & Pareto     & 54  & 47.8 & 0.9287 & $-11534.9$ & 61  & 42.4 & 0.9654 & $-11429.3$\\
          & Log-Norm   & 49  & 52.6 & \textbf{0.9119} & \textbf{$-11443.1$} & 104 & 25.0 & \textbf{0.9199} & \textbf{$-11353.9$}\\
        \cmidrule(l){2-10}
        \multirow{3}{*}{2022-09}
          & BOCPD      & 230 & 12.6 & 0.9040 & $-13071.6$ & 217 & 13.4 & 0.9279 & $-12966.9$\\
          & Pareto     & 97  & 29.7 & 0.9038 & $-13080.4$ & 111 & 26.0 & 0.9256 & $-12960.2$\\
          & Log-Norm   & 78  & 36.9 & \textbf{0.8984} & \textbf{$-12998.4$} & 143 & 20.2 & \textbf{0.8990} & \textbf{$-12912.7$}\\
        \cmidrule(l){2-10}
        \multirow{3}{*}{2022-10}
          & BOCPD      & 195 & 14.0 & 0.9761 & $-12847.3$ & 195 & 14.0 & 1.0451 & $-12564.6$\\
          & Pareto     & 35  & 76.3 & \textbf{0.9613} & $-12884.6$ & 77 & 35.2 & 1.0427 & $-12634.1$\\
          & Log-Norm   & 20  & 130.8& 0.9632 & \textbf{$-12700.2$} & 95 & 28.6 & \textbf{0.9780} & \textbf{$-12506.8$}\\
        \cmidrule(l){2-10}
        \multirow{3}{*}{2022-11}
          & BOCPD      & 124 & 12.7 & 0.8759 & $-6835.0$ & 130 & 12.1 & 0.8986 & $-6821.6$\\
          & Pareto     & 32  & 48.2 & 0.8787 & $-6839.0$ & 28  & 54.8 & \textbf{0.8724} & $-6815.7$\\
          & Log-Norm   & 13  & 113.5& \textbf{0.8751} & \textbf{$-6828.6$} & 66 & 23.7 & 0.8858 & \textbf{$-6813.3$}\\
        \bottomrule
    \end{tabular}
\end{table}

\newpage
\noindent\textbf{\textit{Summary of Algorithms' Performance:}} To facilitate a comprehensive review of the algorithms' performance, Table~\ref{tab:summary_mse} presents a consolidated view of the Global MSE across all out-of-sample datasets, under each of the two calibration criteria. For each dataset and criterion, the lowest MSE (indicating the most accurate forecast) is highlighted in bold.

\begin{table}[H]
    \centering \footnotesize
    \setlength{\tabcolsep}{5pt}
    \renewcommand{\arraystretch}{0.95}
    \caption{Summary of Global MSE across all out-of-sample datasets. \\ P-BOSD: Pareto BOSD; LN-BOSD: Log-Normal BOSD.}
    \label{tab:summary_mse}
    \begin{minipage}[t]{0.48\textwidth}
    \centering
    \textit{MSE-calibrated}\\[2pt]
    \begin{tabular}{@{}lccc@{}}
        \toprule
        \textbf{Dataset} & \textbf{BOCPD} & \textbf{P-BOSD} & \textbf{LN-BOSD} \\
        \midrule
        MSFT 2022-08 & 0.8432 & 0.8552 & \textbf{0.8422} \\
        MSFT 2022-09 & \textbf{0.8783} & 0.8922 & 0.8790 \\
        MSFT 2022-10 & 0.8980 & 0.9141 & \textbf{0.8952} \\
        MSFT 2022-11 & 0.8526 & 0.8507 & \textbf{0.8441} \\
        \cmidrule(r){1-4}
        AAPL 2022-08 & 0.9250 & 0.9287 & \textbf{0.9119} \\
        AAPL 2022-09 & 0.9040 & 0.9038 & \textbf{0.8984} \\
        AAPL 2022-10 & 0.9761 & \textbf{0.9613} & 0.9632 \\
        AAPL 2022-11 & 0.8759 & 0.8787 & \textbf{0.8751} \\
        \bottomrule
    \end{tabular}
    \end{minipage}\hfill
    \begin{minipage}[t]{0.48\textwidth}
    \centering
    \textit{Log-Lik-calibrated}\\[2pt]
    \begin{tabular}{@{}lccc@{}}
        \toprule
        \textbf{Dataset} & \textbf{BOCPD} & \textbf{P-BOSD} & \textbf{LN-BOSD} \\
        \midrule
        MSFT 2022-08 & 0.9719 & \textbf{0.8675} & 0.9475 \\
        MSFT 2022-09 & 1.1142 & \textbf{0.9162} & 1.0505 \\
        MSFT 2022-10 & 1.0640 & \textbf{0.9183} & 1.0210 \\
        MSFT 2022-11 & 0.9618 & \textbf{0.8632} & 0.9375 \\
        \cmidrule(r){1-4}
        AAPL 2022-08 & 0.9649 & 0.9654 & \textbf{0.9199} \\
        AAPL 2022-09 & 0.9279 & 0.9256 & \textbf{0.8990} \\
        AAPL 2022-10 & 1.0451 & 1.0427 & \textbf{0.9780} \\
        AAPL 2022-11 & 0.8986 & \textbf{0.8724} & 0.8858 \\
        \bottomrule
    \end{tabular}
    \end{minipage}
\end{table}

\section{From BOCPD to BOCPDMS: Model Selection}\label{sec:bocpdms}

\subsection{The problem with a single model}

In the standard BOCPD framework, we fix one model. For instance, we might decide upfront that ``the data within each regime is iid Gaussian with unknown mean.''

This is a very strong assumption, especially for multivariate financial data. The data could be far more complex than that and we might not be sure about the right structure. In the standard BOCPD, we pick one such model before running the algorithm and if we guess wrong, the algorithm's performance will suffer.

The BOCPDMS framework \citep{knoblauch2018bocpdms} addresses this by introducing a \textbf{model universe} $\mathcal{M} = \{m_1, \ldots, m_M\}$: a finite collection of candidate models. Instead of betting on a single model, the algorithm maintains a probability distribution over all models in $\mathcal{M}$ and lets the data decide, in real time, which one fits best.

\subsection{The new latent variable: the model index}

Recall that in standard BOCPD, the latent state at time $t$ is just the run length $r_t$. In BOCPDMS, the latent state becomes a \textbf{pair} $(r_t, m_t)$, where $r_t \in \{0, 1, 2, \ldots\}$ is the run length, the number of time steps since the last changepoint, and $m_t \in \mathcal{M}$ is the model index.

Basically, instead of yielding a 1D posterior $P(r_t \mid x_{1:t})$ as in BOCPD, the algorithm now provides a \textbf{2D posterior} $P(r_t, m_t \mid x_{1:t})$, a matrix where each entry tells us: ``given all the data observed so far, what is the probability that the current regime has been running for $r$ steps and is described by model $m$?''

\subsection{The model prior}

The model prior $q(m)$ is a distribution over $\mathcal{M}$ --- in the simplest case uniform, $q(m) = 1/M$. We introduce it through the generative perspective, then be precise about the inference process.

\vspace{5pt}
\textbf{Generative perspective.} If we simulate data, the model index evolves as a Markov chain coupled to the run length:
\vspace{-10pt}
\begin{equation}
    q(m_t \mid m_{t-1}, r_t) = \begin{cases}
        \mathbf{1}_{m_{t-1}}(m_t) & \text{if } r_t = r_{t-1} + 1 \quad \text{(no CP)} \\
        q(m_t) & \text{if } r_t = 0 \quad \text{(CP)}.
    \end{cases}
\end{equation}
In words: within a regime the active model is frozen (the indicator forces $m_t = m_{t-1}$); at a changepoint a new model is drawn from the model prior $q(m)$.

\vspace{5pt}
\textbf{What inference actually does.} The online algorithm never draws anything. It maintains $p(r_t, m_t \mid x_{1:t})$ over all pairs, and the transition rule steers the probability mass:
\begin{itemize}[itemsep=-2pt]
    \item \textbf{Growth branch ($r_t = r_{t-1}+1$):} the model index is carried over unchanged.
    \item \textbf{Changepoint branch ($r_t = 0$):} the probability mass is redistributed over the model universe in proportion to $q(m_t)$.
\end{itemize}

\noindent So $q(m)$ is not a mechanism for choosing a model; it is the weight with which each model is seeded at every potential changepoint. The algorithm does the actual selecting, by concentrating posterior mass on the model that explains the new regime best.

\subsection{The central recursive equation}

As in BOCPD, the heart of BOCPDMS is the message-passing recursive equation that updates, at each time step, the joint distribution which is $p(x_{1:t}, r_t, m_t)$ here.

\begin{adjustbox}{max width=\linewidth, center}
    $ p(x_{1:t}, r_t, m_t) = \sum_{m_{t-1}} \sum_{r_{t-1}} \Bigl\{ \underbrace{f_{m_t}(x_t \mid x_{1:(t-1)}, r_t)}_{\text{UPM for model } m_t} \; \underbrace{q(m_t \mid x_{1:(t-1)}, r_t, m_{t-1})}_{\text{model transition term}} \; \underbrace{p(r_t \mid r_{t-1})}_{\text{hazard rate}} \; \underbrace{p(x_{1:(t-1)}, r_{t-1}, m_{t-1})}_{\text{previous message}} \Bigr\} $
\end{adjustbox}

\begin{itemize}[itemsep=-2pt]
    \item $f_{m_t}(x_t \mid x_{1:(t-1)}, r_t)$ is the \textbf{Underlying Predictive Model} (UPM) for model $m_t$. ``Given that the regime has run length $r_t$ and is described by model $m_t$, how likely is the new observation $x_t$?'' As in the BOCPD framework, this is where conjugacy matters.
    \item $q(m_t \mid x_{1:(t-1)}, r_t, m_{t-1})$ is the \textbf{model transition term}: ``Keep the same model if no changepoint, and if there is a changepoint, give every model a given starting weight.''
    \item $p(r_t \mid r_{t-1})$ is the \textbf{hazard rate}.
    \item $p(x_{1:(t-1)}, r_{t-1}, m_{t-1})$ is the \textbf{message}, the joint distribution from the previous time step.
\end{itemize}

\noindent\textit{NB:} When $|\mathcal{M}| = 1$ (a single model), the sums over $m_t$ disappear, the model transition term becomes trivially 1, and we recover the recursion of the standard BOCPD.

\vspace{5pt}\noindent The recursion splits naturally into two cases.

\textbf{Growth probability ($r_t = r_{t-1} + 1$):} If the regime continues, the hazard rate is $1 - H(r_{t-1} + 1) = 1 - H(r_t)$ and the model stays the same ($m_t = m_{t-1}$):
\vspace{-5pt}
$$ p(x_{1:t}, r_t = r_{t-1}+1, m_t) = f_{m_t}(x_t \mid x_{1:(t-1)}, r_t) \cdot (1 - H(r_t)) \cdot p(x_{1:(t-1)}, r_{t-1}, m_t) $$

\vspace{-5pt}\textbf{Changepoint probability ($r_t = 0$):} If a changepoint occurs, the total probability mass is distributed across the new candidate models according to the prior $q(m_t)$:
\vspace{-5pt}
$$ p(x_{1:t}, r_t = 0, m_t) = f_{m_t}(x_t \mid x_{1:(t-1)}, r_t) \cdot q(m_t) \cdot \sum_{m_{t-1}} \sum_{r_{t-1}} \Big\{ H(r_{t-1}+1) \cdot p(x_{1:(t-1)}, r_{t-1}, m_{t-1}) \Big\} $$

\vspace{-8pt} Once we have the joint distribution $p(x_{1:t}, r_t, m_t)$ for all pairs of $(r_t, m_t)$, we compute the \textbf{evidence} $ p(x_{1:t}) = \sum_{m_t \in \mathcal{M}} \sum_{r_t = 0}^{t} p(x_{1:t}, r_t, m_t) $ which serves as the normalizing constant to calculate the proper posterior distribution from the joint distribution.

\subsection{The posteriors: what the algorithm gives us}

\hspace{1cm}\textbf{(a) The joint model-and-run-length posterior:} $p(r_t, m_t \mid x_{1:t}) = p(x_{1:t}, r_t, m_t)/p(x_{1:t})$. This is the full 2D posterior: an entry $(r, m)$ tells us the probability that the current regime has been running for $r$ steps and is described by model $m$, given everything we have observed.

\textbf{(b) The model posterior:} $p(m_t \mid x_{1:t}) = \sum_{r_t} p(r_t, m_t \mid x_{1:t})$. This is the marginal probability of each model, regardless of how long the current regime has been running: ``which model best describes the data?''

\textbf{(c) The global run-length distribution:} $p(r_t \mid x_{1:t}) = \sum_{m_t \in \mathcal{M}} p(r_t, m_t \mid x_{1:t})$. This is the direct analog of the BOCPD posterior $P(r_t \mid x_{1:t})$, and it can be used for changepoint detection in the same way.

\subsection{Prediction: Bayesian Model Averaging}

\cite{knoblauch2018bocpdms} call the one-step-ahead point forecast the \textbf{recursive forecast} $\widehat{\mathbf{x}}_{t+1} = \mathbb{E}(\mathbf{x}_{t+1} \mid \mathbf{x}_{1:t})$. It is exactly the object we called the \textbf{predictive mean} in the BOCPD framework, and we will keep calling it this way. The algorithm yields the posterior distribution from which we get the \textbf{posterior predictive distribution}:
\vspace{-5pt}
\begin{equation}\label{eq:post_pred}
    p(\mathbf{x}_{t+1} \mid \mathbf{x}_{1:t}) = \sum_{r_t} \sum_{m_t \in \mathcal{M}} p(\mathbf{x}_{t+1} \mid \mathbf{x}_{1:t}, r_t, m_t) \cdot p(r_t, m_t \mid \mathbf{x}_{1:t})
\end{equation}

\vspace{-8pt} Each hypothesis $(r_t, m_t)$ produces its own prediction, and the final prediction is a weighted average of all of them, weighted by how plausible each hypothesis is. To get a point forecast, we compute the predictive mean by taking the expectation of \eqref{eq:post_pred}:
\begin{align}
    \hat{\boldsymbol{\mu}}_{t+1} &= \mathbb{E}\!\left[\mathbf{x}_{t+1} \mid \mathbf{x}_{1:t}\right]
    = \int \mathbf{x}\; \sum_{r_t, m_t} p(\mathbf{x} \mid \mathbf{x}_{1:t}, r_t, m_t)\, p(r_t, m_t \mid \mathbf{x}_{1:t})\; d\mathbf{x} \nonumber \\
    &= \sum_{r_t, m_t} \underbrace{\mathbb{E}\!\left[\mathbf{x}_{t+1} \mid \mathbf{x}_{1:t}, r_t, m_t\right]}_{\text{mean of hypothesis } (r_t, m_t)} \cdot\; p(r_t, m_t \mid \mathbf{x}_{1:t}).
    \label{eq:pred_mean}
\end{align}

\textbf{Transposing the BOCPD predictive mean.} Recall the BOCPD case, where the predictive mean was $\hat{\mu}_t = \sum_{r_t} \mu_{r_t}\, P(r_t \mid x_{1:t})$: a posterior-weighted average of the regime-specific means $\mu_{r_t}$, each of which came out of the Gaussian UPM $\mathcal{N}(\mu_{r_t}, \sigma^2_{r_t} + \sigma^2)$.
\begin{align}
    \mathbb{E}\!\left[x_{t+1} \mid x_{1:t}\right]
    &= \int x \sum_{r_t} P(x \mid r_t, x_t^{(r)})\, P(r_t \mid x_{1:t})\, dx \nonumber \\
    &= \sum_{r_t} P(r_t \mid x_{1:t}) \underbrace{\int x\, P(x \mid r_t, x_t^{(r)})\, dx}_{\text{mean of } \mathcal{N}(\mu_{r_t},\, \sigma^2 + \sigma^2_{r_t})\ =\ \mu_{r_t}}
    \;=\; \sum_{r_t} \mu_{r_t}\, P(r_t \mid x_{1:t}).
\end{align}

\noindent In both frameworks the hypothesis-specific parameter is learned online and converges within a regime: $\mu_{r_t}$ starts at the prior mean $\mu_0$ and converges to the true regime level, exactly as $\hat{\mathbf{c}}_r$ starts at $\mathbf{0}$ and converges to the true $\mathbf{A}_l$. But there is a difference in what drives the forecast once learning has settled. In BOCPD, $\mu_{r_t}$ is essentially the running mean of the regime and it progressively flattens out onto the true level. In the BVAR, $\hat{\mathbf{c}}_r$ likewise settles, but the forecast $\hat{\mathbf{c}}_r^T\boldsymbol{\phi}_{t+1}$ never flattens, so the prediction keeps tracking the data no matter how long the regime has run. The BVAR forecast thus retains a permanent source of variation that the BOCPD forecast loses. This is the difference between predicting a \emph{level} and predicting a \emph{dynamic}.

\vspace{5pt}\textbf{The BVAR case.} Concretely, the mean of a multivariate Student-$t$ equals its location parameter as soon as $\nu_r = 2a_r > 1$: $\mathbb{E}[\mathbf{x}_{t+1} \mid \mathbf{x}_{1:t}, r_t, m_t] = \hat{\mathbf{c}}_r^{T}\,\boldsymbol{\phi}_{t+1}$.
\vspace{-5pt}
\begin{proposition}[Predictive mean]\label{prop:pred_mean_bocpdms}
The predictive mean, which is the first moment of the posterior predictive \eqref{eq:post_pred}, writes
\vspace{-8pt}
\begin{equation}
    \hat{\boldsymbol{\mu}}_{t+1} \;=\; \sum_{r_t} \sum_{m_t \in \mathcal{M}} \hat{\mathbf{c}}_r^{T}\,\boldsymbol{\phi}_{t+1} \cdot\; p(r_t, m_t \mid \mathbf{x}_{1:t}),
\end{equation}
\vspace{-5pt} a vector in $\mathbb{R}^S$: one forecast per component of the series.
\end{proposition}

\section{Bayesian parameter learning inside BOCPDMS}\label{sec:paramlearning}

In the standard BOCPD in which we suppose the data is distributed over a Gaussian with unknown mean $\mu$ and known variance $\sigma^2$, the unknown mean is learned on-the-fly via conjugate Bayesian updating. The posterior parameters $\mu_r$ and $\sigma_r^2$ depend on the run length $r$, and the UPM computes a different predictive distribution for each run length hypothesis:
$$p(x_t \mid r_{t-1}, x_{t-1}^{(r)}) = \mathcal{N}(x_t \mid \mu_{r_{t-1}}, \sigma_{r_{t-1}}^2 + \sigma^2)$$
where $\mu_{r_{t-1}}$ is the posterior mean of the regime parameter and $\sigma_{r_{t-1}}^2$ is the posterior variance, obtained by observing data points within the current regime (more data $\Rightarrow$ less uncertainty).

This is the \textbf{Bayesian learning} that makes BOCPD powerful: the algorithm does not just detect changepoints, it simultaneously learns the parameters of the regime online.

\subsection{Our multivariate model: the Bayesian VAR}

\cite{knoblauch2018bocpdms} propose filling the model universe $\mathcal{M}$ with \textbf{Bayesian Vector Autoregressions (BVAR)}, which are much richer than simple AR(1) and naturally include the Bayesian learning. A BVAR model of lag $L$ for an $S$-dimensional time series $\mathbf{x}_t = (x_t^{(1)}, \ldots, x_t^{(S)})^T$ takes the form:

\vspace{5pt}

\noindent
\begin{minipage}{0.49\textwidth}
    \begin{equation}
        \sigma^2 \sim \text{InverseGamma}(a, b) \label{eq:bvar1}
    \end{equation}
\end{minipage}%
\hfill
\begin{minipage}{0.49\textwidth}
    \begin{equation}
        \boldsymbol{\varepsilon}_t \mid \sigma^2 \sim \mathcal{N}(\mathbf{0},\ \sigma^2 \cdot \boldsymbol{\Omega}) \label{eq:bvar2}
    \end{equation}
\end{minipage}

\noindent
\begin{minipage}{0.49\textwidth}
    \begin{equation}
        \mathbf{c} \mid \sigma^2 \sim \mathcal{N}(\boldsymbol{\beta}_0=0,\ \sigma^2 \cdot \mathbf{V}_c) \label{eq:bvar3}
    \end{equation}
\end{minipage}%
\hfill
\begin{minipage}{0.49\textwidth}
    \begin{equation}
        \mathbf{x}_t = \boldsymbol{\alpha} + \mathbf{B}\mathbf{Z}_t + \sum_{l=1}^{L} \mathbf{A}_l \mathbf{x}_{t-l} + \boldsymbol{\varepsilon}_t \label{eq:bvar4}
    \end{equation}
\end{minipage}

\begin{itemize}[itemsep=-4pt]
    \item \textbf{Equation~\eqref{eq:bvar4}} is the observation equation: the current observation $\mathbf{x}_t$ is a linear combination of the $L$ most recent past observations through the coefficient matrices $\mathbf{A}_1, \ldots, \mathbf{A}_L$, plus exogenous variables $\mathbf{Z}_t$, plus an intercept $\boldsymbol{\alpha}$, plus noise.
    \item \textbf{Equation~\eqref{eq:bvar3}} is the \textbf{prior on the parameters}. All the model parameters (intercept, exogenous coefficients, autoregressive coefficients) are stacked into a single big vector $\mathbf{c}$, and we put a Gaussian prior on it.
    \item \textbf{Equation~\eqref{eq:bvar2}} is the noise model: the innovation $\boldsymbol{\varepsilon}_t$ is Gaussian with covariance $\sigma^2 \boldsymbol{\Omega}$, where $\boldsymbol{\Omega}$ is a \emph{known} diagonal matrix. Being diagonal, it rules out contemporaneous (same-time) correlation between locations; a single scalar $\sigma^2$ then sets the overall noise level. \cite{knoblauch2018bocpdms} keep $\boldsymbol{\Omega}$ general-diagonal in the model definition, but the NIG variant used in their reference implementation fixes $\boldsymbol{\Omega} = \mathbf{I}_S$, so all locations share the same noise variance $\sigma^2$.
    \item \textbf{Equation~\eqref{eq:bvar1}} is the prior on $\sigma^2$: an Inverse-Gamma distribution with shape $a$ and scale $b$. With the Gaussian prior on $\mathbf{c}$, it forms a Normal-Inverse-Gamma conjugate setup.
\end{itemize}

\textbf{The matrices $\mathbf{A}_l$ are unknown and learned from the data.} In the BVAR, the algorithm observes data within a regime and progressively refines its estimate of $\mathbf{A}_l$ (and $\sigma^2$) via Bayesian updating. The run length $r_t$ matters: a regime that has lasted 5 time steps has a very rough estimate of $\mathbf{A}_l$, while one that has lasted 200 has a precise estimate.

\textbf{Our setting.} We assume $\mathbf{Z}_t = \mathbf{0}$, $\boldsymbol{\alpha} = \mathbf{0}$ (no exogenous variables and no intercept), so the only parameters are the autoregressive matrices $\mathbf{A}_1, \ldots, \mathbf{A}_L$. We gather them into a single $k \times S$ coefficient matrix $\mathbf{c}$ (with $k = LS$), whose transpose stacks the $\mathbf{A}_l$ side by side: $\mathbf{c}^T = (\mathbf{A}_1, \ldots, \mathbf{A}_L) \in \mathbb{R}^{S \times LS}$. Recovering the $\mathbf{A}_l$ from $\mathbf{c}$ is just reading off blocks of its transpose.

\vspace{5pt}\noindent\textit{NB:} \cite{knoblauch2018bocpdms} write $\mathbf{c}$ as a vector stacking the vectorised coefficients, $\mathbf{c} = (\text{vec}(\mathbf{A}_1)^T, \ldots, \text{vec}(\mathbf{A}_L)^T)^T$, so that recovering $\mathbf{A}_l$ is a reshape. We instead use the equivalent matrix layout of their reference implementation, in which $\mathbf{c}$ is $k \times S$ and $\mathbf{A}_l$ is read off by transposition. The two describe the same coefficients; we adopt the matrix form because it matches the code and makes the update $\hat{\mathbf{c}} = \mathbf{F}\mathbf{W}$ read directly as a matrix product.

\subsection{Why conjugacy matters (again)}

The reason we choose the BVAR formulation is \textbf{conjugacy}. Recall that conjugacy is what gives us a closed-form Bayesian update for BOCPD: the posterior has the same functional form as the prior, and updating just means adjusting a few numbers (the sufficient statistics).

In the univariate BOCPD case ($\mu$ unknown and $\sigma^2$ known), we had the prior $\mu \sim \mathcal{N}(\mu_0, \sigma_0^2)$, the likelihood $x_t \mid \mu \sim \mathcal{N}(\mu, \sigma^2)$, and the posterior $\mu \mid x_{1:N} \sim \mathcal{N}(\mu_N, \sigma_N^2)$ with simple update rules. The BVAR generalizes this to the multivariate regression setting:
\begin{itemize}[itemsep=-2pt]
    \item \textbf{Prior:} $\mathbf{c} \mid \sigma^2 \sim \mathcal{N}(\mathbf{0}, \ \sigma^2 \mathbf{V}_c)$ and $\sigma^2 \sim \text{InverseGamma}(a, b)$.
    \item \textbf{Likelihood:} $\mathbf{x}_t \mid \mathbf{c}, \sigma^2 \sim \mathcal{N}(\mathbf{X}_t \mathbf{c}, \ \sigma^2 \boldsymbol{\Omega})$ where $\mathbf{X}_t$ is the regressor matrix (defined below).
    \item \textbf{Posterior:} Normal-Inverse-Gamma, with updated sufficient statistics.
\end{itemize}

This conjugacy guarantees that the Bayesian update can be done in \textbf{closed form}. The UPM also has a closed form --- it is a multivariate Student-$t$ distribution. All updates are \textbf{additive}. The full derivation, valid for any known SPD $\boldsymbol{\Omega}$, is given in Appendix~\ref{app:spd}.

\subsection{The Bayesian update for BVAR parameters}

Now we detail what the sufficient statistics are. This is the multivariate analog of the precision and mean update rules we derived for the univariate Gaussian in the BOCPD section.

Suppose the current regime has run length $r$, so it contains the $r+1$ observations $\mathbf{x}_{t-r}, \ldots, \mathbf{x}_t$. At each step $i$, the \textbf{regressor vector} $\boldsymbol{\phi}_i \in \mathbb{R}^{k}$ ($k = LS$) stacks the $L$ lagged observations: $\boldsymbol{\phi}_i = (\mathbf{x}_{i-1}^T, \ldots, \mathbf{x}_{i-L}^T)^T \in \mathbb{R}^{LS}$.

We then define two matrices: the \textbf{response matrix} $\mathbf{Y}$ and the \textbf{regressor matrix} $\mathbf{X}$:
$$\mathbf{Y}_{(t-r):t} = \begin{pmatrix} \mathbf{x}_{t-r}^T \\ \vdots \\ \mathbf{x}_t^T \end{pmatrix} \in \mathbb{R}^{(r+1)\times S},
\qquad
\mathbf{X}_{(t-r):t} = \begin{pmatrix} \boldsymbol{\phi}_{t-r}^T \\ \vdots \\ \boldsymbol{\phi}_t^T \end{pmatrix} \in \mathbb{R}^{(r+1)\times k}.$$

Row $i$ of $\mathbf{X}$ holds the predictors, row $i$ of $\mathbf{Y}$ the target: $\mathbf{X}$ is essentially $\mathbf{Y}$ shifted by one step, since each observation first plays the role of a response, then becomes a regressor for later steps. Dropping the $(t-r){:}t$ subscripts for readability, the sufficient statistics are:
\begin{align}
    \mathbf{F}(r, t) &= \left(\mathbf{X}^T \mathbf{X} + \mathbf{V}_c^{-1}\right)^{-1} \label{eq:F_update}\\
    \mathbf{W}(r, t) &= \mathbf{X}^T \mathbf{Y} \label{eq:W_update}
\end{align}

When a new observation arrives, we do not need to recompute everything from scratch. The updates are incremental:
\begin{itemize}[itemsep=-4pt]
    \item $\mathbf{W}(r, t) = \mathbf{W}(r-1, t-1) + \boldsymbol{\phi}_t\, \mathbf{x}_t^T$ --- a simple additive update.
    \item $\mathbf{P}(r, t) = \mathbf{P}(r-1, t-1) + \boldsymbol{\phi}_t\, \boldsymbol{\phi}_t^T$, where $\mathbf{P}(r, t) = \mathbf{F}(r, t)^{-1}$ is the precision.
\end{itemize}

The principle is the same: accumulate evidence, refine the estimate. Each data point tightens the posterior, whether it is over a scalar mean $\mu$ or a matrix of autoregressive coefficients $\mathbf{A}_l$.

\textbf{From sufficient statistics to the tracked coefficients.} The MAP estimate of the coefficient matrix is a deterministic function of the two statistics above:
$$\hat{\mathbf{c}}(r, t) = \mathbf{F}(r, t)\, \mathbf{W}(r, t) = \big(\mathbf{X}^T\mathbf{X} + \mathbf{V}_c^{-1}\big)^{-1}\mathbf{X}^T\mathbf{Y},$$
a $k \times S$ matrix whose transpose recovers the autoregressive coefficients, $\hat{\mathbf{c}}^T = (\hat{\mathbf{A}}_1, \ldots, \hat{\mathbf{A}}_L)$. This is a precision-weighted blend of the prior mean (zero) and the data. When $r$ is small, $\mathbf{V}_c^{-1}$ dominates and $\hat{\mathbf{A}}_l \approx \mathbf{0}$; as $r$ grows, $\mathbf{X}^T\mathbf{X}$ dominates and $\hat{\mathbf{c}}$ approaches the unregularised least-squares fit (the prior washes out).

$\hat{\mathbf{c}}(r,t)$ is not maintained recursively: the conjugate update only ever touches the sufficient statistics $(\mathbf{P}, \mathbf{W})$, and $\hat{\mathbf{c}}$ is recomputed from them. It is precisely this quantity, at the most likely run length $r^\ast_t$, that we track through time.

\textbf{Learning the variance $\sigma^2$.} The same statistics also drive the variance. The Normal-Inverse-Gamma prior is conjugate, so the posterior of $\sigma^2$ given run length $r$ is again Inverse-Gamma, $\sigma^2 \mid r, t \sim \text{InverseGamma}(a_r, b_r)$. The updated shape and scale are
$$a_r = a_0 + \frac{S(r+1)}{2}, \qquad
b_r = b_0 + \tfrac{1}{2}\Big(\mathrm{YY}_r + \boldsymbol{\beta}_0^T\mathbf{V}_c^{-1}\boldsymbol{\beta}_0 - \hat{\mathbf{c}}_r^T\mathbf{P}_r\hat{\mathbf{c}}_r\Big),$$
where $\mathrm{YY}_r$ is the response cross-product and $\boldsymbol{\beta}_0$ the prior coefficient mean. In our setting the prior is centred at zero ($\boldsymbol{\beta}_0 = \mathbf{0}$), so the middle term drops and the scale simplifies to
$$b_r = b_0 + \tfrac{1}{2}\Big(\mathrm{YY}_r - \operatorname{tr}\big(\boldsymbol{\Omega}^{-1}\mathbf{W}_r^T\mathbf{P}_r^{-1}\mathbf{W}_r\big)\Big),$$
which is the exact form implemented in the reference code. Both $a_r$ and $b_r$ are pure functions of the running sufficient statistics, so $\sigma^2$ is learned online at no extra cost.

The full posterior of $\sigma^2$ is never collapsed to a point during inference --- it is marginalised out to produce the Student-$t$ UPM of the next subsection, which is the whole benefit of NIG conjugacy. However, if we want to visualise how the variance is learned within a regime, a natural summary is the posterior mean of the Inverse-Gamma read at the most likely run length $r^\ast_t$: $\widehat{\sigma^2}(r,t) = b_r/(a_r - 1)$ for $a_r > 1$.

\subsection{The UPM with Bayesian parameter learning}

At time $t$, the algorithm simultaneously maintains multiple hypotheses about the run length: ``maybe $r_t = 0$ (a changepoint just happened), maybe $r_t = 1$, \ldots, maybe $r_t = t$ (no changepoint has ever occurred).'' For each of these hypotheses, the parameters of the model are \textit{different}. If $r_t = 2$, the posterior mean is based on only 2 observations and is very uncertain. If $r_t = 100$, the posterior mean is very precise.

This is why, in the BOCPD, the posterior mean $\mu_r$ and variance $\sigma_r^2$ are indexed by $r$ (one entry per possible run length). When the UPM is evaluated, it returns a vector of $t$ different probabilities --- one for each run length hypothesis. In the univariate BOCPD:
$$ p(x_t \mid r_{t-1}, x_{t-1}^{(r)}) = \mathcal{N}(x_t \mid \mu_{r_{t-1}}, \sigma_{r_{t-1}}^2 + \sigma^2) \qquad \text{for } r = 0, 1, \ldots, t-1$$

For the BVAR, we maintain a collection of sufficient statistics, $\{\mathbf{F}(r, t), \mathbf{W}(r, t), a_r, b_r\}$, one for each possible run length $r$. At each time step, we compute a different predictive probability for each $r$ using the posterior parameters specific to that run length: $f_m(\mathbf{x}_t \mid \mathbf{x}_{1:(t-1)}, r_t)$.
\newpage
Thanks to conjugacy, the UPM has a closed-form solution: it is a \textbf{multivariate Student-$t$} distribution. Writing $\boldsymbol{\phi}_t$ for the regressor vector that predicts $\mathbf{x}_t$, and writing $h_r = \boldsymbol{\phi}_t^{T}\mathbf{F}(r,t)\,\boldsymbol{\phi}_t$, its degrees of freedom, location and scale depend on the sufficient statistics accumulated over the $r+1$ observations of the current regime:
\vspace{-8pt}
\begin{equation}\label{eq:student_t_upm}
    f_m(\mathbf{x}_t \mid \mathbf{x}_{1:(t-1)}, r_t) = \text{Student-}t_{2a_r}\!\left(\mathbf{x}_t \;\Big|\; \hat{\mathbf{c}}_r^T \boldsymbol{\phi}_t, \;\; \frac{b_r}{a_r}\,(1 + h_r)\,\boldsymbol{\Omega}\right).
\end{equation}

\vspace{-8pt}
\textit{NB:} The Student-$t$ shape is a consequence of marginalising the unknown variance $\sigma^2$, not of the multivariate setting. If $\sigma^2$ is known and only the mean $\mathbf{c}$ is integrated out, Normal--Normal conjugacy leaves a Gaussian predictive --- this is the standard BOCPD setup, and the reason the univariate example above is Gaussian. If instead $\sigma^2$ is also unknown and marginalised, Normal-Inverse-Gamma conjugacy yields a Student-$t$ predictive.

\section{Learning the hyperparameters online}\label{sec:hyperlearning}

\subsection{Hyperparameter optimization}

Our BVAR model learns $\mathbf{c}$ and $\sigma^2$ from the data. But the priors on these quantities depend on four \textbf{hyperparameters} $\boldsymbol{\nu}_m$ that the Bayesian update never touches (Table~\ref{tab:hyperparams}).

\begin{table}[H]
    \centering \small
    \caption{The hyperparameters $\boldsymbol{\nu}_m$.}
    \label{tab:hyperparams}
    \begin{tabular}{|c|l|c|}
        \hline
        \textbf{Hyperparameter} & \textbf{Role} & \textbf{Shape} \\ \hline
        $a_0$ & Shape of the Inverse-Gamma prior on $\sigma^2$ & scalar \\ \hline
        $b_0$ & Scale of the Inverse-Gamma prior on $\sigma^2$ & scalar \\ \hline
        $\mathbf{V}_c$ & Prior covariance on the AR coefficients & $(k, k)$ \\ \hline
        $\boldsymbol{\Omega}$ & Noise structure matrix & $(S, S)$ \\ \hline
    \end{tabular}
\end{table}

These parameters shape the priors before any data is seen. And the problem is: if we set them badly, the algorithm's performance gets worse. In the BOCPD and BOSD framework, the hyperparameters were the parameters of the prior of regime means ($\mu_0$ and $\sigma_0$) and they were tuned by hand or with an offline calibration pass. Both are painful: hand-tuning is guesswork, and offline calibration needs a separate training phase before we can even start.

\cite{knoblauch2018bocpdms} propose two approaches. The first is an offline approach, similar to what we've done to ``calibrate'' the BOCPD: run BOCPDMS $K$ times on a training set $\mathbf{x}_{1:T'}$, and find the hyperparameters that maximize the Predictive log-Likelihood.

The second approach is much more appealing because it is fully \textbf{online}. The idea is to update the hyperparameters at every time step using gradient ascent \citep{caron2012} and it has two major advantages:

\vspace{-8pt}
\begin{enumerate}[itemsep=-2pt]
    \item \textbf{No separate training phase:} inference and hyperparameter learning happen simultaneously, which means the algorithm can be ``cold-started'' without any pre-tuning.
    \item \textbf{No computational overhead:} the gradient can be computed recursively alongside the main BOCPDMS recursion, so the online complexity is unchanged.
\end{enumerate}

\vspace{-8pt}
\noindent This is a very promising feature: we would not need to manually choose the hyperparameters and the algorithm would learn them from the data itself.

\subsection{The idea: online gradient ascent on the hyperparameters}

This is actually a very simple idea: \textbf{gradient ascent}. At every time step, after we observe the new data point, we ask ``would a small change in the hyperparameters have made our prediction better?'' If yes, we nudge them in that direction.
\begin{equation}
    \boldsymbol{\nu}_{m, t+1} = \boldsymbol{\nu}_{m, t} + \alpha_t \, \nabla_{\boldsymbol{\nu}_{m,t}} \log p(\mathbf{x}_{t+1} \mid \mathbf{x}_{1:t}, \boldsymbol{\nu}_{m_{1:t}})
\end{equation}
\begin{itemize}[itemsep=-2pt]
    \item $\boldsymbol{\nu}_{m,t}$ is the current value of the hyperparameters of model $m$ at time $t$.
    \item $\log p(\mathbf{x}_{t+1} \mid \mathbf{x}_{1:t}, \boldsymbol{\nu}_m)$ is the \textbf{predictive log-likelihood}: how well the model predicted the incoming point using everything up to time $t$. High value = good prediction.
    \item $\nabla_{\boldsymbol{\nu}_m}$ is the gradient, and $\alpha_t$ is the \textbf{learning rate}: how big a step we take.
\end{itemize}

So, after seeing $\mathbf{x}_{t+1}$, we look back at the prediction we made for it, and move the hyperparameters a little in the direction that would have made that prediction more likely.

We need the gradient of the predictive log-likelihood which is a weighted sum, over all run lengths, of the UPM at each run length, weighted by how plausible that run length is:
\begin{equation}
    \log p(\mathbf{x}_{t+1} \mid \mathbf{x}_{1:t}, \boldsymbol{\nu}_m) = \log \sum_{r=0}^{t} \underbrace{p(\mathbf{x}_{t+1} \mid r_t, m, \boldsymbol{\nu}_m)}_{\text{UPM at run length } r} \cdot \underbrace{p(r_t, m \mid \mathbf{x}_{1:t})}_{\text{posterior weight}}
\end{equation}
Because the log-predictive is a logarithm of a sum, its gradient has a clean form:
\begin{equation}\label{eq:score_grad}
    \nabla_{\boldsymbol{\nu}_m} \log p(\mathbf{x}_{t+1} \mid \mathbf{x}_{1:t}, \boldsymbol{\nu}_m) = \frac{\displaystyle\sum_{r} p(r_t, m \mid \mathbf{x}_{1:t}) \cdot \nabla_{\boldsymbol{\nu}_m} \, p(\mathbf{x}_{t+1} \mid r_t, m, \boldsymbol{\nu}_m)}{\displaystyle\sum_{r} p(r_t, m \mid \mathbf{x}_{1:t}) \cdot p(\mathbf{x}_{t+1} \mid r_t, m, \boldsymbol{\nu}_m)}
\end{equation}

So the gradient is a \textbf{weighted average of the gradients of each individual UPM}, where the weights are the posterior probabilities. So the run lengths and models that the algorithm currently believes in the most contribute the most to the hyperparameter update which is intuitive: we trust the gradient coming from the hypotheses we currently find most plausible.

So, once we know how to differentiate a single Student-$t$ UPM with respect to a hyperparameter, we are done: we compute that derivative at every run length, average them with the posterior weights, and plug the result into the update.

\subsection{Which hyperparameters do we actually learn?}

In theory, we could learn all four. In practice, we do not. Optimizing the two scalars $a_0$ and $b_0$ is cheap and stable. Optimizing the full matrices $\mathbf{V}_c$ and $\boldsymbol{\Omega}$ is a different story: they are covariance matrices, so every gradient step would have to keep them positive-definite, which is fragile and expensive. So we restrict ourselves to the two Inverse-Gamma scalars.

\textbf{$\bullet \ b_0$ (the scale of the prior on $\sigma^2$):} This is the simplest and most impactful, because $b_0$ directly sets the a priori scale of the variance.
$$ b_r = b_0 + \tfrac{1}{2}\Big(\mathrm{YY}_r - \operatorname{tr}\big(\boldsymbol{\Omega}^{-1}\mathbf{W}_r^{T}\mathbf{P}_r^{-1}\mathbf{W}_r\big)\Big) \ \Longrightarrow\  \frac{\partial b_r}{\partial b_0} = 1 \ \Longrightarrow\  \frac{\partial \log f_r}{\partial b_0} = \frac{\partial \log f_r}{\partial b_r} $$

\textit{NB:} the general formula additionally carries a $\boldsymbol{\beta}_0^{T}\mathbf{V}_c^{-1}\boldsymbol{\beta}_0$ term that vanishes here because the coefficient prior is centred at zero ($\boldsymbol{\beta}_0=\mathbf{0}$).

\textbf{$\bullet \ a_0$ (the shape of the prior on $\sigma^2$):}
$$ a_r = a_0 + \tfrac{S(r+1)}{2} \ \Longrightarrow\  \frac{\partial a_r}{\partial a_0} = 1 \ \Longrightarrow\  \frac{\partial \log f_r}{\partial a_0} = \frac{\partial \log f_r}{\partial a_r} $$

\textit{NB:} The reference implementation of \cite{knoblauch2018bocpdms} increments the shape by $1/2$ per time step rather than per scalar observation, i.e. $a_r = a_0 + (r+1)/2$, while its scale $b_r$ accumulates $S$ squared residuals per step. The two are inconsistent for $S>1$: the posterior mean $b_r/(a_r-1)$ then converges to $S\sigma^2$ instead of $\sigma^2$. We verified this empirically on synthetic data with $S=2$ and $\sigma^2=1$, where the estimate settles at $2.0$.

\subsection{The core computation: differentiating the Student-t}

It now comes down to differentiating the log-density of a multivariate Student-$t$. For an $S$-dimensional Student-$t$ with $\nu$ degrees of freedom, location $\boldsymbol{\mu}$ and scale matrix $\boldsymbol{\Sigma}$, it is:
\begin{equation}\label{eq:student_t_logdensity}
    \log f(\mathbf{x} \mid \nu, \boldsymbol{\mu}, \boldsymbol{\Sigma}) = C(\nu) - \tfrac{1}{2}\log|\boldsymbol{\Sigma}| - \frac{\nu + S}{2}\log\!\Big(1 + \tfrac{1}{\nu}(\mathbf{x} - \boldsymbol{\mu})^{T}\boldsymbol{\Sigma}^{-1}(\mathbf{x} - \boldsymbol{\mu})\Big)
\end{equation}
with $C(\nu) = \log\Gamma(\tfrac{\nu+S}{2}) - \log\Gamma(\tfrac{\nu}{2}) - \tfrac{S}{2}\log(\nu\pi)$. In our BVAR, the density is evaluated at $\mathbf{x} = \mathbf{x}_t$, the location is $\boldsymbol{\mu}_r = \hat{\mathbf{c}}_r^{T}\boldsymbol{\phi}_t$, the degrees of freedom are $\nu_r = 2a_r$, and the scale is $\boldsymbol{\Sigma}_r = \frac{b_r}{a_r}(1 + h_r)\boldsymbol{\Omega}$ with $h_r = \boldsymbol{\phi}_t^{T}\mathbf{P}_r^{-1}\boldsymbol{\phi}_t$.

A quantity that will keep coming back is $Q_r = \tfrac{1}{\nu_r}(\mathbf{x} - \boldsymbol{\mu}_r)^{T}\boldsymbol{\Sigma}_r^{-1}(\mathbf{x} - \boldsymbol{\mu}_r)$. The good news is that $Q_r$ is computed when we evaluate the UPM, so we can grab it for free.

\begin{proposition}[$Q_r$ does not depend on $a_r$]\label{prop:Qr_indep_a}
Writing $\boldsymbol{\Sigma}_r^{-1} = \tfrac{a_r}{b_r}\mathbf{M}_r^{-1}$ with $\mathbf{M}_r = (1+h_r)\boldsymbol{\Omega}$ (a matrix that involves neither $a_r$ nor $b_r$), we get
$$Q_r = \frac{1}{2a_r}\cdot\frac{a_r}{b_r}\cdot K_r = \frac{K_r}{2b_r}, \qquad K_r := (\mathbf{x}-\boldsymbol{\mu}_r)^T\mathbf{M}_r^{-1}(\mathbf{x}-\boldsymbol{\mu}_r).$$
The factor $a_r$ cancels exactly. \textbf{$Q_r$ is a function of $b_r$ only}, never of $a_r$.
\end{proposition}

\subsubsection*{Derivative with respect to $b_r$}

Since the scale is proportional to $b_r$ ($\partial \boldsymbol{\Sigma}_r/\partial b_r = \boldsymbol{\Sigma}_r/b_r$), differentiating each term of \eqref{eq:student_t_logdensity} with respect to $b_r$ (using $\partial Q_r/\partial b_r = -Q_r/b_r$ from Proposition~\ref{prop:Qr_indep_a}) gives:
\begin{equation}\label{eq:grad_b}
    \boxed{\dfrac{\partial \log f_r}{\partial b_r} = \dfrac{1}{2b_r}\left[\dfrac{(\nu_r+S)\,Q_r}{1+Q_r} - S\right]}
\end{equation}

\subsubsection*{Derivative with respect to $a_r$}

Because $a_r$ affects both $\nu_r=2a_r$ and the scale through $b_r/a_r$, we differentiate each of the three terms of \eqref{eq:student_t_logdensity} with respect to $a_r$, holding $b_r$ fixed:
\begin{itemize}[itemsep=-2pt]
    \item $C(2a_r)$ contributes $\psi(a_r+\tfrac{S}{2}) - \psi(a_r) - \frac{S}{2a_r}$;
    \item $-\tfrac12\log|\boldsymbol{\Sigma}_r|$, i.e.\ $-\tfrac{S}{2}\log(b_r/a_r)+\text{const}$, contributes $+\frac{S}{2a_r}$;
    \item the log-term contributes $-\log(1+Q_r)$, and \textit{nothing else} $\partial Q_r/\partial a_r = 0$.
\end{itemize}
The two terms in $S/(2a_r)$ cancel \textbf{exactly}, leaving a remarkably compact result:
\begin{equation}\label{eq:grad_a}
    \boxed{\dfrac{\partial \log f_r}{\partial a_r} = \psi\!\Big(a_r+\dfrac{S}{2}\Big) - \psi(a_r) - \log(1+Q_r)}
\end{equation}

\noindent This is the good, verified form from the reference code of \cite{knoblauch2018bocpdms}.

\newpage
\subsection{Can the noise matrix be non-diagonal?}\label{sec:omega}

Recall the noise model of the BVAR: $\boldsymbol{\varepsilon}_t \mid \sigma^2 \sim \mathcal{N}(\mathbf{0}, \sigma^2 \boldsymbol{\Omega})$, where \cite{knoblauch2018bocpdms} take $\boldsymbol{\Omega}$ to be a \emph{known} diagonal matrix, and where the reference implementation (and our experiments so far) fixes $\boldsymbol{\Omega} = \mathbf{I}_S$. But actually the diagonality restriction is not essential: 
\vspace{-5pt}
\begin{itemize}[itemsep=-2pt]
    \item \textbf{Any known SPD matrix $\boldsymbol{\Omega}$ preserves conjugacy.} The Normal-Inverse-Gamma machinery (closed-form posterior, additive sufficient statistics, Student-$t$ UPM) carries over to a full, non-diagonal $\boldsymbol{\Omega}$, as long as it is fixed. We sketch this below and derive it in Appendix~\ref{app:spd}; the underlying result is standard in Bayesian multivariate regression.
    \item \textbf{What breaks the framework is learning $\boldsymbol{\Omega}$}, not its shape. An unknown noise covariance requires a different conjugate family (Inverse-Wishart), which turns out to be compatible with the full BVAR but incompatible with the SSBVAR, and this, with two further arguments, is the reason behind the choice of \cite{knoblauch2018bocpdms}.
\end{itemize}

\vspace{-8pt}
\noindent This dividing line is exactly the one drawn by \cite{bishop2006pattern}: in Section~3.3 of his book, he develops conjugate Bayesian regression under a noise precision that is \emph{known}, and nothing in the derivation requires it to be isotropic or diagonal. Conversely, as soon as the covariance is unknown, the conjugate family must change: Gaussian-gamma in the univariate case, Gaussian-Wishart in the multivariate case (his Section~2.3.6).

\subsubsection*{Conjugacy holds for any known SPD Omega}

With $\boldsymbol{\Omega} = \mathbf{I}_S$ that \cite{knoblauch2018bocpdms} actually use, the noise is isotropic, and the prior on the stacked coefficient vector is $\mathbf{c} \mid \sigma^2 \sim \mathcal{N}(\mathbf{0}, \sigma^2 \mathbf{V}_c)$. Column by column, this is $S$ independent conjugate Bayesian linear regressions with known noise precision, exactly as in \cite{bishop2006pattern}, Section~3.3, so conjugacy (closed-form posterior, additive sufficient statistics, Student-$t$ predictive) is immediate.

\textbf{Where diagonality is (not) used.} Writing the posterior of $\mathbf{c}$ with Bayes' rule, and completing the square, we find that the posterior precision and mean are $\mathbf{P} = \mathbf{V}_c^{-1} + \mathbf{X}^T\mathbf{X}$ and $\hat{\mathbf{c}} = \mathbf{P}^{-1}\mathbf{X}^T\mathbf{Y}$, and that $\boldsymbol{\Omega}^{-1}$ factors out, leaving $\mathbf{P}$ and $\hat{\mathbf{c}}$ free of $\boldsymbol{\Omega}$. Diagonality of $\boldsymbol{\Omega}$ is used nowhere. \textbf{The NIG framework is conjugate for any fixed SPD $\boldsymbol{\Omega}$}, so the entire BOCPDMS machinery runs unmodified.

\textbf{Our implementation already supports it.} It is written for a general $\boldsymbol{\Omega}$ from the start. It places $\boldsymbol{\Omega}$ only in the noise-related quantities (the cross-product $\mathbf{x}_t^T\boldsymbol{\Omega}^{-1}\mathbf{x}_t$, the trace term in $b_r$, and the predictive scale) and keeps the precision $\mathbf{P}_r$ free of $\boldsymbol{\Omega}$. This is precisely the coupled-prior configuration; indeed, the fact that $\mathbf{P}_r$ is stored as a single $k \times k$ matrix is the signature of that choice, since the plain prior would not factor into one. For $\boldsymbol{\Omega} = \mathbf{I}_S$ it reduces to the plain prior of \cite{knoblauch2018bocpdms}, so our runs so far are unaffected by the distinction, and passing a non-diagonal SPD $\boldsymbol{\Omega}$ requires no structural change to the code.

\newpage
\section{Experiments on synthetic data}\label{sec:synthexp}

This section tests our BOCPDMS implementation through a series of controlled experiments on synthetic data. In all of them the model universe $\mathcal{M}$ contains a single, well-specified $\mathrm{VAR}(1)$ model, so the model-selection machinery is not yet exercised: we first isolate and validate the Bayesian-learning engine and the online hyperparameter optimisation.

\subsection*{Setup and reporting conventions}

The algorithm is never told the true coefficients, the true innovation variance, or the changepoint locations: it only receives the stream $\mathbf{x}_{1:T}$ and the prior specification. Unless stated otherwise, every experiment uses $S = 2$, $L = 1$, $T = 500$, $H = 1/150$.

\textbf{Hyperparameters.} In BOCPD the hyperparameters were $(\mu_0, \sigma_0)$: a prior mean and spread for the regime mean. Here they are less directly readable, so Table~\ref{tab:read_hyper} details what each one controls. The unknown scale $\sigma^2$ is shared by the noise and the coefficient prior; $\mathbf{V}_c$ and $\boldsymbol{\Omega}$ are shape matrices, held fixed (only $a_0, b_0$ are learned online).

\begin{table}[H]
    \centering \small
    \caption{Reading the hyperparameters $\boldsymbol{\nu}_m$.}
    \label{tab:read_hyper}
    \begin{tabular}{|c|p{4.1cm}|p{5.2cm}|c|}
        \hline
        & \textbf{What it fixes} & \textbf{How to read it} & \textbf{Value} \\ \hline
        $\boldsymbol{\Omega}$ & relative noise level of each component & $\mathbf{I}_S$: all components share one variance $\sigma^2$; $\mathrm{diag}(\omega_s)$ would give component $s$ the variance $\sigma^2\omega_s$ & $\mathbf{I}_S$  \\ \hline
        $a_0$ & concentration of the prior on $\sigma^2$ & the confidence: $a_r = a_0 + S(r+1)/2$, so each step adds $S/2$ & $2.0$ \\ \hline
        $b_0$ & a priori scale of $\sigma^2$ & sets the prior mean $\mathbb{E}[\sigma^2] = b_0/(a_0-1) $ & $0.5$ \\ \hline
        $\mathbf{V}_c$ & width of the coefficient prior ($\sim$ analogue of $\sigma_0$) & prior std $\sqrt{v}\,\sigma$ per coefficient, i.e. $10\sigma$ & $10^2\cdot\mathbf{I}_k$ \\ \hline
    \end{tabular}
\end{table}

\textbf{The learning figures.} The light curve is the estimate at $r^\ast_t = \argmax_{r,m} p(r_t, m_t \mid \mathbf{x}_{1:t})$ and the darker one is the posterior-averaged reading that weights every hypothesis by the model-specific run-length distribution:
\vspace{-5pt}
$$\bar{\mathbf{c}}_t = \sum_{r_t} \hat{\mathbf{c}}(r_t, t)\, p(r_t \mid m^\ast_t, \mathbf{x}_{1:t}),
\qquad \bar{\sigma}^2_t = \sum_{r_t} \frac{b_{r_t}}{a_{r_t}-1}\, p(r_t \mid m^\ast_t, \mathbf{x}_{1:t}).$$

\vspace{-5pt}\noindent Every panel of Figures~\ref{fig:exp1} to~\ref{fig:exp6_7} reports, for the corresponding experiment: (a) data; (b) run-length posterior; (c) online learning of $\hat{\mathbf{A}}$; (d) online learning of $\widehat{\sigma^2}$; (e) online hyperparameters learning.

\subsection{One change in the AR matrix (diagonal matrices)}

We simulate a bivariate VAR(1) of length $T = 500$ with a single changepoint at $t^\star = 250$:
$$\mathbf{x}_t = \mathbf{A}\,\mathbf{x}_{t-1} + \boldsymbol{\varepsilon}_t, \qquad \boldsymbol{\varepsilon}_t \overset{iid}{\sim} \mathcal{N}(\mathbf{0}, \sigma^2 \mathbf{I}_2), \qquad \sigma^2 = 1,$$
$$\mathbf{A}^{(1)} = \begin{pmatrix} 0.9 & 0 \\ 0 & 0.7 \end{pmatrix}
    \;\longrightarrow\;
    \mathbf{A}^{(2)} = \begin{pmatrix} 0.1 & 0 \\ 0 & 0.2 \end{pmatrix}.$$
Both matrices are diagonal, so the two components evolve independently: $x_t^{(1)} = a_{11}x_{t-1}^{(1)} + \varepsilon_t^{(1)}$ and $x_t^{(2)} = a_{22}x_{t-1}^{(2)} + \varepsilon_t^{(2)}$.

\begin{figure}[H]
    \centering
    \includegraphics[width=0.7\textwidth]{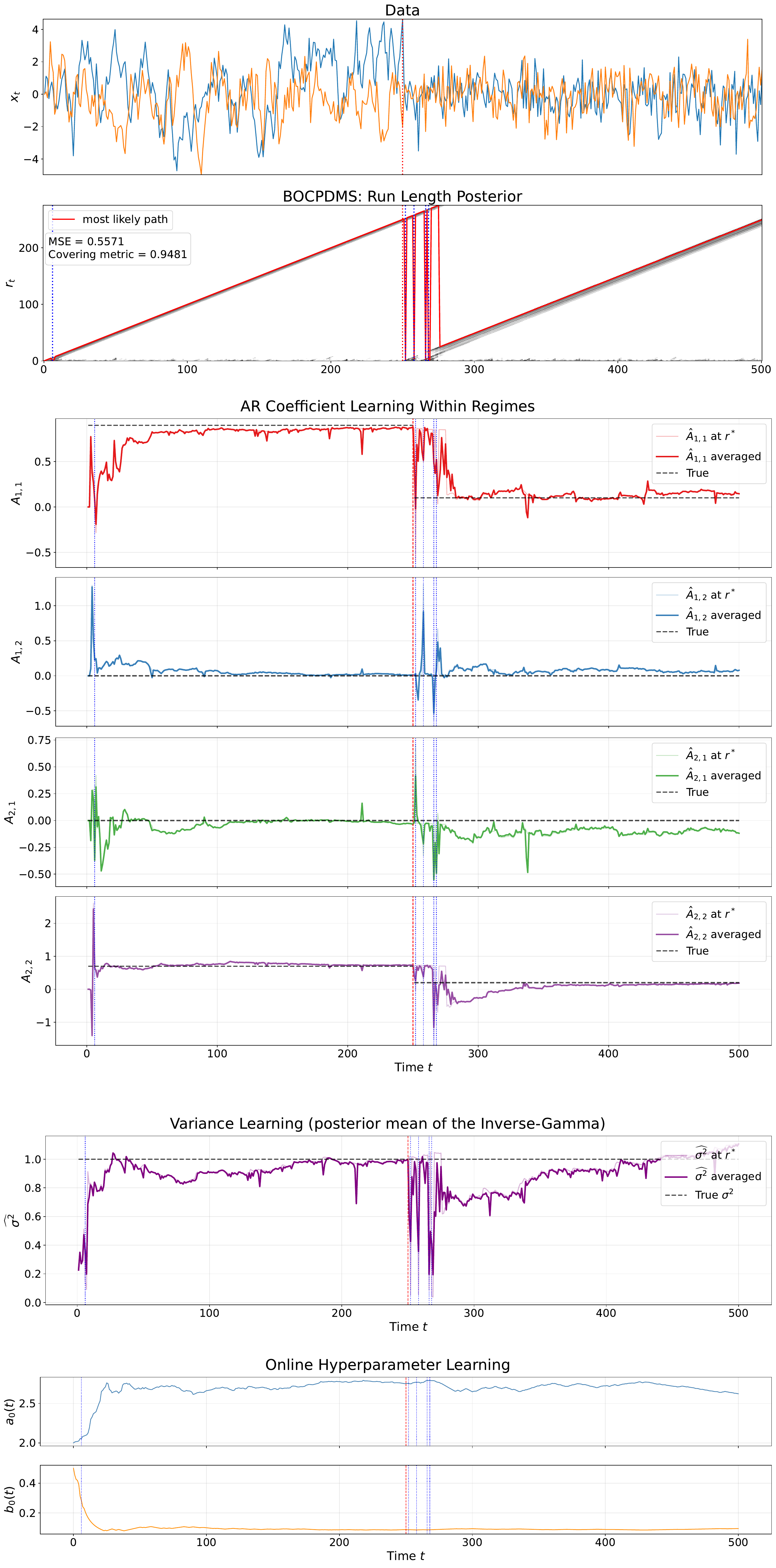}
    \caption{\textit{Experiment 1.} (a) data; (b) run-length posterior; (c) online learning of $\hat{\mathbf{A}}$; (d) online learning of $\widehat{\sigma^2}$; (e) online hyperparameters learning.}
    \label{fig:exp1}
\end{figure}

\subsection{One change in the AR matrix (cross-dependencies)}

Same setting, but the post-changepoint matrix now carries cross terms:
$$\mathbf{A}^{(1)} = \begin{pmatrix} 0.9 & 0 \\ 0 & 0.7 \end{pmatrix}
    \;\longrightarrow\;
    \mathbf{A}^{(2)} = \begin{pmatrix} 0.1 & 0.3 \\ 0.4 & 0.2 \end{pmatrix}.$$
See Figure~\ref{fig:exp2_3}(a).

\subsection{One change in the variance of the noise}

Same bivariate VAR(1), but now the autoregressive matrix is held constant ($\mathbf{A} = \diag(0.9, 0.7)$), and the innovation variance changes:
$$\boldsymbol{\varepsilon}_t \sim \mathcal{N}(\mathbf{0}, \sigma_{(1)}^2\mathbf{I}_2) \;\longrightarrow\; \mathcal{N}(\mathbf{0}, \sigma_{(2)}^2\mathbf{I}_2), \qquad \sigma_{(1)}^2 = 1.0,\ \sigma_{(2)}^2 = 5.0$$
See Figure~\ref{fig:exp2_3}(b).

\begin{figure}[H]
    \centering
    \begin{subfigure}[b]{0.48\linewidth}
        \includegraphics[width=\linewidth]{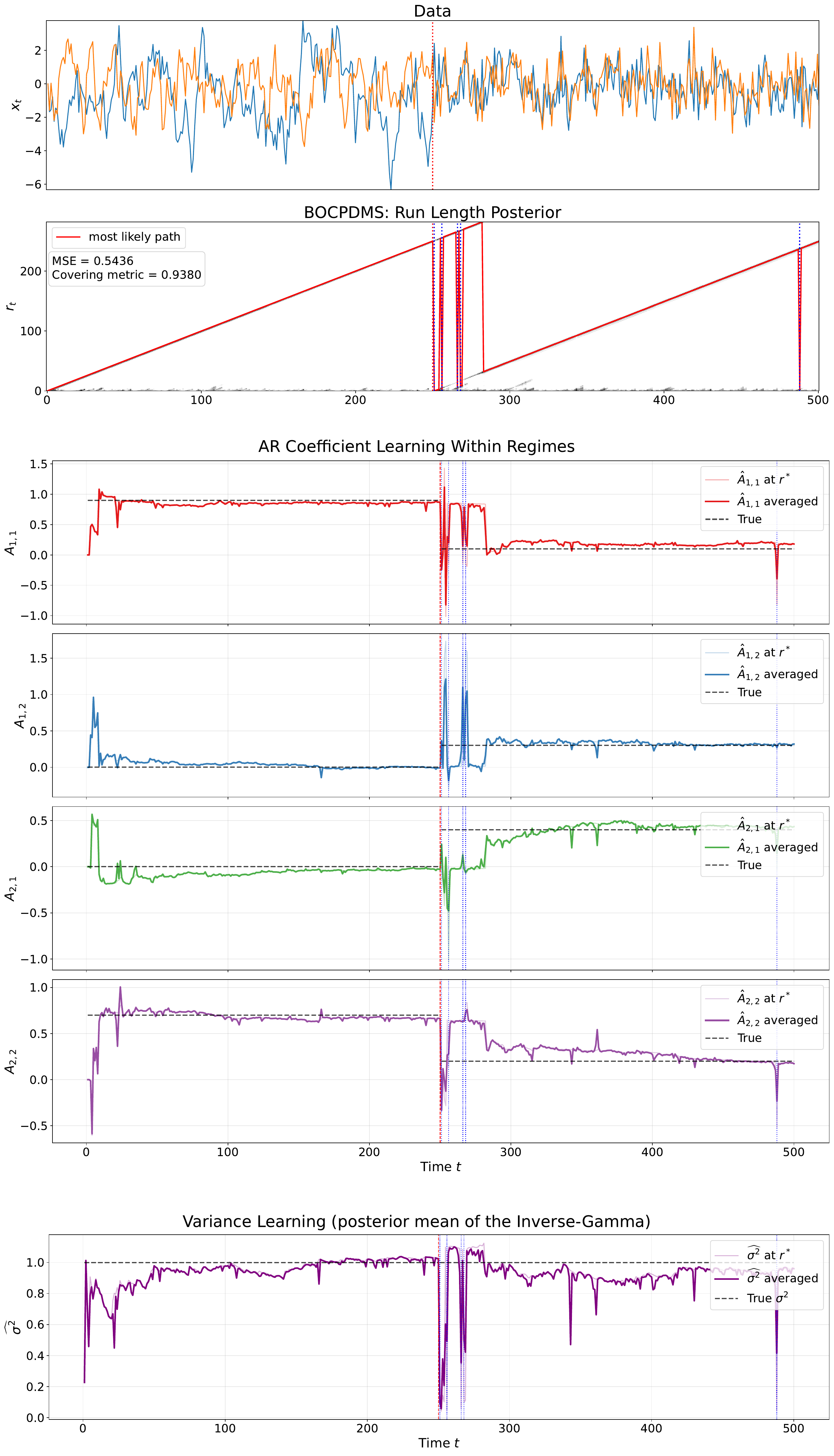}
        \caption{Exp. 2 --- AR matrix, cross-dependencies}
    \end{subfigure}\hfill
    \begin{subfigure}[b]{0.48\linewidth}
        \includegraphics[width=\linewidth]{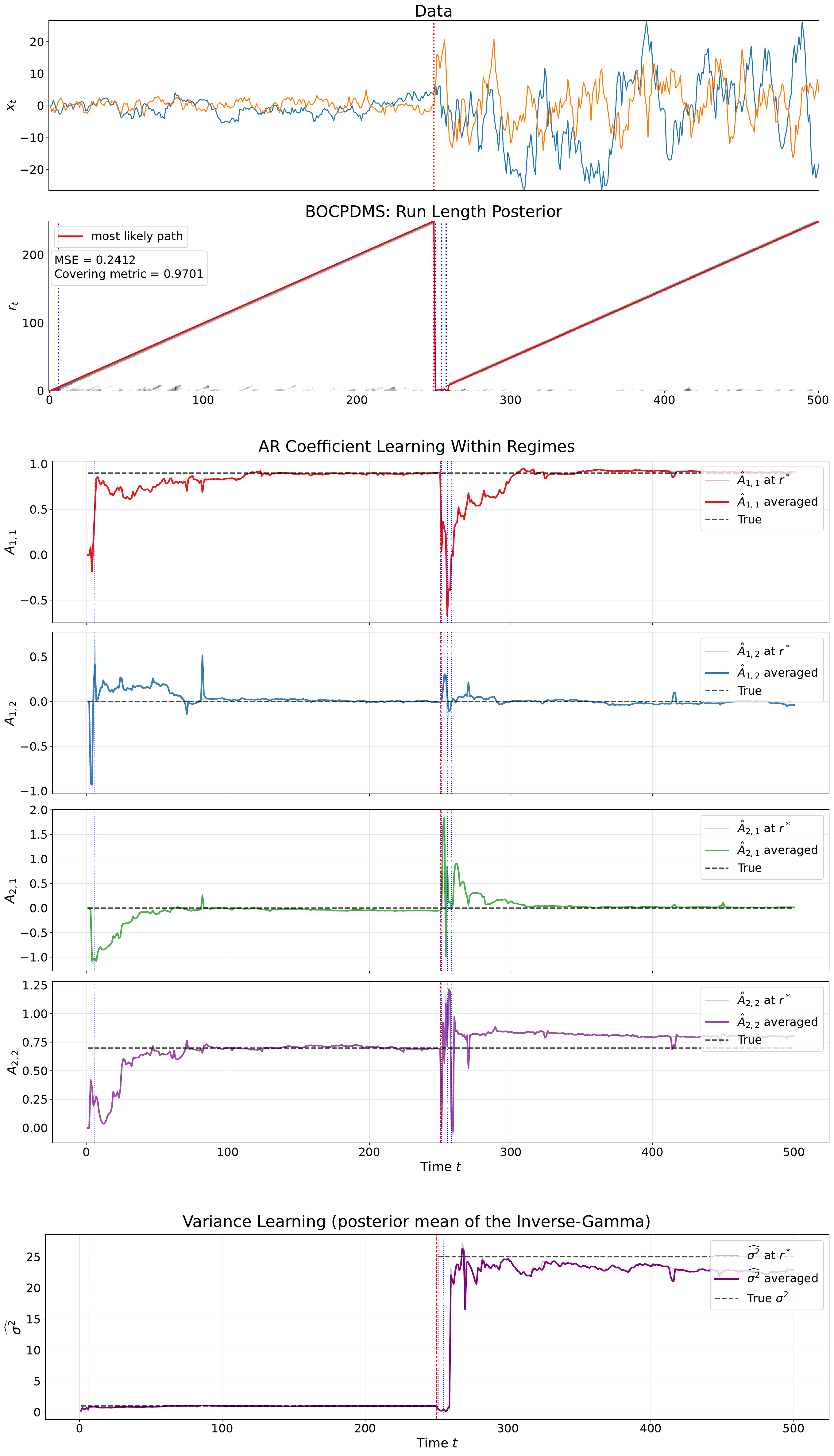}
        \caption{Exp. 3 --- noise variance}
    \end{subfigure}
    \caption{Experiments 2--3: a single changepoint at $t^\star = 250$, affecting the AR matrix with cross terms, then the innovation variance.}
    \label{fig:exp2_3}
\end{figure}

\newpage
\subsection{One change in the mean of the noise}

Same bivariate VAR(1), $\mathbf{A} = \diag(0.9, 0.7)$ held constant, and the innovation mean changes:
$$\boldsymbol{\varepsilon}_t \sim \mathcal{N}(\boldsymbol{\mu}_1, \sigma^2\mathbf{I}_2) \;\longrightarrow\; \mathcal{N}(\boldsymbol{\mu}_2, \sigma^2\mathbf{I}_2), \qquad \boldsymbol{\mu}_1 = (1.0, 1.0),\ \boldsymbol{\mu}_2 = (5.0, 5.0), \ \sigma = 1$$
See Figure~\ref{fig:exp4_5}(a).

\subsection{Multiple changepoints in the AR matrix}

We generate a bivariate VAR(1) process with $\boldsymbol{\varepsilon}_t \sim \mathcal{N}(\mathbf{0},\,\mathbf{I}_2)$. The changepoints follow a constant hazard $H = 1/h$, $h = 150$. At every changepoint a fresh AR matrix $\mathbf{A}$ is drawn with i.i.d.\ standard-Gaussian entries and then rescaled so that its spectral radius equals $0.9$, which guarantees the process stays stationary. Only $\mathbf{A}$ changes across regimes. The model is correctly specified in its hazard, of course. See Figure~\ref{fig:exp4_5}(b).

\begin{figure}[H]
    \centering
    \begin{subfigure}[b]{0.48\linewidth}
        \includegraphics[width=\linewidth]{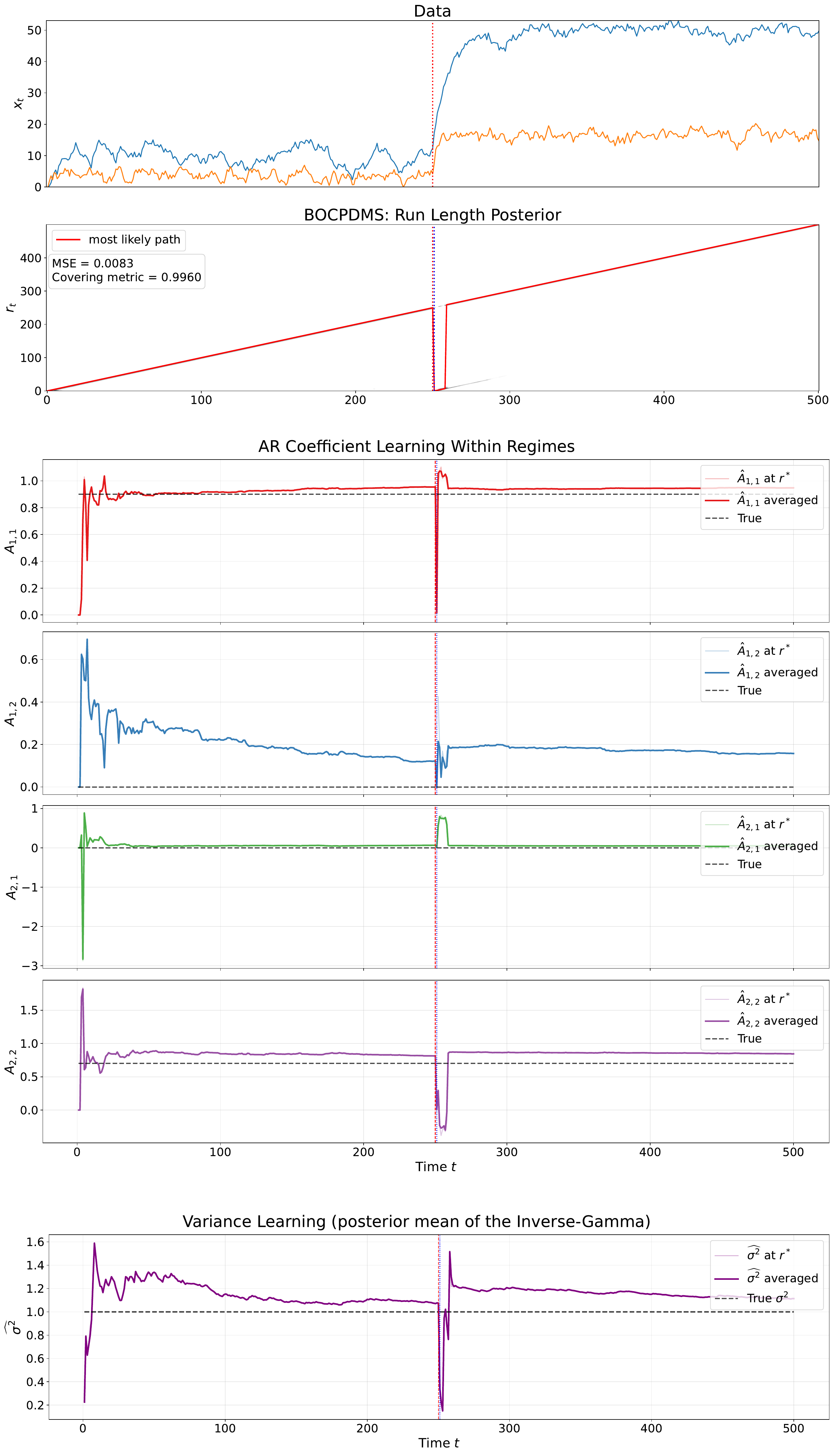}
        \caption{Exp. 4 --- noise mean, single changepoint}
    \end{subfigure}\hfill
    \begin{subfigure}[b]{0.48\linewidth}
        \includegraphics[width=\linewidth]{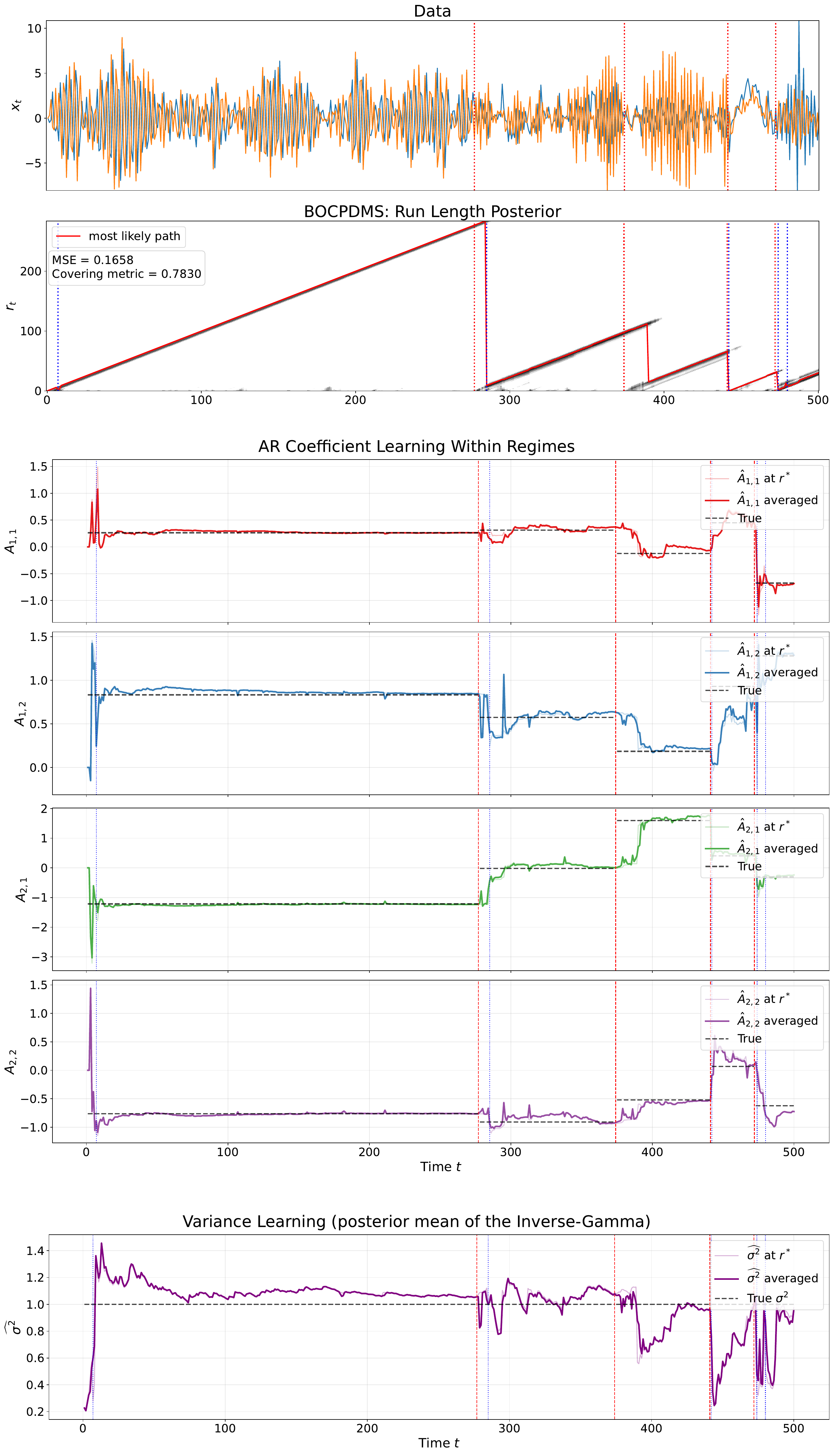}
        \caption{Exp. 5 --- AR matrix, multiple changepoints}
    \end{subfigure}
    \caption{Experiments 4--5: a single changepoint in the innovation mean, then multiple changepoints drawn with constant hazard $h=150$ in the AR matrix.}
    \label{fig:exp4_5}
\end{figure}

\newpage
\subsection{Multiple changepoints in the variance of the noise}

Identical generative setup to Experiment~5, but the breaks affect the innovation variance only, the AR matrix being held constant $\mathbf{A}=\diag(0.9,\,0.7)$. At every changepoint a new scalar standard deviation $\sigma$ is drawn from a Gaussian prior truncated to the positive half-line, and the innovation is $\boldsymbol{\varepsilon}_t \sim \mathcal{N}(\mathbf{0},\,\sigma^2\,\mathbf{I}_2)$; the change is thus isotropic across the two components. See Figure~\ref{fig:exp6_7}(a).

\subsection{Multiple changepoints in the mean of the noise}

Identical generative setup to Experiment~5, but the breaks affect the innovation mean only; the AR matrix $\mathbf{A}=\diag(0.9,\,0.7)$ and the noise level $\sigma = 1$ are held constant. At every changepoint a new mean vector $\boldsymbol{\mu} \in \mathbb{R}^2$ is drawn, each component i.i.d. from a Gaussian distribution, and $\boldsymbol{\varepsilon}_t \sim \mathcal{N}(\boldsymbol{\mu},\,\sigma^2\mathbf{I}_2)$.

\textit{NB:} the induced shift in the level of $\mathbf{x}_t$ is amplified by a per-component gain $1/(1-a)$. See Figure~\ref{fig:exp6_7}(b).

\begin{figure}[H]
    \centering
    \begin{subfigure}[b]{0.48\linewidth}
        \includegraphics[width=\linewidth]{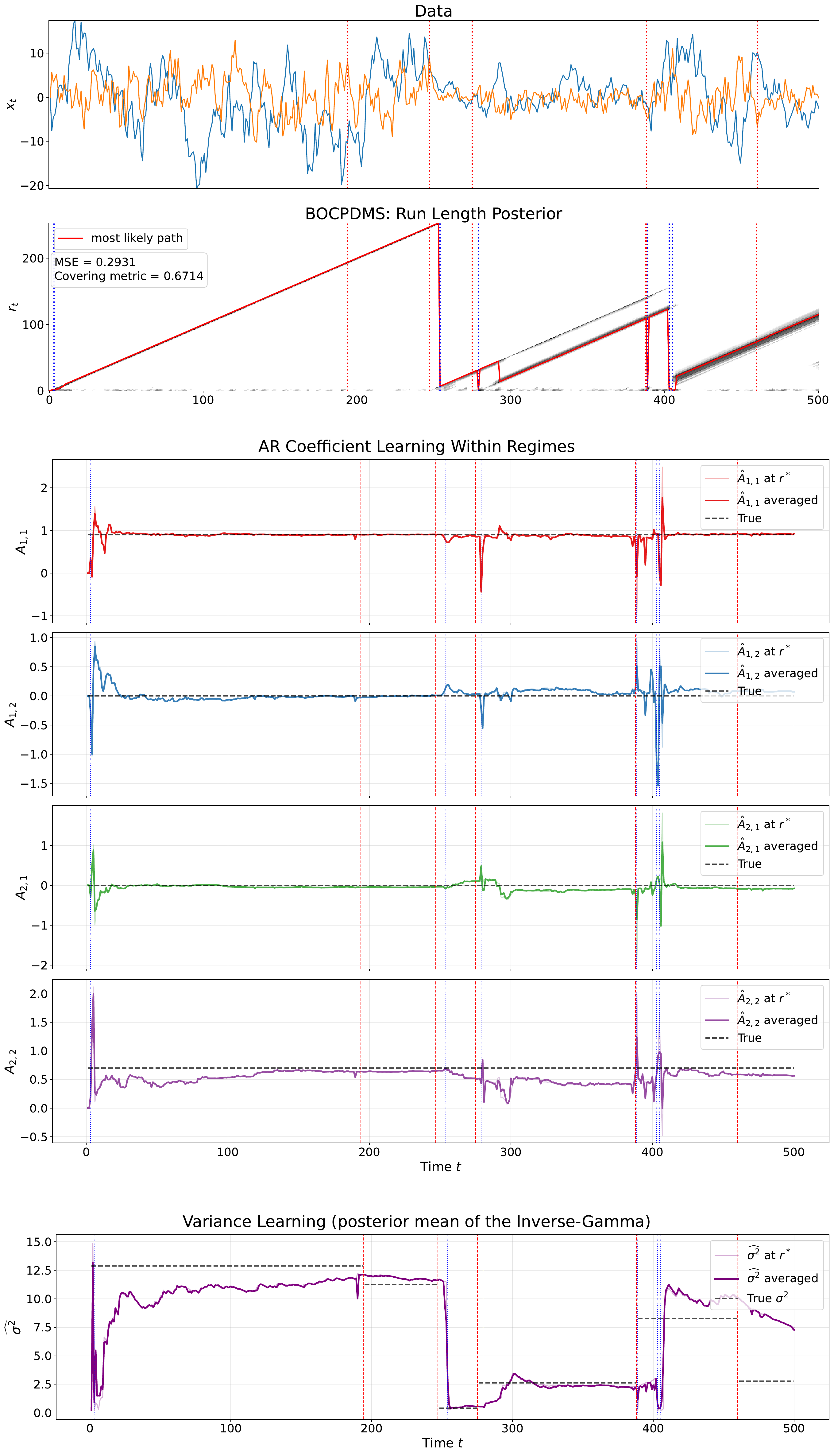}
        \caption{Exp. 6 --- noise variance}
    \end{subfigure}\hfill
    \begin{subfigure}[b]{0.48\linewidth}
        \includegraphics[width=\linewidth]{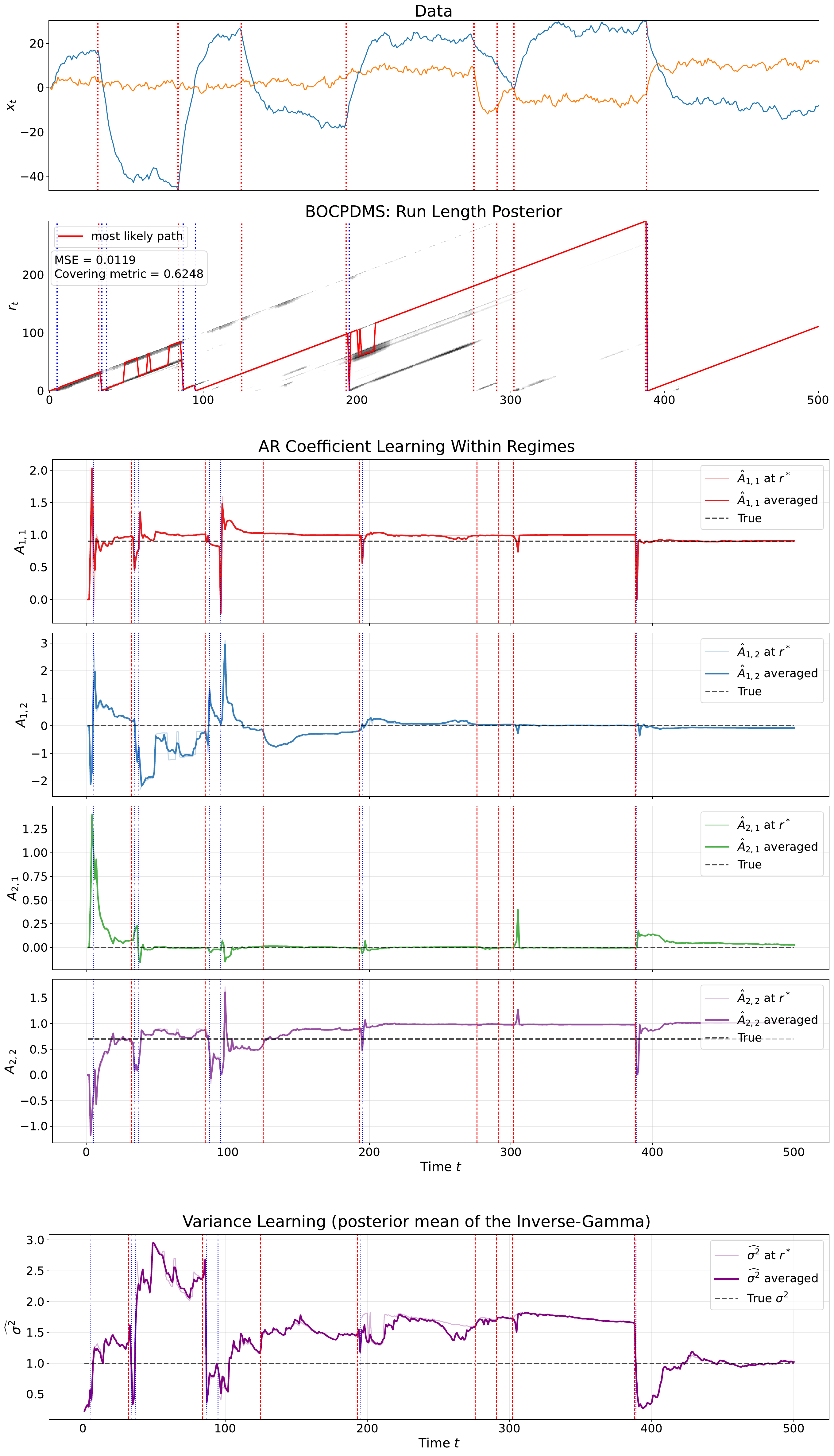}
        \caption{Exp. 7 --- noise mean}
    \end{subfigure}
    \caption{Experiments 6--7: multiple changepoints drawn with constant hazard $h=150$, affecting the innovation variance, then the innovation mean.}
    \label{fig:exp6_7}
\end{figure}

\newpage
\section{Experiments on real data}\label{sec:realmulti}

We now work with a bivariate series built from the same source and the same two stocks as in Section~\ref{sec:realdata}, but with a different sampling scheme. The multivariate model requires the two components to be observed at the same instants, which volume-clock sampling does not guarantee: buckets of $N$ executions of MSFT and of AAPL are not aligned in time. We therefore revert to calendar time and aggregate the signed volume over fixed $60$-second buckets, so that component $1$ and component $2$ of $\mathbf{x}_t = (x_t^{\text{MSFT}}, x_t^{\text{AAPL}})^T$ refer to the same minute of the same trading day. A trading day contains $390$ one-minute bars; over the month of June 2022 this gives $T = 8190$ observations, that is $21$ trading days.

\textbf{The two series are strongly dependent.} Their contemporaneous correlation is high, which is what makes a joint model worth considering at all: if the two flows were unrelated, a bivariate model could only lose against two separate univariate ones. Two further choices are fixed identically for every model, so that the comparison is fair:
\begin{itemize}[itemsep=-4pt]
    \item \textbf{No intercept.} Our BVAR has no intercept ($\boldsymbol{\alpha} = \mathbf{0}$), so to keep the comparison symmetric the frequentist baselines also run without one.
    \item \textbf{Scaling.} Same as in the experiments on the univariate baseline.
\end{itemize}

\subsection{The frequentist baseline: an expanding-window VAR(q)}

Some of our experiments compare BOCPDMS against a \textbf{VAR($q$) forecaster with no regimes}. This is a stripped-down baseline: no Bayesian prior, no changepoint detection, a single set of coefficients assumed valid over the whole sample.
A VAR($q$) writes the current observation as a linear combination of the $q$ most recent ones,
\begin{equation}
    \mathbf{x}_t = \mathbf{A}_1 \mathbf{x}_{t-1} + \cdots + \mathbf{A}_q \mathbf{x}_{t-q} + \boldsymbol{\varepsilon}_t = \mathbf{C}^T \boldsymbol{\phi}_t + \boldsymbol{\varepsilon}_t,
\end{equation}
with the same regressor vector $\boldsymbol{\phi}_t = (\mathbf{x}_{t-1}^T, \ldots, \mathbf{x}_{t-q}^T)^T \in \mathbb{R}^{k}$ ($k = qS$) and the same coefficient layout $\mathbf{C}^T = (\mathbf{A}_1, \ldots, \mathbf{A}_q)$ that we used for the BVAR. Each $\mathbf{A}_l$ is a $S \times S$ matrix, so the two stocks may influence each other across lags, the cross-dependence that Experiment~2 isolates.

\textbf{Estimation by OLS.} Given a window of data, we stack the predictors in a matrix $\mathbf{X}$ (row $i$ is $\boldsymbol{\phi}_i^T$), the targets in $\mathbf{Y}$ (row $i$ is $\mathbf{x}_i^T$), and fit the coefficients by ordinary least squares.

\textbf{The expanding window.} Rather than fitting once, we replay the position of a forecaster who, at each instant, only knows the past. At time $t$ we estimate $\hat{\mathbf{C}}$ on $\mathbf{x}_{1:(t-1)}$, forecast $\hat{\mathbf{x}}_t = \hat{\mathbf{C}}^T\boldsymbol{\phi}_t$, then observe the true $\mathbf{x}_t$ and record the error. We then move one step forward: the window \emph{grows} by one observation, and we repeat. Two properties follow:
\vspace{-4pt}
\begin{itemize}[itemsep=-2pt]
    \item \textbf{No look-ahead.} The forecast of $\mathbf{x}_t$ uses only strictly earlier data.
    \item \textbf{A refining estimate.} Early on, $\hat{\mathbf{C}}$ is fit on very few points and is noisy. We therefore skip the very first steps, where the estimate is meaningless, and start evaluating only at a time $t_{\min}$ large enough that $\mathbf{X}^T\mathbf{X}$ is invertible, i.e.\ $t_{\min} \geq q + k$.
\end{itemize}
\vspace{-9pt}

\subsection{Evaluation metrics}

All models are evaluated on their one-step-ahead forecasts over a \textbf{common window} $t \in \{t_{\min}, \ldots, T\}$, chosen so that every model is scored on exactly the same time points. We report the \emph{normalized} mean squared error
\begin{equation}\label{eq:mse}
    \mathrm{MSE} = \frac{\sum_c \overline{\mathrm{err}_c^2}}{\sum_c \mathrm{Var}(x_c)}, \qquad \mathrm{err}_{t,c} = x_{t,c} - \hat{\mu}_{t,c},
\end{equation}
where both the average and the variance are taken over the evaluation window, and the sums run over the $S$ components. We also report the \textbf{per-component} values $\overline{\mathrm{err}_c^2}/\mathrm{Var}(x_c)$.

\textbf{The normalizer is the realized variance on the evaluation window}. It is computed once and passed identically to every model.

\subsection{A scale problem}

Our first runs of the bivariate BOCPDMS produced an MSE \emph{well above} $1$ --- that is, worse than simply forecasting zero. So we compared the forecast to the target on two grounds. First, their correlation: at $\rho \approx 0.36$, it is nearly as high as that of the frequentist VAR ($\approx 0.42$), so the forecast is far from meaningless, it moves with the data. Second, its variance: the BOCPDMS forecast has a variance several times that of the data itself, whereas the frequentist VAR produces forecasts with about a fifth of the data variance. The forecast therefore points in roughly the right direction but is grossly over-sized, and it is this over-sizing, not a lack of predictive content, that produces the large error.

Two further observations confirm it. The largest BOCPDMS forecasts \textbf{exceed the largest value ever taken by the data}. The model predicts values the series never attains, which no reasonable forecast should do. And the top $1\%$ of squared errors carries a noticeably larger share of the total MSE than it does for the frequentist VAR on the same data.

This is because the signed data is heavy-tailed. Within a short regime, the autoregressive coefficient is estimated on few points and can be badly off; when it is then multiplied by an incoming spike, the product is a very big forecast. The frequentist VAR is better because its coefficient is a stable global fit. Here is a slightly counter-intuitive finding: \emph{the adaptivity of BOCPDMS, which is its strength on regime-switching data, becomes a liability when the innovations are heavy-tailed and the regimes are short.}

\textbf{Why the prior did not prevent it.} The coefficient prior is supposed to keep this under control. It failed because it acts on both series through a single scalar, $\mathbf{V}_c = v_c\mathbf{I}$, while on raw data the two series can have very different variances, here one carries roughly twice the variance of the other. The same $v_c$ therefore has very different effects on the two series, and no single value can suit both at once. This is what led us to put the two series on a common scale.

\vspace{-10pt}
\paragraph{Standardization.}
The remedy that follows is to put the two series on the same scale, so that the prior acts on them identically. We standardize each component separately,
\begin{equation}\label{eq:standardization}
    x'_{t,c} = \frac{x_{t,c}}{\sigma_c}, \qquad c = 1, \ldots, S,
\end{equation}

\vspace{-8pt}\noindent
with $\sigma_c$ the standard deviation of component $c$. This is an \textbf{invertible} transformation.

\textbf{Estimating $\sigma_c$ without looking ahead.} If the scale is computed on the whole sample, the model uses information from the future, which would flatter the results. We therefore estimate $\sigma_c$ on a \textbf{prefix sample} and apply them unchanged to the rest of the sample.

This raises a practical difficulty. Because the data are heavy-tailed, the sample standard deviation of a short prefix is itself an unstable quantity: a calm stretch can under-estimate the scale of the rest of the month by a factor of two. We observed exactly this with a short prefix. The remedy we adopt is simply to make the prefix \textbf{long enough to be representative}: we set $t_{\min}$ to several full trading days, estimate $\sigma_c$ on that stretch, and evaluate on the remainder.

\newpage
\subsection{Parameter choices}

All settings are fixed once and shared by the three experiments.
\begin{itemize}[itemsep=-4pt]
    \item \textbf{Prior $\mathbf{V}_c = 0.001\,\mathbf{I}$.} Selected by a sweep: it minimizes the error.
    \item \textbf{$a_0 = 2.0$, $b_0 = 0.5$, learned online}; their effect on this data is negligible.
    \item \textbf{Hazard $H=1/h$ constant, with $h=390/3=130$.}
    \item \textbf{$t_{\min} = 1170$ ($3$ trading days).} Burn-in of the most demanding baseline, and the window on which the scale and $\boldsymbol{\Omega}$ are estimated.
\end{itemize}

\noindent\textbf{The three specifications of $\boldsymbol{\Omega}$.} All are estimated on the prefix $\mathbf{x}_{1:t_{\min}}$ only. Writing $\mathbf{x}_i = (x_i^{\text{MSFT}}, x_i^{\text{AAPL}})^T$ for the prefix observations and $\bar{\mathbf{x}}$ for their mean:

\vspace{5pt}\textbf{Diagonal $\boldsymbol{\Omega}$ --- independent noise.} We keep the variance of each series and set the cross terms to zero.
On standardized data, it is $\mathbf{I}_2$, the series has unit variance on the prefix.

\vspace{5pt}\textbf{Full $\boldsymbol{\Omega}$ --- covariance of the observations.} We keep the cross term as well, i.e.\ the full empirical covariance matrix. This is the simplest way to take cross-correlation into consideration.

\vspace{-2pt}
\begin{adjustbox}{max width=\textwidth}
$\displaystyle
    \boldsymbol{\Omega}_{\mathrm{diag}} = \mathrm{diag}\Big(\mathrm{Var}(x^{\text{MSFT}}),\; \mathrm{Var}(x^{\text{AAPL}})\Big), \quad \mathrm{Var}(x^{c}) = \frac{1}{t_{\min}}\sum_{i=1}^{t_{\min}} \big(x_i^{c} - \bar{x}^{c}\big)^2,
    \qquad
    \boldsymbol{\Omega}_{\mathrm{full}} = \frac{1}{t_{\min}-1}\sum_{i=1}^{t_{\min}} (\mathbf{x}_i - \bar{\mathbf{x}})(\mathbf{x}_i - \bar{\mathbf{x}})^T .$
\end{adjustbox}

\vspace{5pt}\textbf{$\boldsymbol{\Omega}_{\mathrm{var}}$ --- covariance of the VAR innovations.} $\boldsymbol{\Omega}$ describes the covariance of the noise $\boldsymbol{\varepsilon}_t$, i.e.\ of what remains once the predictable part $\mathbf{A}\mathbf{x}_{t-1}$ has been removed, whereas the observations still contain that predictable part. So we estimate it from the residuals, and our frequentist baseline already produces them. We fit an expanding-window VAR on the prefix sample (the first $t_{\min} = 1170$ observations), and collect its one-step-ahead residuals $\hat{\boldsymbol{\varepsilon}}_t = \mathbf{x}_t - \hat{\mathbf{C}}_t^T\boldsymbol{\phi}_t$. Their empirical covariance gives $\boldsymbol{\Omega}_{\mathrm{var}} = \frac{1}{n}\sum_t \hat{\boldsymbol{\varepsilon}}_t \hat{\boldsymbol{\varepsilon}}_t^T$, where the sum runs over the retained residuals and $n$ is their number.

\vspace{5pt}\noindent\textit{NB:} we discard the first $50$ residuals. At the very start of the prefix the coefficients are fitted on a handful of points and are wildly off, so the residuals they produce are huge and would inflate the estimate. Without this precaution the estimated innovation variance of one series came out larger than its data variance.

\subsection{Results}

We report the MSE over the common window $[t_{\min}:T]$, with $t_{\min} = 1170$ and $T = 8190$.

Estimated on the sample $[\,0{:}1170]$: $\sigma = (6.33,\,11.09)$ with ratio AAPL/MSFT $= 1.75$, against $\sigma = (9.68,\,13.72)$ and ratio $1.42$ on the evaluation window, and $(9.28,\,13.37)$ with ratio $1.44$ over the whole month. The standardized series have variance $(2.34,\,1.53)$ on $[t_{\min}:T]$.

\textbf{Experiment 1 --- BOCPDMS versus expanding-window VAR($q$)} (Table~\ref{tab:var_q}).
\vspace{-6pt}
\begin{table}[H]
    \centering \small \small
    \caption{Expanding-window VAR($q$), $q = 1,\ldots,10$; best order $q^{\star} = 2$. For reference, BOCPDMS attains $0.8891$ with a VAR(1) and $\mathbf{0.8707}$ with a VAR(2). The BOCPDMS VAR(1) sits above the best frequentist VAR, while the BOCPDMS VAR(2) sits below it.}
    \label{tab:var_q}
    \begin{tabular}{|c|c||c|c||c|c||c|c||c|c|}
        \hline
        $q$ & MSE & $q$ & MSE & $q$ & MSE & $q$ & MSE & $q$ & MSE \\ \hline
        1 & 0.8815 & 3 & 0.9011 & 5 & 0.9179 & 7 & 0.9237 & 9  & 0.9222 \\ \hline
        2 & \textbf{0.8791} & 4 & 0.9112 & 6 & 0.9190 & 8 & 0.9208 & 10 & 0.9234 \\ \hline
    \end{tabular}
\end{table}

\newpage
\textbf{Experiments 2 and 3 --- univariate baseline and noise structure.} With the first $50$ residuals discarded, $\boldsymbol{\Omega}_{\mathrm{var}} = \begin{pmatrix} 0.820 & 0.481 \\ 0.481 & 0.754 \end{pmatrix}$. Table~\ref{tab:exp23} collects both experiments.
\vspace{-6pt}
\begin{table}[H]
    \centering \small \small
    \caption{Experiment 2 (top block): bivariate BOCPDMS against two independent univariate BOCPDs. Experiment 3 (bottom block): effect of the noise specification, with the best frequentist VAR as a reference. The two univariate BOCPDs detect $97$ (MSFT) and $84$ (AAPL) changepoints.}
    \label{tab:exp23}
    \begin{tabular}{|l|l|c|c|c|}
    \hline
     & \textbf{Model} & \textbf{MSE} & \textbf{MSFT} & \textbf{AAPL} \\ \hline
    \multicolumn{5}{|l|}{\textit{Experiment 2 --- multivariate vs. two univariate filters}}\\ \hline
    M4 & BOCPDMS VAR(2), bivariate & 0.8707 & 0.8582 & 0.8899 \\ \hline
    M1 & $2\times$ BOCPD, univariate, summed & \textbf{0.7705} & 0.7351 & 0.8247 \\ \hline
    \multicolumn{5}{|l|}{\textit{Experiment 3 --- contemporaneous noise structure}}\\ \hline
    M2 & BOCPDMS VAR(1), diagonal $\boldsymbol{\Omega}$ & 0.8891 & 0.8733 & 0.9134 \\ \hline
    M3 & BOCPDMS VAR(1), full $\boldsymbol{\Omega}$ (observations) & 0.8874 & 0.8727 & 0.9098 \\ \hline
    M4 & BOCPDMS VAR(2), diagonal $\boldsymbol{\Omega}$ & \textbf{0.8707} & 0.8582 & 0.8899 \\ \hline
    M5 & BOCPDMS VAR(2), full $\boldsymbol{\Omega}$ (observations) & 0.8916 & 0.8753 & 0.9164 \\ \hline
    M6 & BOCPDMS VAR(2), $\boldsymbol{\Omega} = \boldsymbol{\Omega}_{\mathrm{var}}$ (innovations) & 0.8877 & 0.8720 & 0.9118 \\ \hline
    M7 & VAR($q=2$) & 0.8791 & 0.8255 & 0.9611 \\ \hline
\end{tabular}
\end{table}

\textbf{Diebold--Mariano tests.} Pairwise comparisons under squared-error loss at horizon $h=1$ on $[t_{\min}:T]$ (Table~\ref{tab:dm}). Each cell gives the DM statistic and the $p$-value. A negative statistic favours the row model. Only the M1--M4 duel of Experiment~2 involves the univariate baseline; the remaining pairs are those of Experiment~3.

\begin{table}[H]
    \centering \small
    \caption{DM statistics (top) and $p$-values (bottom, smaller type), on each series.}
    \label{tab:dm}
    \begin{minipage}{0.48\textwidth}
        \centering
        \begin{tabular}{|c|c|c|c|c|c|c|}
            \hline
            \textbf{MSFT} & M2 & M3 & M4 & M5 & M6 & M7 \\ \hline
            M1 & & & $-0.78$ & & & \\
               & & & \scriptsize 0.434 & & & \\ \hline
            M2 & $+1.21$ & $+1.00$ & $-0.12$ & $+0.07$ & $+0.38$ & \\
               & \scriptsize 0.225 & \scriptsize 0.318 & \scriptsize 0.907 & \scriptsize 0.942 & \scriptsize 0.702 & \\ \hline
            M3 & & $+0.96$ & $-0.15$ & $+0.04$ & $+0.38$ & \\
               & & \scriptsize 0.336 & \scriptsize 0.884 & \scriptsize 0.966 & \scriptsize 0.705 & \\ \hline
            M4 & & & $-1.52$ & $-1.35$ & $+0.26$ & \\
               & & & \scriptsize 0.129 & \scriptsize 0.177 & \scriptsize 0.794 & \\ \hline
            M5 & & & & $+0.68$ & $+0.39$ & \\
               & & & & \scriptsize 0.494 & \scriptsize 0.694 & \\ \hline
            M6 & & & & & $+0.37$ & \\
               & & & & & \scriptsize 0.714 & \\ \hline
        \end{tabular}
    \end{minipage}\hfill
    \begin{minipage}{0.48\textwidth}
        \centering
        \begin{tabular}{|c|c|c|c|c|c|c|}
            \hline
            \textbf{AAPL} & M2 & M3 & M4 & M5 & M6 & M7 \\ \hline
            M1 & & & $-0.48$ & & & \\
               & & & \scriptsize 0.630 & & & \\ \hline
            M2 & $+0.94$ & $+1.26$ & $-0.14$ & $+0.07$ & $-0.48$ & \\
               & \scriptsize 0.347 & \scriptsize 0.206 & \scriptsize 0.888 & \scriptsize 0.946 & \scriptsize 0.634 & \\ \hline
            M3 & & $+1.15$ & $-0.34$ & $-0.09$ & $-0.52$ & \\
               & & \scriptsize 0.251 & \scriptsize 0.737 & \scriptsize 0.929 & \scriptsize 0.605 & \\ \hline
            M4 & & & $-1.43$ & $-1.20$ & $-0.72$ & \\
               & & & \scriptsize 0.154 & \scriptsize 0.229 & \scriptsize 0.473 & \\ \hline
            M5 & & & & $+0.78$ & $-0.48$ & \\
               & & & & \scriptsize 0.433 & \scriptsize 0.629 & \\ \hline
            M6 & & & & & $-0.52$ & \\
               & & & & & \scriptsize 0.602 & \\ \hline
        \end{tabular}
    \end{minipage}
\end{table}

None of the pairwise differences reaches conventional significance: on this sample, no model is demonstrably a better forecaster than another. The M1--M4 duel deserves a word, since a $12\%$ gap in MSE returning $p = 0.43$ is surprising at first sight. The loss differential is overwhelmingly carried by a handful of observations: the ten worst bars out of $7020$ account for about $70\%$ of $\sum|d_t|$. The test therefore measures the variance of those few points more than any systematic difference. Two robustness diagnostics support this reading: winsorizing the most extreme $0.5\%$ of $|d_t|$ makes the advantage of the univariate pair clearly significant ($\mathrm{DM} = 3.50$ on MSFT, $8.60$ on AAPL), and under absolute-error loss the same holds on AAPL ($\mathrm{DM} = 8.40$, $p < 10^{-4}$), MSFT remaining a draw ($0.32$, $p = 0.75$).

\newpage
\subsection{Implementation notes}
\label{sec:difficulties}

\hspace{1cm}\textbf{Computational cost.} The real datasets are far larger than the $500$-point synthetic series used in the controlled experiments, and each calibration run tests several thousand parameter combinations, so the grid searches for the optimal BOCPD and BOSD ($K=1$) parameters are infeasible in pure Python. We first parallelised the grid search over $7$ of the $8$ available cores, which divided the running time by roughly the same factor, but calibration runs still took hours. We therefore compiled the core filtering loops of both algorithms with \texttt{numba}, a just-in-time compiler that translates annotated Python functions into optimised machine code. \texttt{numba} supports \texttt{numpy} operations but not \texttt{scipy} ones, which the original loop used. Runs that previously took hours finished in about fifteen minutes on much larger grids. The same work proved even more necessary for BOCPDMS: on our bivariate real datasets ($T > 8000$), a single run takes minutes with the compiled kernel and is prohibitively long without it.

\textbf{The changepoint heuristic.} While working on BOSD ($K=1$) we could not explain some unsatisfactory results on synthetic data, where the covering metric is available. The main symptom was over-segmentation: far too many changepoints. Part of it comes from the Pareto hazard, whose sharp peak just after $d_{\min}$ (see Appendix~\ref{app:durations}) makes the filter hyper-vigilant at the start of a regime and inflates the chance of declaring a change (one of the motivations for introducing the log-normal specification, which has no such peak). But this did not remove the problem. Most of it turned out to be due to the detection heuristic. The original rule declares a changepoint whenever the most likely path falls below a fixed threshold. On synthetic data resembling the real series --- with $\sigma_0^2$ small relative to $\sigma^2$, so that regimes differ little and changepoints are masked by noise --- the most likely path hugs the horizontal axis, even though clear triangles are visible in the probability mass above it. With a threshold of $5$, a path dropping to $2$ and climbing back four times produces four spurious changepoints. We adopted a new heuristic that keeps the threshold but additionally requires an actual \emph{drop}, comparing the most likely run length at $t$ to its value at $t-1$. We experimented with other rules as well, some abandoning the most likely path entirely, aiming for something closer to what one reads by eye from the triangles and drops in the posterior mass. The difficulty is that the rule must remain \emph{online}. Methods to extract the segmentation from the full run-length posterior matrix exist, but they are offline. Since this was not the object of the present work, we kept the first satisfactory rule; it improved the results considerably.

\textbf{A pre-computed hazard vector.} Our implementation of the duration-aware algorithm pre-computes the hazard vector $H(r)$ before the filtering loop, whereas the reference \texttt{C++} codebase recomputes it at each step. The choice is equivalent for fixed duration parameters and strictly faster, but it rules out online re-estimation of the duration law and makes the memory footprint grow with $T$.

\textbf{An inconsistency in the reference implementation.} Validating the online learning of $\sigma^2$ on synthetic data with $\sigma^2 = 1$ gave a posterior mean converging to $2.0$. The cause was not our code but the shape convention of the reference implementation, $a_r = a_0 + (r+1)/2$, which is inconsistent with a scale accumulating $S$ residuals per step whenever $S>1$ (Section~\ref{sec:hyperlearning}).

\textbf{BOCPDMS.} We re-implemented the algorithm of \cite{knoblauch2018bocpdms} from the authors' reference code, but following the conventions of our own BOCPD codebase, and rewrote the plotting routines to display the online-learning outputs. The first runs on real data gave an MSE above $1$, which looked like an implementation error but was not: the standardization of the data fixed the symptom without altering the conclusion.

\newpage
\section{Conclusion}\label{sec:conclusion}

This report asked whether Bayesian Online Changepoint Detection captures the
regime structure of high-frequency order flow, and whether duration-awareness and
multivariate model selection improve on the baseline.

For the first extension the answer is positive. Replacing the constant hazard rate
by a run-length-dependent one, derived from an explicit duration law within a
hidden semi-Markov formulation, improves both segmentation and forecasting. The
log-normal specification dominates the geometric baseline and the Pareto
alternative on both assets and under both calibration criteria, and the Monte
Carlo validation confirms that each model recovers the duration law it assumes.
The financial motivation and the statistical result agree: order flow has no
characteristic timescale, and a duration law with a heavier tail than the
geometric describes it better. Reaching this conclusion also required work of a
different nature. The calibration grids were infeasible at the scale of the real
datasets until the core loops were parallelised and compiled, and the
over-segmentation observed initially turned out to originate in the detection
heuristic rather than in the duration law, which led us to redesign it.

For the second extension the answer is negative, and arguably more instructive.
The multivariate filter, equipped with Bayesian vector autoregressions and online
hyperparameter learning, is outperformed by two independent univariate filters
whatever the noise specification, and the Diebold--Mariano tests separate none of
the models on this sample. The adaptivity that makes the model effective on
regime-switching data becomes a liability on heavy-tailed innovations with short
regimes: an autoregressive coefficient estimated on a handful of observations,
multiplied by an incoming volume spike, produces forecasts larger than any value
the series ever attains. A globally fitted vector autoregression is immune
precisely because it cannot adapt. This result should nonetheless be read as
provisional, since only a small part of the framework has been exercised here: the
model universe was never populated with competing specifications, which is the
mechanism that motivates it.

Two methodological observations also emerge. A conjugate model must be verified on
synthetic data with known parameters before any implementation of it is trusted,
a precaution that revealed the inconsistency in the shape convention of the
reference implementation discussed above. And a scaling artefact must be
distinguished from a genuine modelling limitation, which here required examining
the correlation, the variance and the error concentration of the forecasts rather
than a single aggregate figure.

Several extensions follow directly. The most immediate is the addition of an
intercept to the Bayesian vector autoregression: prepending a constant to the
regressor vector changes neither the sufficient statistics, nor the predictive
distribution, nor the gradients, preserves conjugacy trivially, and makes the
multivariate model nest the univariate one that outperformed it, since vanishing
autoregressive coefficients leave a constant mean per regime. It requires only a
block-structured coefficient prior, so that the level is free to move while the
dynamics stay regularised, and our implementation already accepts one. Level-only,
autoregression-only and level-plus-autoregression specifications could then be
placed together in the model universe and left to compete, which is the use of
model selection that our single-model runs never exercised. A second extension is
the combination of the run-length-dependent hazard with the multivariate
framework, with which our implementation is already compatible. A third would be
to learn the hazard rate, or the parameters of the duration law, online by
gradient ascent, in the spirit of what is already implemented for the two
Inverse-Gamma hyperparameters; this would remove the last offline step from the
pipeline.

\clearpage
\bibliographystyle{plainnat}
\bibliography{biblio}

\clearpage
\thispagestyle{empty}
\vspace*{\fill}
\noindent {\Huge \textbf{Appendices}}
\vspace*{\fill}

\clearpage
\appendix

\section{Proofs of the results of Section~\ref{sec:litreview}}\label{app:proofs}

\textit{This appendix collects the derivations underlying the results stated in Section~\ref{sec:litreview}. We first establish the Gaussian conjugate update,
then prove the three propositions that rely on it: the Gaussian predictive
posterior, the recursive message-passing decomposition, and the closed-form
Underlying Predictive Model.}

\subsection{Gaussian conjugate update: unknown mean $\mu$, known variance $\sigma^2$}
\label{app:gauss_update}

Let data $\mathcal{D} = \{x_1, \dots, x_N\}$ be drawn independently from $\mathcal{N}(\mu,\sigma^2)$. Our aim is to infer the unknown mean $\mu$ sequentially. The likelihood function is given by:
\begin{equation}\label{eq:app_likelihood}
    p(\mathcal{D}|\mu) = \prod_{n=1}^N p(x_n|\mu) = \frac{1}{(2\pi\sigma^2)^{N/2}}\exp \Big\{-\frac{1}{2\sigma^2}\sum_{n=1}^N(x_n-\mu)^2\Big\}
\end{equation}

To exploit conjugacy, we choose a Gaussian prior for the mean:
\begin{equation*}
    p(\mu) = \mathcal{N}(\mu|\mu_0,\sigma^2_0) = \frac{1}{\sqrt{2\pi\sigma_0^2}}\exp\Big\{-\frac{1}{2\sigma_0^2}(\mu-\mu_0)^2\Big\}
\end{equation*}

Following Bayes' rule, the posterior is the product of the likelihood and the prior: $p(\mu|\mathcal{D})\propto p(\mathcal{D}|\mu)p(\mu)$. We can expand the terms in the exponent and drop any terms that do not depend on $\mu$, absorbing them into the proportionality constant (details in \cite{murphy2007conjugate}):
\begin{align}
    p(\mu|\mathcal{D}) & \propto \exp \Big\{-\frac{1}{2\sigma^2}\sum_{n=1}^N(x_n-\mu)^2-\frac{1}{2\sigma_0^2}(\mu-\mu_0)^2\Big\} \notag \\
    & = \exp\Big\{-\frac{1}{2\sigma^2}\Big(\sum_{n=1}^Nx_n^2+N\mu^2-2\mu\sum_{n=1}^Nx_n\Big)-\frac{1}{2\sigma_0^2}\Big(\mu^2+\mu_0^2-2\mu\mu_0\Big)\Big\} \notag \\
    & \propto \exp\Big\{-\frac{1}{2}\Big(\frac{N\mu^2}{\sigma^2}-\frac{2\mu N\mu_{ML}}{\sigma^2}+\frac{\mu^2}{\sigma_0^2}-\frac{2\mu\mu_0}{\sigma_0^2}\Big)\Big\} \notag \\
    & = \exp\Big\{-\frac{1}{2}\Big(\mu^2\underbrace{\Big(\frac{N}{\sigma^2}+\frac{1}{\sigma_0^2}\Big)}_{a}-2\mu\underbrace{\Big(\frac{N\mu_{ML}}{\sigma^2}+\frac{\mu_0}{\sigma_0^2}\Big)}_{b}\Big)\Big\}\label{eq:app_last_term}
\end{align}
where $\mu_{ML} = \frac{1}{N}\sum_{n=1}^N x_n$ is the empirical mean (which is the Maximum Likelihood solution).

To recognize this as a Gaussian distribution $\mathcal{N}(\mu|\mu_N,\sigma_N^2)$, the exponent must be of the form $-\frac{1}{2\sigma_N^2}(\mu-\mu_N)^2$. We achieve this by "completing the square" and, ignoring the constant term, the posterior becomes: $ p(\mu|\mathcal{D}) \propto \exp \Big\{-\frac{a}{2}\Big(\mu-\frac{b}{a}\Big)^2\Big\} $

Identifying the terms yields the formal update rules stated in Proposition~\ref{prop:gauss_conj_update}: the posterior remains Gaussian, $p(\mu|\mathcal{D}) = \mathcal{N}(\mu|\mu_N,\sigma_N^2)$, with
$$ \frac{1}{\sigma_N^2} = \frac{N}{\sigma^2}+\frac{1}{\sigma_0^2} \hspace{1cm}\Big|\hspace{1cm} \mu_N = \frac{N\sigma_0^2}{N\sigma_0^2+\sigma^2}\mu_{ML}+\frac{\sigma^2}{N\sigma_0^2+\sigma^2}\mu_0 $$

\subsection{Gaussian Predictive Posterior}\label{app:pred_post}

\begin{proposition*}[Gaussian Predictive Posterior]
Given the posterior distribution of the mean $\mathcal{N}(\mu_N, \sigma_N^2)$ and the known observation variance $\sigma^2$, the posterior predictive distribution is:
$$ p(x_{N+1}|\mathcal{D}) = \int \mathcal{N}(x_{N+1}|\mu,\sigma^2)\mathcal{N}(\mu|\mu_N,\sigma_N^2)d\mu \ = \ \pmb{\mathcal{N}(\mu_N,\sigma_N^2+\sigma^2)} $$
\end{proposition*}

\begin{proof}
To evaluate this predictive distribution, we must compute the integral of the product of two Gaussians. Dropping the normalizing constants that do not depend on $\mu$ or $x_{N+1}$, we have:
\begin{align*}
    p(x_{N+1}|\mathcal{D}) & \propto \int \exp\Big\{-\frac{1}{2\sigma^2}(x_{N+1}-\mu)^2\Big\} \exp\Big\{-\frac{1}{2\sigma_N^2}(\mu-\mu_N)^2\Big\} d\mu
\end{align*}
We expand the terms in the exponent and group them by powers of $\mu$:

\begin{adjustbox}{max width=\linewidth, center}
$\displaystyle
    -\frac{1}{2} \left[ \frac{x_{N+1}^2 - 2x_{N+1}\mu + \mu^2}{\sigma^2} + \frac{\mu^2 - 2\mu\mu_N + \mu_N^2}{\sigma_N^2} \right] = -\frac{1}{2} \left[ \mu^2 \underbrace{\left(\frac{1}{\sigma^2} + \frac{1}{\sigma_N^2}\right)}_{A} - 2\mu \underbrace{\left(\frac{x_{N+1}}{\sigma^2} + \frac{\mu_N}{\sigma_N^2}\right)}_{B} + \underbrace{\left(\frac{x_{N+1}^2}{\sigma^2} + \frac{\mu_N^2}{\sigma_N^2}\right)}_{C} \right]
$
\end{adjustbox}

By completing the square for $\mu$, the expression in the brackets becomes $A\big(\mu - \frac{B}{A}\big)^2 + \big(C - \frac{B^2}{A}\big)$.

The integral over $\mu$ of the first term $\exp\big\{-\frac{A}{2}\big(\mu - \frac{B}{A}\big)^2\big\}$ is a standard Gaussian integral which evaluates to a constant ($\sqrt{2\pi/A}$) that does not depend on $x_{N+1}$. The remaining term, which does depend on $x_{N+1}$, is $\exp\big\{-\frac{1}{2}\big(C - \frac{B^2}{A}\big)\big\}$. Let us algebraically simplify $C - \frac{B^2}{A}$:
\begin{align*}
    C - \frac{B^2}{A} &= \left(\frac{x_{N+1}^2}{\sigma^2} + \frac{\mu_N^2}{\sigma_N^2}\right) - \frac{\left(\frac{x_{N+1}}{\sigma^2} + \frac{\mu_N}{\sigma_N^2}\right)^2}{\frac{1}{\sigma^2} + \frac{1}{\sigma_N^2}}
    = \frac{x_{N+1}^2 \sigma_N^2 + \mu_N^2 \sigma^2}{\sigma^2 \sigma_N^2} - \frac{(\sigma_N^2 x_{N+1} + \sigma^2 \mu_N)^2}{\sigma^2 \sigma_N^2 (\sigma^2 + \sigma_N^2)} \\
    &= \frac{(x_{N+1}^2 \sigma_N^2 + \mu_N^2 \sigma^2)(\sigma^2 + \sigma_N^2) - (\sigma_N^4 x_{N+1}^2 + \sigma^4 \mu_N^2 + 2\sigma^2 \sigma_N^2 x_{N+1} \mu_N)}{\sigma^2 \sigma_N^2 (\sigma^2 + \sigma_N^2)} \\
    &= \frac{\sigma^2 \sigma_N^2 (x_{N+1}^2 + \mu_N^2 - 2x_{N+1}\mu_N)}{\sigma^2 \sigma_N^2 (\sigma^2 + \sigma_N^2)}
    = \pmb{\frac{(x_{N+1} - \mu_N)^2}{\sigma^2 + \sigma_N^2}}
\end{align*}

This is exactly the kernel of a Gaussian distribution with mean $\mu_N$ and variance $\sigma^2 + \sigma_N^2$:
\begin{equation*}
    p(x_{N+1}|\mathcal{D}) \propto \exp\Big\{-\frac{1}{2(\sigma^2 + \sigma_N^2)}(x_{N+1}-\mu_N)^2\Big\}
\end{equation*}
\end{proof}

\subsection{Recursive Message Passing}

\begin{proposition*}[Recursive Message Passing]
By applying Bayes' chain rule and the conditional independence of the run length, the joint distribution decomposes recursively into three distinct computational blocks:
\begin{align*}
    P(r_t, x_{1:t}) &= \sum_{r_{t-1}} P(r_t, x_t | r_{t-1}, x_{1:t-1}) P(r_{t-1}, x_{1:t-1}) \\
    &= \sum_{r_{t-1}} \underbrace{P(r_t | r_{t-1})}_{\text{Hazard Rate}} \underbrace{P(x_t | r_{t-1}, x_{t-1}^{(r)})}_{\text{UPM}} \underbrace{P(r_{t-1}, x_{1:t-1})}_{\text{Previous Message}}
\end{align*}
\end{proposition*}

\begin{proof}
Starting with the term inside the sum, we can separate the newly arriving variables ($r_t, x_t$) from the past variables ($r_{t-1}, x_{1:t-1}$) using Bayes' chain rule:
$$ p(r_t, x_{1:t}, r_{t-1}) = p(r_t, x_t, r_{t-1}, x_{1:t-1}) = p(r_t, x_t | r_{t-1}, x_{1:t-1}) \ \pmb{p(r_{t-1}, x_{1:t-1})} $$

\vspace{-10pt}
The second term here is exactly the \textbf{message} passed from the previous time step. We apply the chain rule to the first term:
$$p(r_t, x_t \mid r_{t-1}, x_{1:t-1}) = p(x_t \mid r_t, r_{t-1}, x_{1:t-1}) \times p(r_t \mid r_{t-1}, x_{1:t-1})$$

\vspace{-10pt}
Now we apply our two conditional independence assumptions.
\begin{itemize}[itemsep = -2pt]
    \item \textbf{Underlying Predictive Model (UPM):} The new observation $x_t$ only depends on the current run length $r_t$ and the data within this run, so $p(x_t \mid r_t, r_{t-1}, x_{1:t-1}) = \pmb{p(x_t \mid r_t, x_{t-1}^{(r)})}$
    \item \textbf{Hazard Rate:} The new run length $r_t$ depends only on the previous run length $r_{t-1}$, so $p(r_t \mid r_{t-1}, x_{1:t-1}) = \pmb{p(r_t \mid r_{t-1})}$.
\end{itemize}
\vspace{-16pt}
\end{proof}

\subsection{Closed-form UPM for the Gaussian case}

\begin{proposition*}[Closed-form UPM for the Gaussian Case]
If the data within a regime is modeled as a Gaussian with an unknown mean and a known variance $\sigma^2$, and we use a conjugate Gaussian prior, the UPM evaluates to a closed-form Gaussian distribution:
$$p(x_t \mid r_{t-1}, x_{t-1}^{(r)}) = \pmb{\mathcal{N}(x_t \mid \mu_{r_{t-1}}, \sigma^2 + \sigma^2_{r_{t-1}})}$$
where $\mu_{r_{t-1}}$ and $\sigma^2_{r_{t-1}}$ are the updated sufficient statistics using strictly the $r_{t-1}$ observations of the current regime, using the Gaussian conjugate update of Appendix~\ref{app:gauss_update}.
\end{proposition*}

\begin{proof}
To evaluate the UPM, we marginalize over the unknown mean of the current regime:
$$p(x_t \mid r_{t-1}, x_{t-1}^{(r)}) = \int p(x_t \mid \mu_{\rho}) p(\mu_{\rho} \mid x_{t-1}^{(r)}) d\mu_{\rho}$$

\vspace{-10pt}
This is identical to the integral solved in Appendix~\ref{app:pred_post}. Because we restricted our model to the exponential family with a conjugate prior, the posterior distribution of the mean $p(\mu_{\rho} \mid x_{t-1}^{(r)})$ is natively a Gaussian $\mathcal{N}(\mu_{r_{t-1}}, \sigma^2_{r_{t-1}})$. The number of valid observations is simply the run length $N = r_{t-1}$. This yields exactly the distribution $\mathcal{N}(x_t \mid \mu_{r_{t-1}}, \sigma^2 + \sigma^2_{r_{t-1}})$.
\end{proof}

\newpage
\section{Regime duration distributions for BOSD ($K=1$)}\label{app:durations}

In the HSMM framework, the distribution of regime durations must be explicitly defined. Because the inference algorithm operates in discrete time steps, we must map probability distributions to a discrete hazard rate. For a duration random variable $d$:
$$ H(r) = P(d = r+1 \mid d > r) = \frac{S(r) - S(r+1)}{S(r)} = 1 - \frac{S(r+1)}{S(r)} $$

\noindent We formally define four distributions and derive their corresponding discrete hazard functions.

\subsection{The Gaussian Model}

\begin{definition}[Normal Distribution]
A random variable $d \sim \mathcal{N}(\mu_d, \sigma_d^2)$ has density $f(x)$ and survival function $S(x)$ given by:
\begin{equation}
    f(x) = \frac{1}{\sigma_d \sqrt{2\pi}} \exp\!\left(-\frac{(x - \mu_d)^2}{2\sigma_d^2}\right), \qquad S(x) = \frac{1}{2} \text{erfc}\left( \frac{x - \mu_d}{\sigma_d \sqrt{2}} \right)
\end{equation}
where $\text{erfc}$ is the complementary error function.
\end{definition}

\vspace{-20pt}
\begin{proposition}[Gaussian Hazard Function]
\begin{equation}
    H_{\text{Gauss}}(r) = 1 - \frac{\text{erfc}\left( (r + 1 - \mu_d)/\sigma_d \sqrt{2} \right)}{\text{erfc}\left( (r - \mu_d)/{\sigma_d \sqrt{2}} \right)}
\end{equation}
\end{proposition}

\subsection{The Pareto Model}

\begin{definition}[Pareto Distribution]
Let $d \sim \mathrm{Pareto}(\alpha, d_{\min})$ with shape parameter $\alpha > 0$ and scale parameter $d_{\min} \geq 1$:
\begin{equation}
    f(x) = \frac{\alpha\, d_{\min}^\alpha}{x^{\alpha+1}}, \qquad S(x) = \left(\frac{d_{\min}}{x}\right)^{\alpha}, \quad x \geq d_{\min}
\end{equation}
\end{definition}

\vspace{-20pt}
\begin{proposition}[Pareto Hazard Function]
\begin{equation}
    H_{\text{Pareto}}(r) = 1 - \frac{S(r+1)}{S(r)} = 1 - \left(\frac{r}{r+1}\right)^\alpha
\end{equation}
\end{proposition}

\subsection{The Log-Normal Model}

The Log-Normal distribution provides a middle ground: it has a characteristic peak (mode) but exhibits a heavier right tail than the Gaussian distribution.

\begin{definition}[Log-Normal Distribution]
Let $d \sim \mathrm{LogNormal}(\mu_\ell, s^2)$.
\begin{equation}
    f(x) = \frac{1}{x\, s\sqrt{2\pi}} \exp\!\left(-\frac{(\ln x - \mu_\ell)^2}{2s^2}\right), \qquad S(x) = \frac{1}{2} \text{erfc}\left( \frac{\ln x - \mu_\ell}{s \sqrt{2}} \right)
\end{equation}
\end{definition}

\vspace{-20pt}
\begin{proposition}[Log-Normal Hazard Function]
\begin{equation}
    H_{\text{LogNorm}}(r) = 1 - \frac{\text{erfc}\left( (\ln(r + 1) - \mu_\ell)/{s \sqrt{2}} \right)}{\text{erfc}\left( (\ln(r) - \mu_\ell)/{s \sqrt{2}} \right)}
\end{equation}
\end{proposition}

\subsection{The Poisson Model}

\begin{definition}[Poisson Distribution]
$d \sim \mathrm{Poisson}(\lambda)$ has probability mass function and survival function:
\begin{equation}
    P(d = k) = \frac{\lambda^k e^{-\lambda}}{k!}, \qquad S(k) = P(d > k) = \sum_{j=k+1}^{\infty} \frac{\lambda^j e^{-\lambda}}{j!}
\end{equation}
\end{definition}

\vspace{-30pt}
\begin{proposition}[Poisson Hazard Function]
Because the Poisson distribution is natively discrete, its hazard function can be evaluated directly from its PMF and SF: $H_{\text{Poisson}}(r) = P(d = r+1)/S(r)$.
\end{proposition}

\subsection{Summary}

A wrong approach to computing the discrete hazard rate from a continuous distribution is to approximate the numerator $P(d = r+1)$ using the pointwise continuous probability density function. The exact discrete probability mass must be computed by integrating the continuous density over the interval. Figure~\ref{fig:hazards} illustrates this discrepancy by plotting both the biased approximation (PDF-based) and the exact formulation.

\begin{figure}[H]
    \centering
    \begin{subfigure}[b]{0.49\linewidth}
        \includegraphics[width=\linewidth]{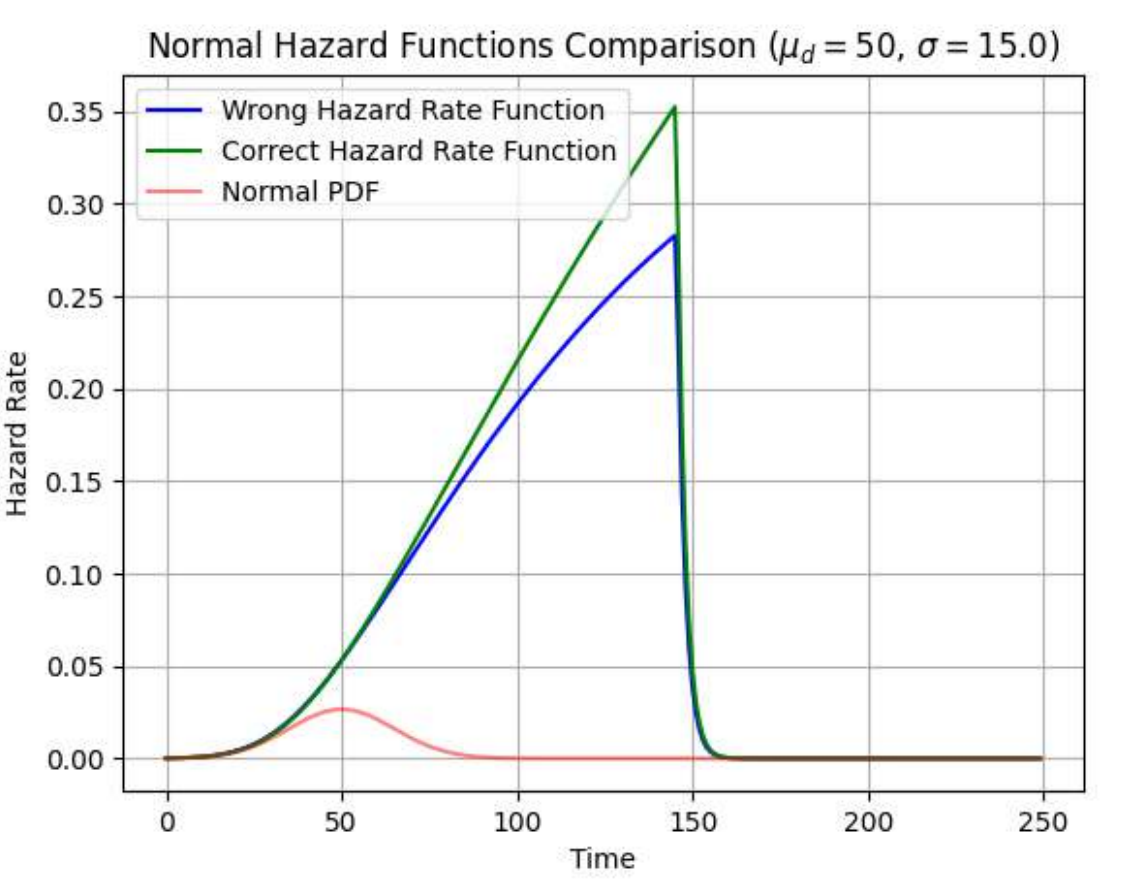}
        \caption{Normal}
    \end{subfigure}\hfill
    \begin{subfigure}[b]{0.49\linewidth}
        \includegraphics[width=\linewidth]{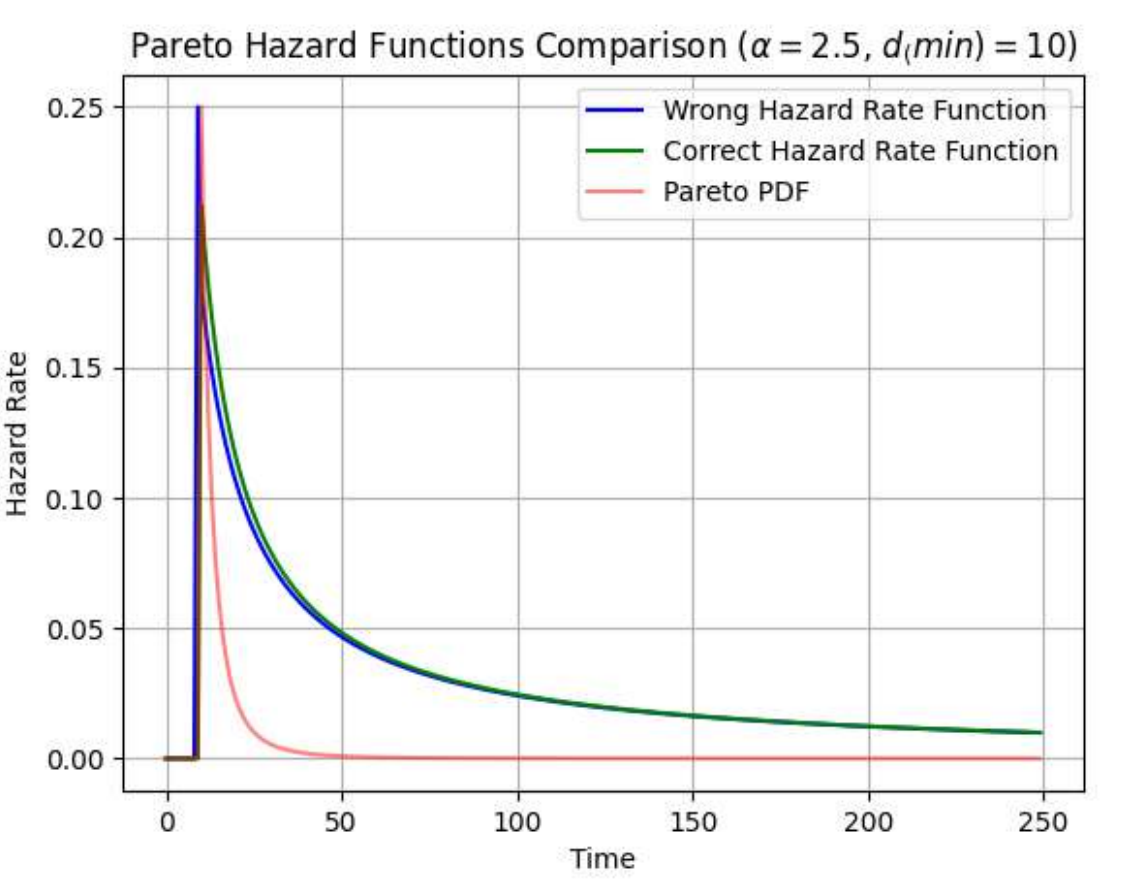}
        \caption{Pareto}
    \end{subfigure}

    \vspace{3pt}
    \begin{subfigure}[b]{0.49\linewidth}
        \includegraphics[width=\linewidth]{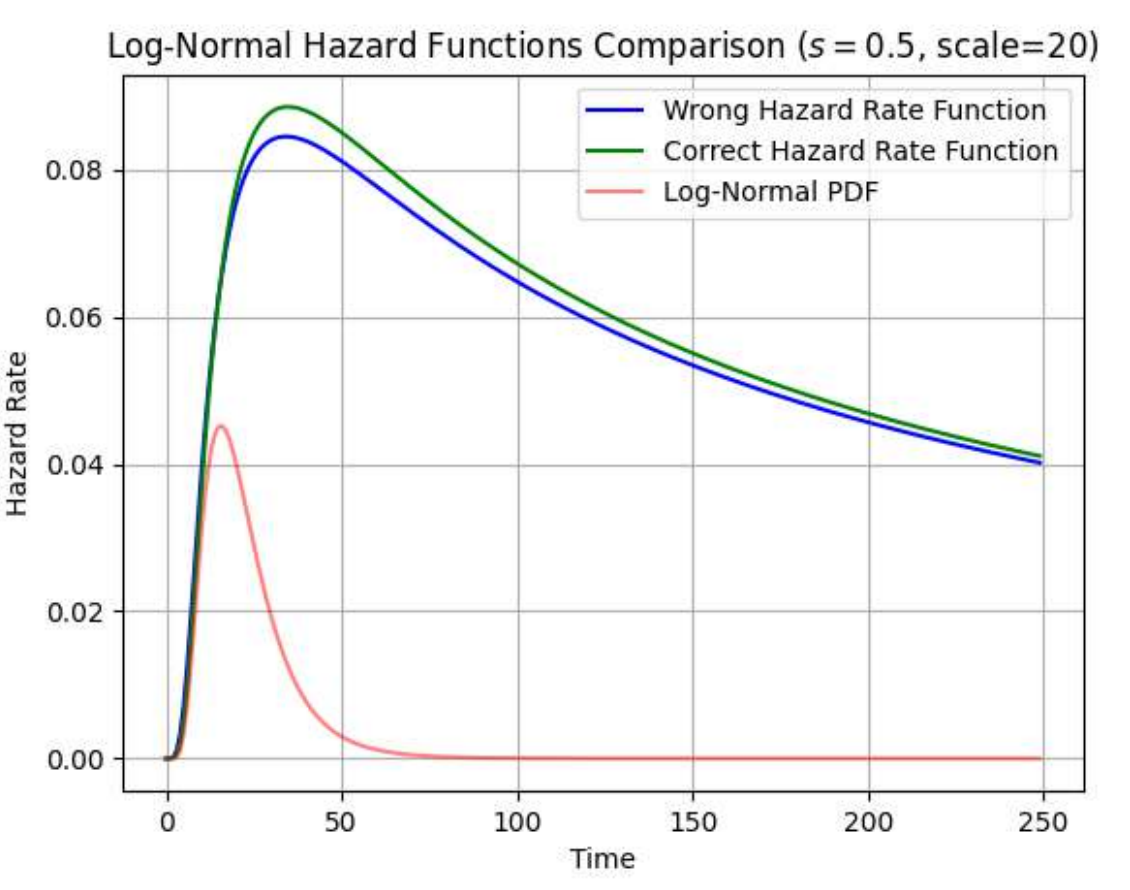}
        \caption{Log-Normal}
    \end{subfigure}\hfill
    \begin{subfigure}[b]{0.49\linewidth}
        \includegraphics[width=\linewidth]{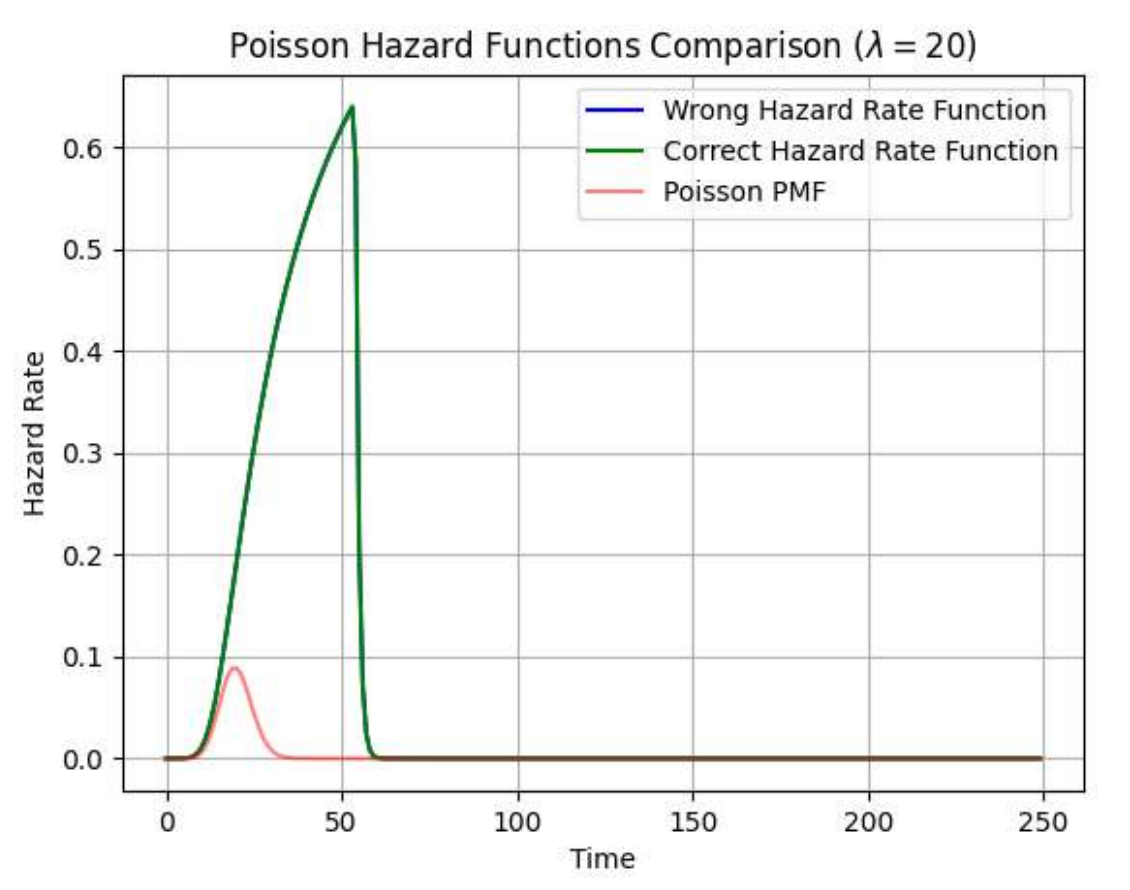}
        \caption{Poisson}
    \end{subfigure}
    \caption{Discrete hazard rate functions $H(r)$ derived from the four duration distributions.}
    \label{fig:hazards}
\end{figure}

\section{Outputs of the univariate out-of-sample runs}
\label{app:rlpost}

\begin{figure}[H]
    \centering
    \begin{subfigure}[b]{0.48\linewidth}
        \includegraphics[width=\linewidth]{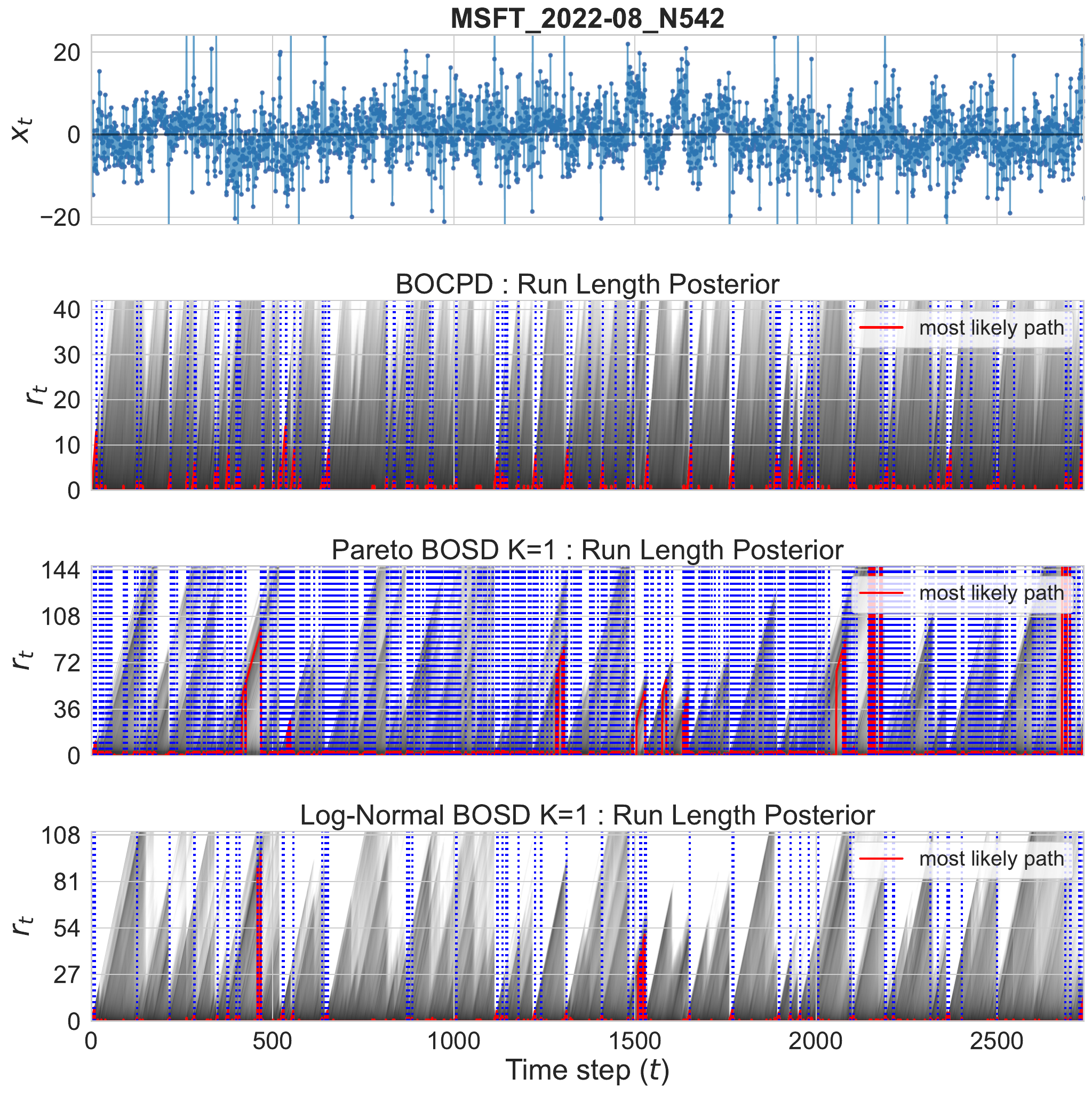}
        \caption{MSFT --- MSE-calibrated}
    \end{subfigure}\hfill
    \begin{subfigure}[b]{0.48\linewidth}
        \includegraphics[width=\linewidth]{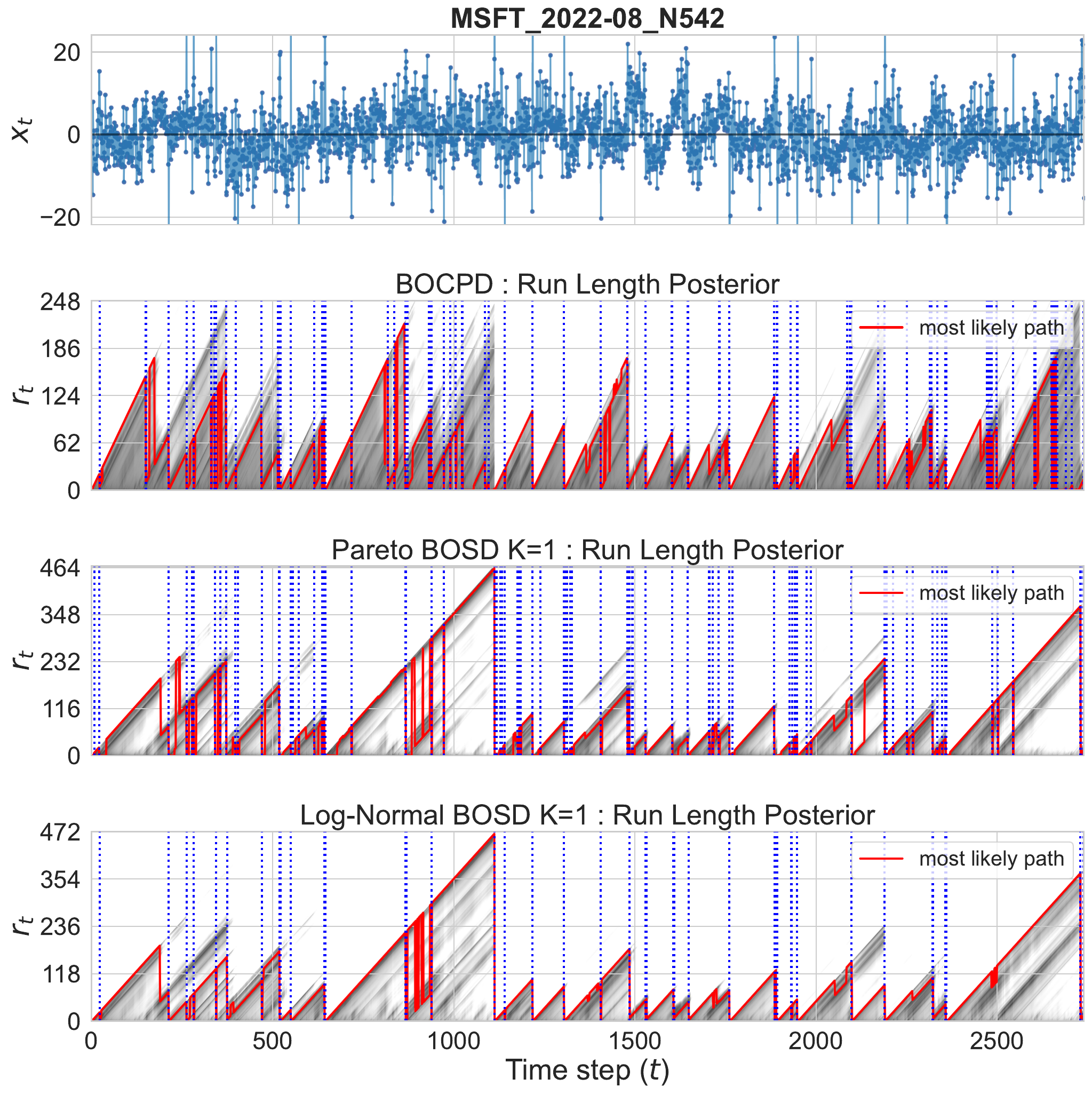}
        \caption{MSFT --- Log-Lik-calibrated}
    \end{subfigure}

    \vspace{4pt}
    \begin{subfigure}[b]{0.48\linewidth}
        \includegraphics[width=\linewidth]{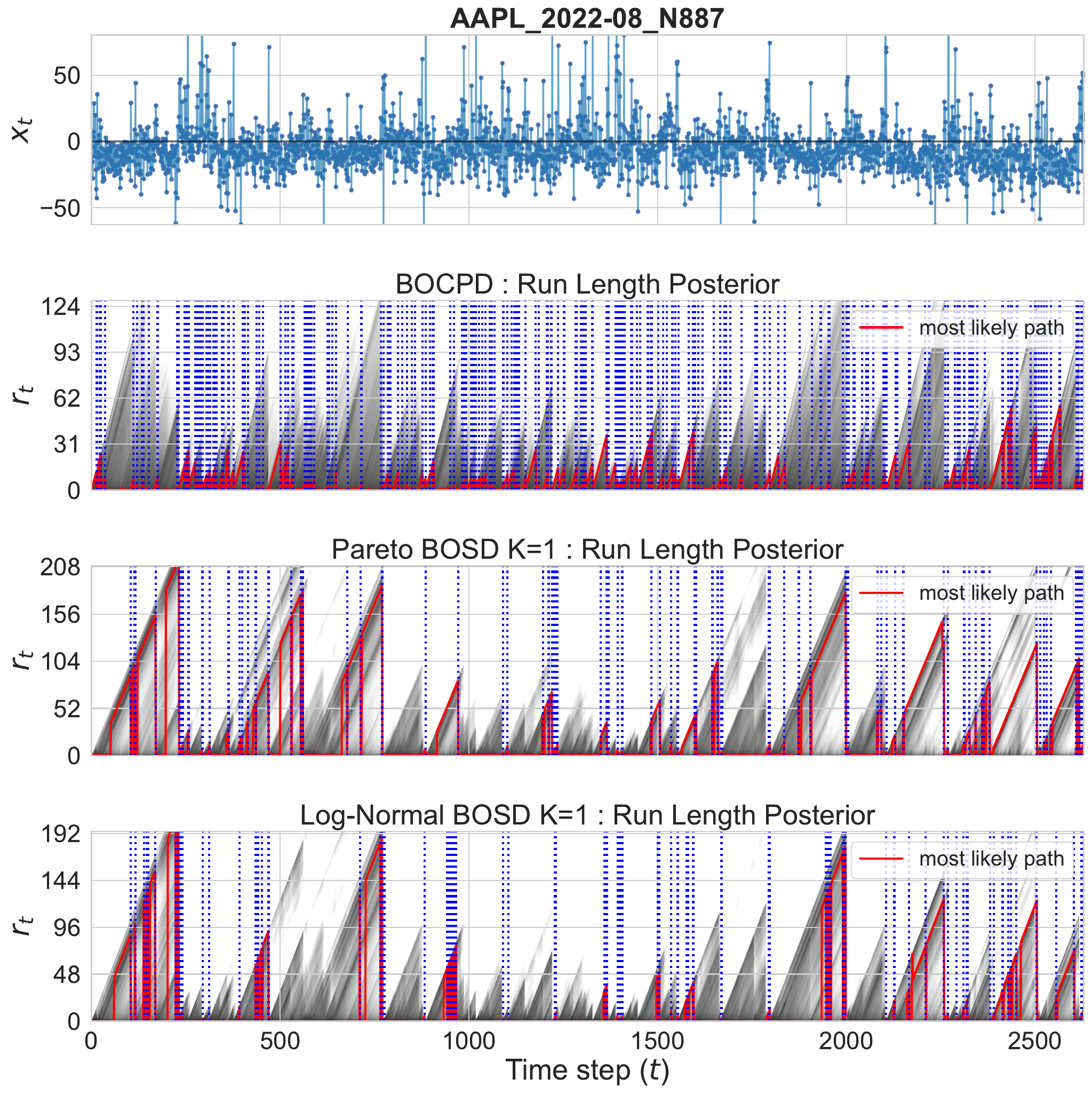}
        \caption{AAPL --- MSE-calibrated}
    \end{subfigure}\hfill
    \begin{subfigure}[b]{0.48\linewidth}
        \includegraphics[width=\linewidth]{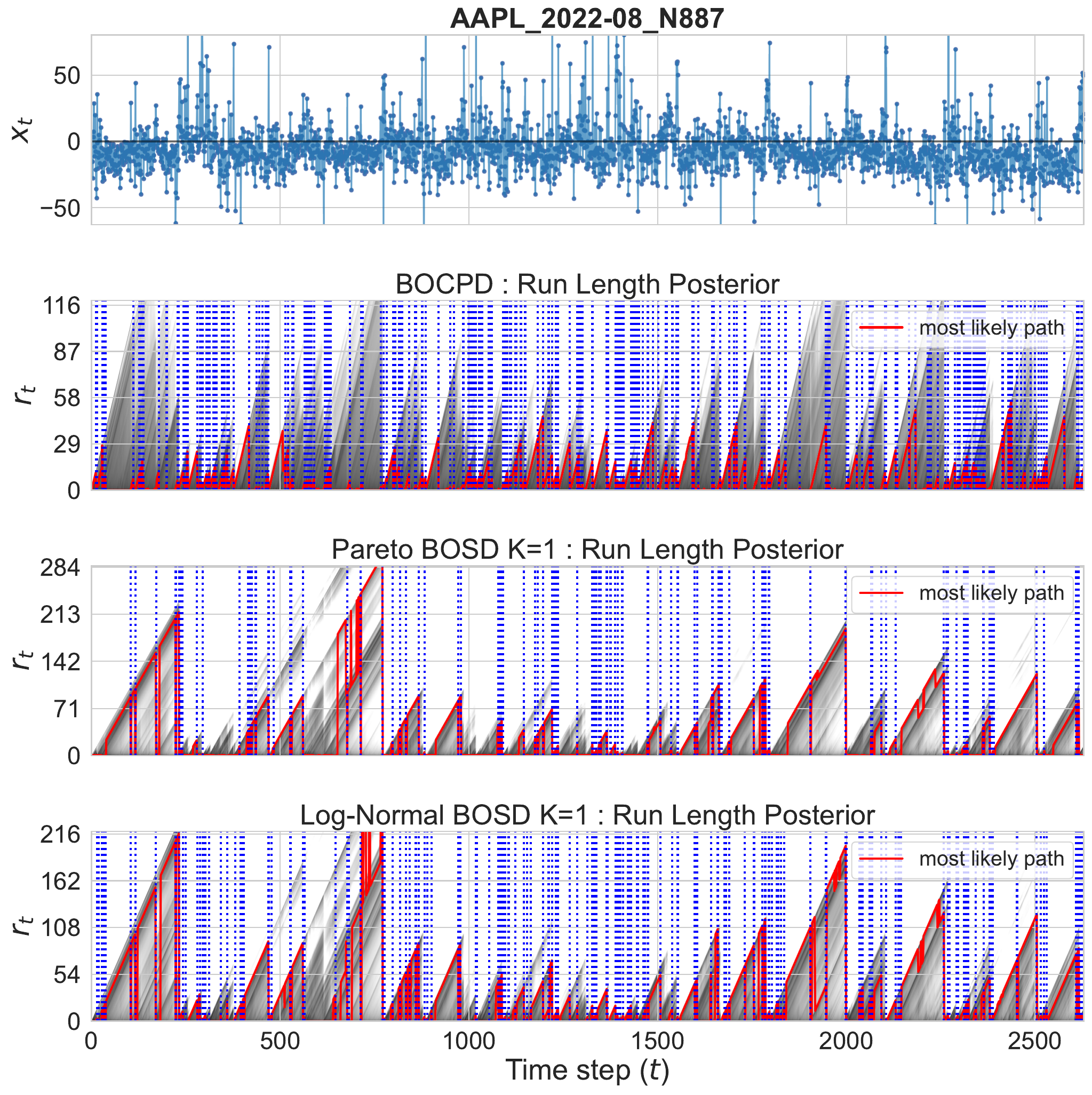}
        \caption{AAPL --- Log-Lik-calibrated}
    \end{subfigure}
    \caption{Run Length Posteriors on the 2022-08 datasets, under both calibration criteria.}
    \label{fig:rlpost_2022_08}
\end{figure}

\begin{figure}[H]
    \centering
    \begin{subfigure}[b]{0.48\linewidth}
        \includegraphics[width=\linewidth]{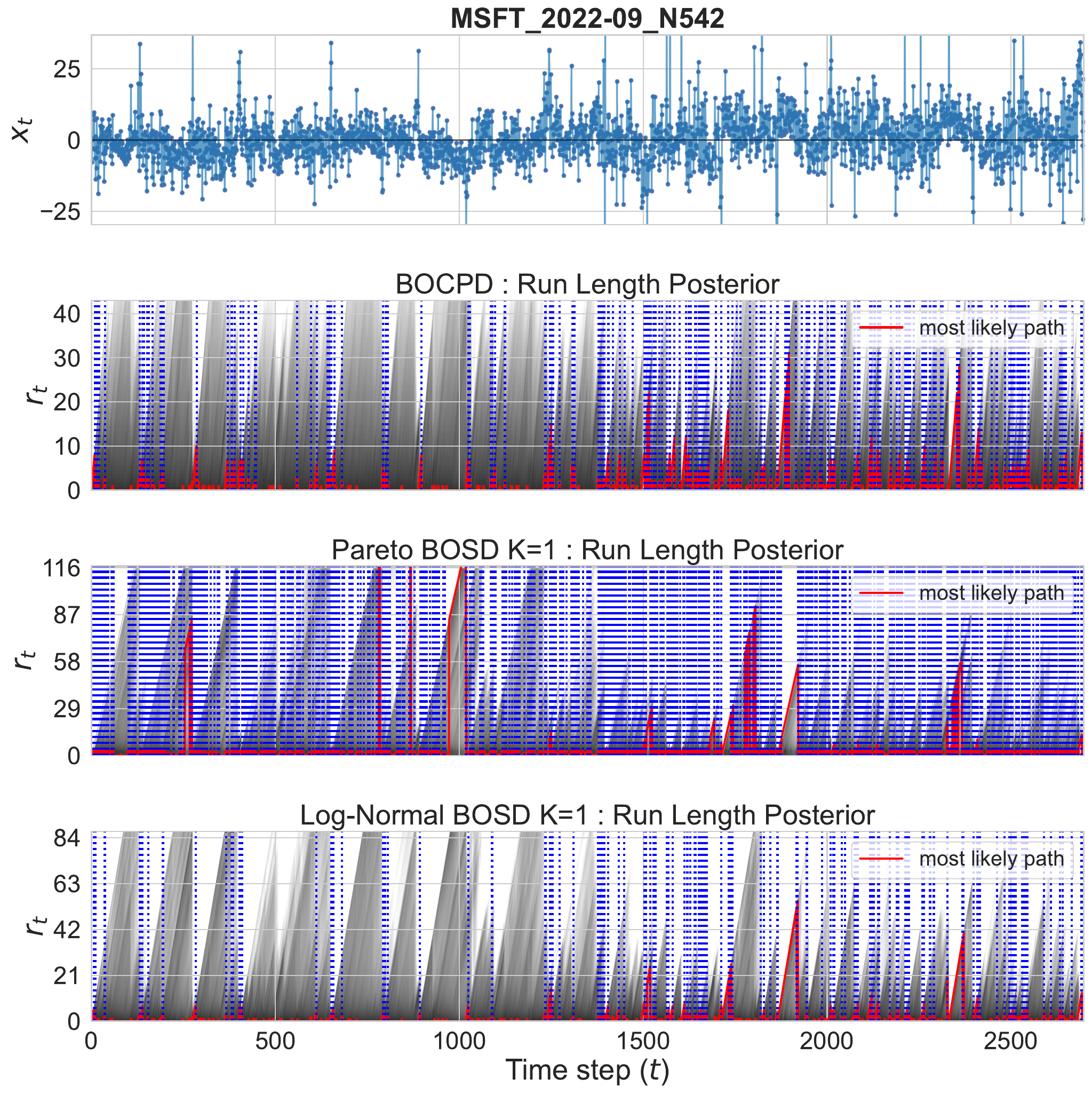}
        \caption{MSFT --- MSE-calibrated}
    \end{subfigure}\hfill
    \begin{subfigure}[b]{0.48\linewidth}
        \includegraphics[width=\linewidth]{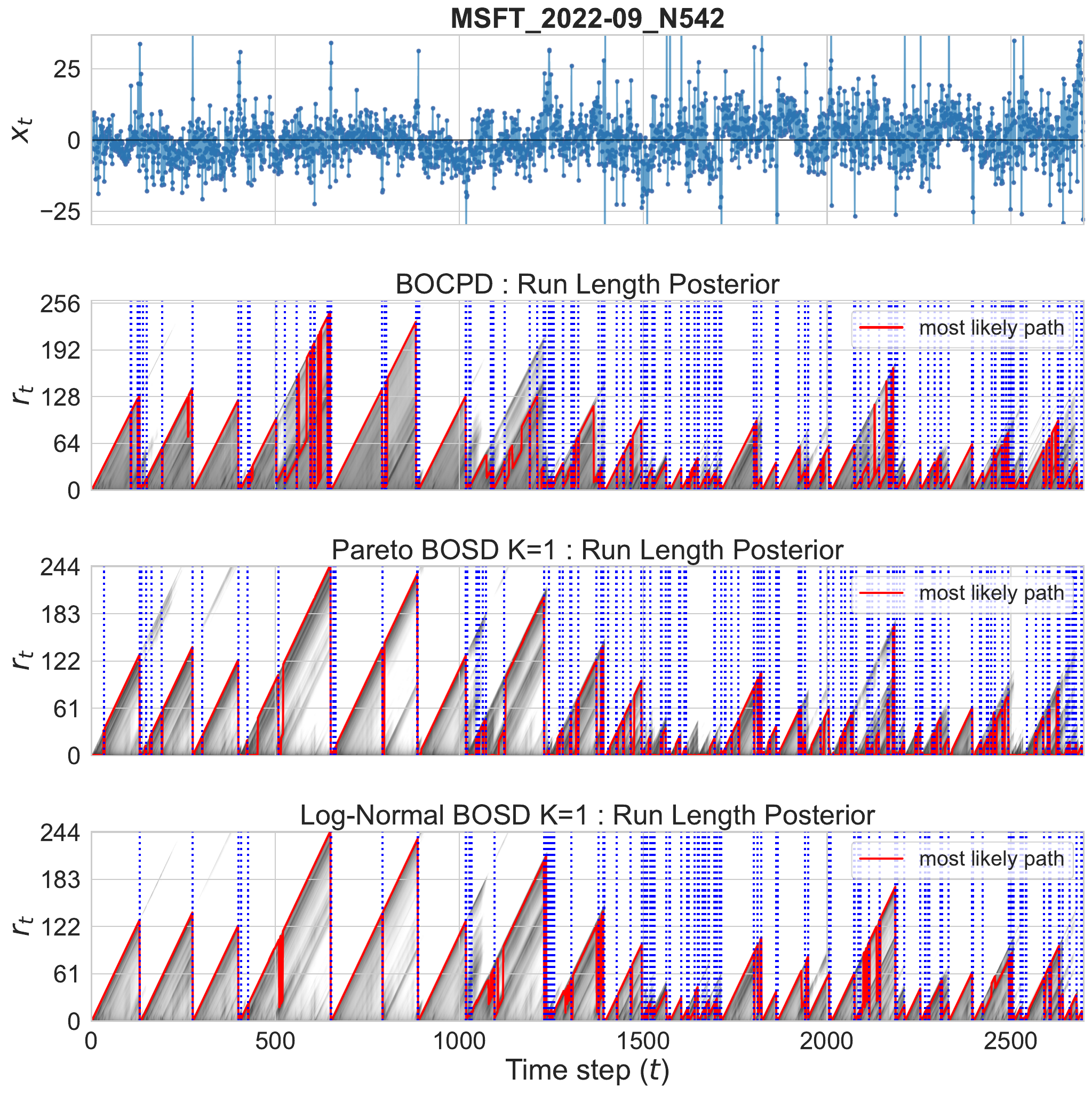}
        \caption{MSFT --- Log-Lik-calibrated}
    \end{subfigure}

    \vspace{4pt}
    \begin{subfigure}[b]{0.48\linewidth}
        \includegraphics[width=\linewidth]{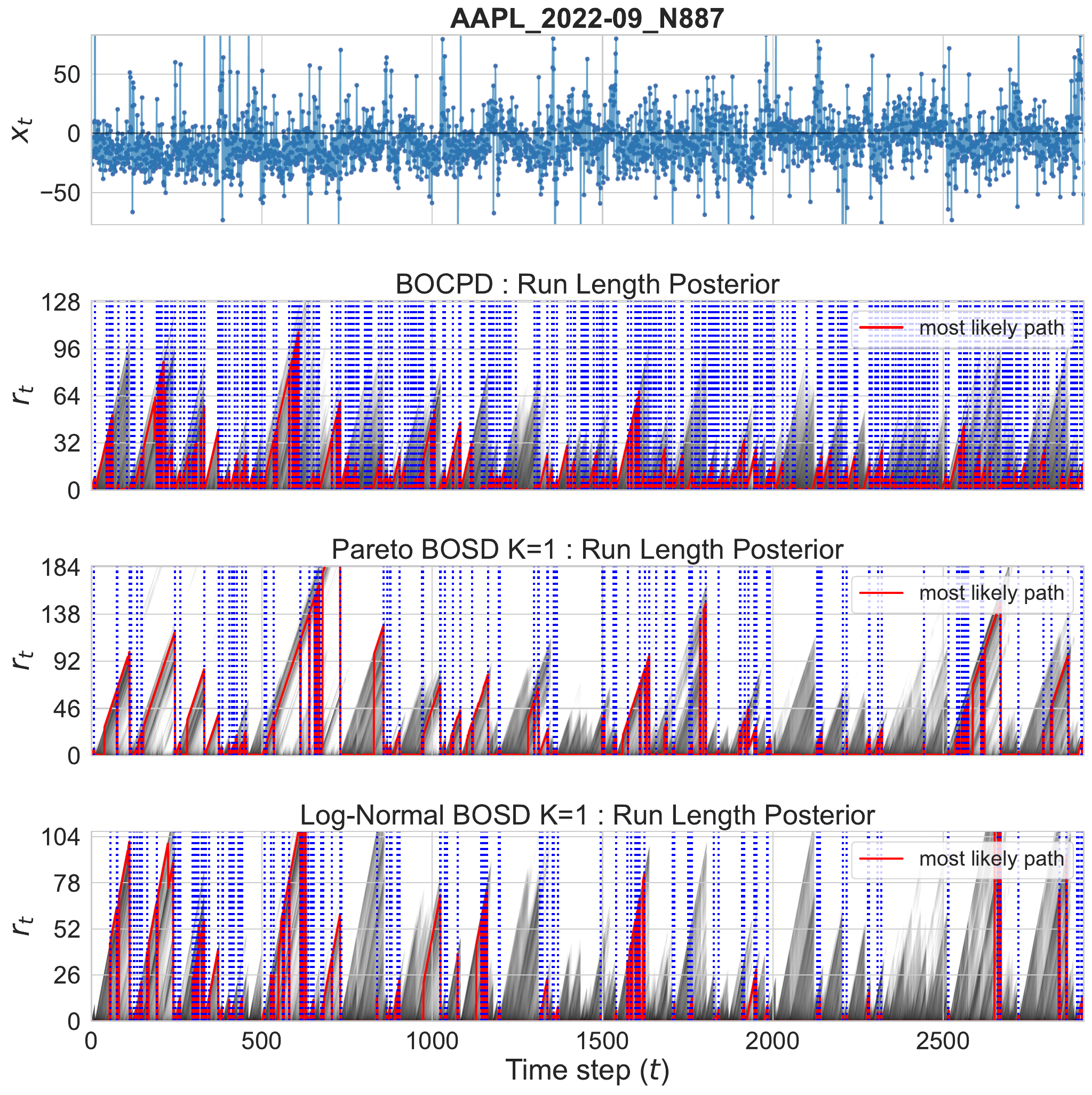}
        \caption{AAPL --- MSE-calibrated}
    \end{subfigure}\hfill
    \begin{subfigure}[b]{0.48\linewidth}
        \includegraphics[width=\linewidth]{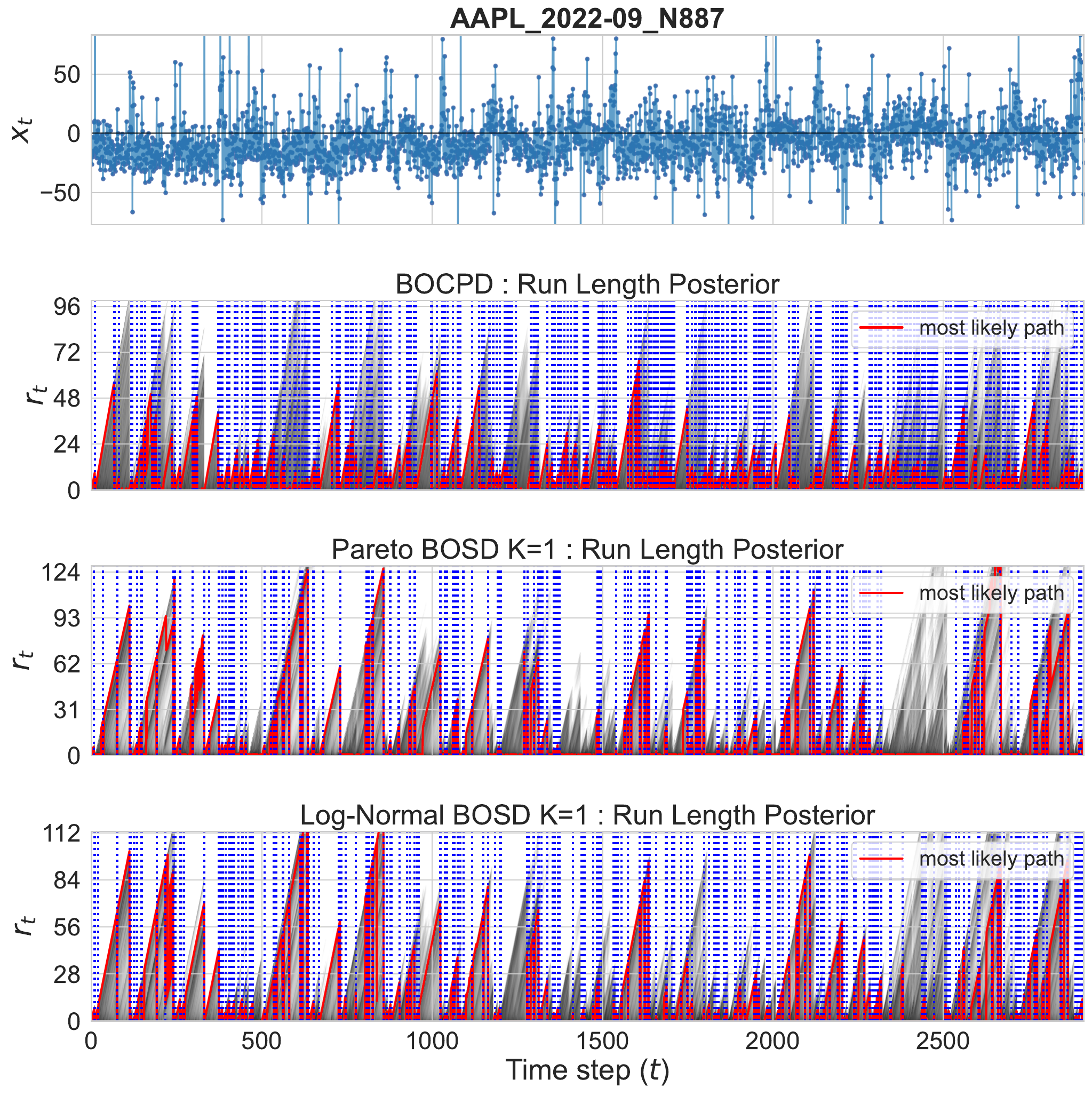}
        \caption{AAPL --- Log-Lik-calibrated}
    \end{subfigure}
    \caption{Run Length Posteriors on the 2022-09 datasets, under both calibration criteria.}
    \label{fig:rlpost_2022_09}
\end{figure}

\begin{figure}[H]
    \centering
    \begin{subfigure}[b]{0.48\linewidth}
        \includegraphics[width=\linewidth]{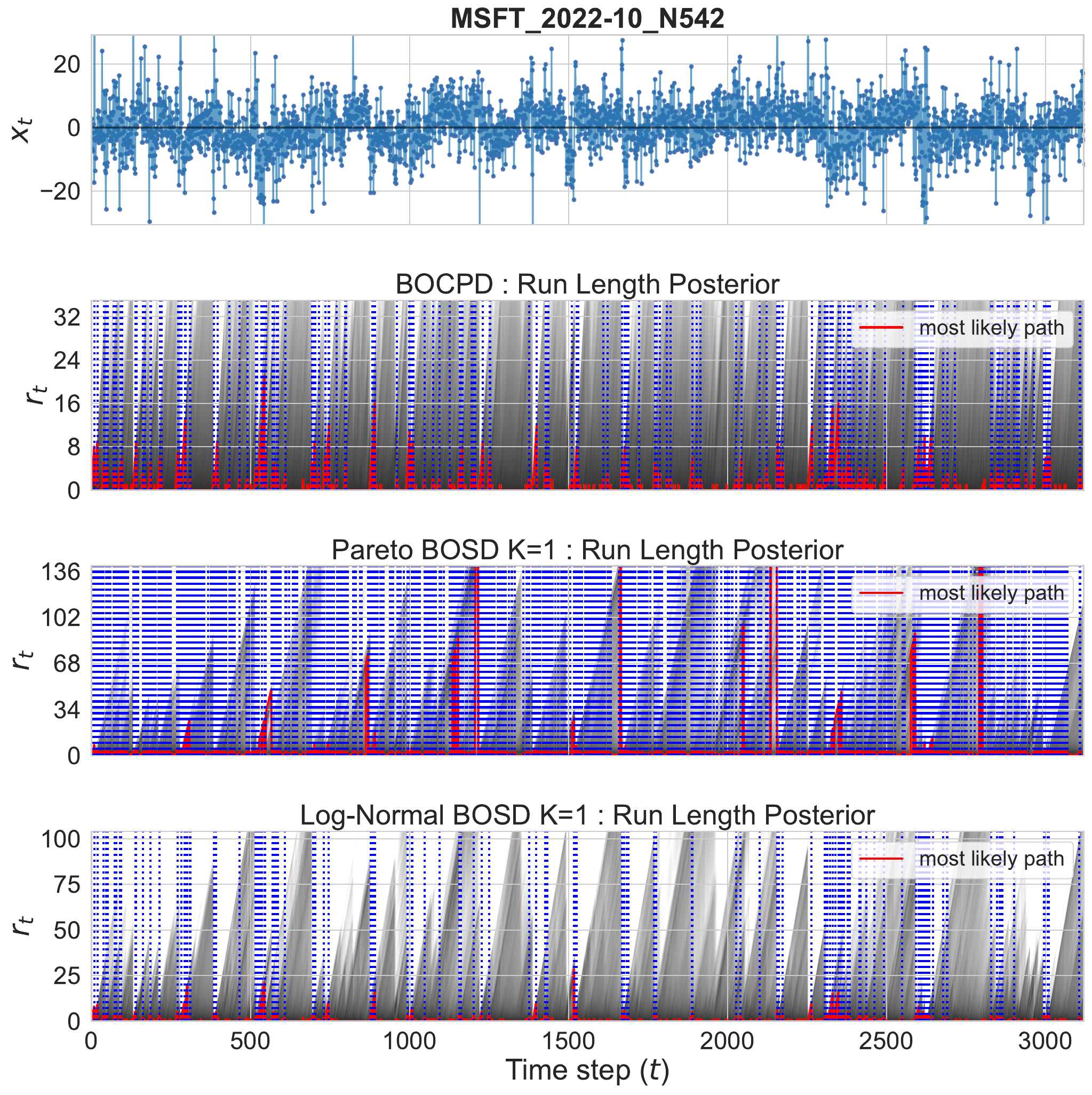}
        \caption{MSFT --- MSE-calibrated}
    \end{subfigure}\hfill
    \begin{subfigure}[b]{0.48\linewidth}
        \includegraphics[width=\linewidth]{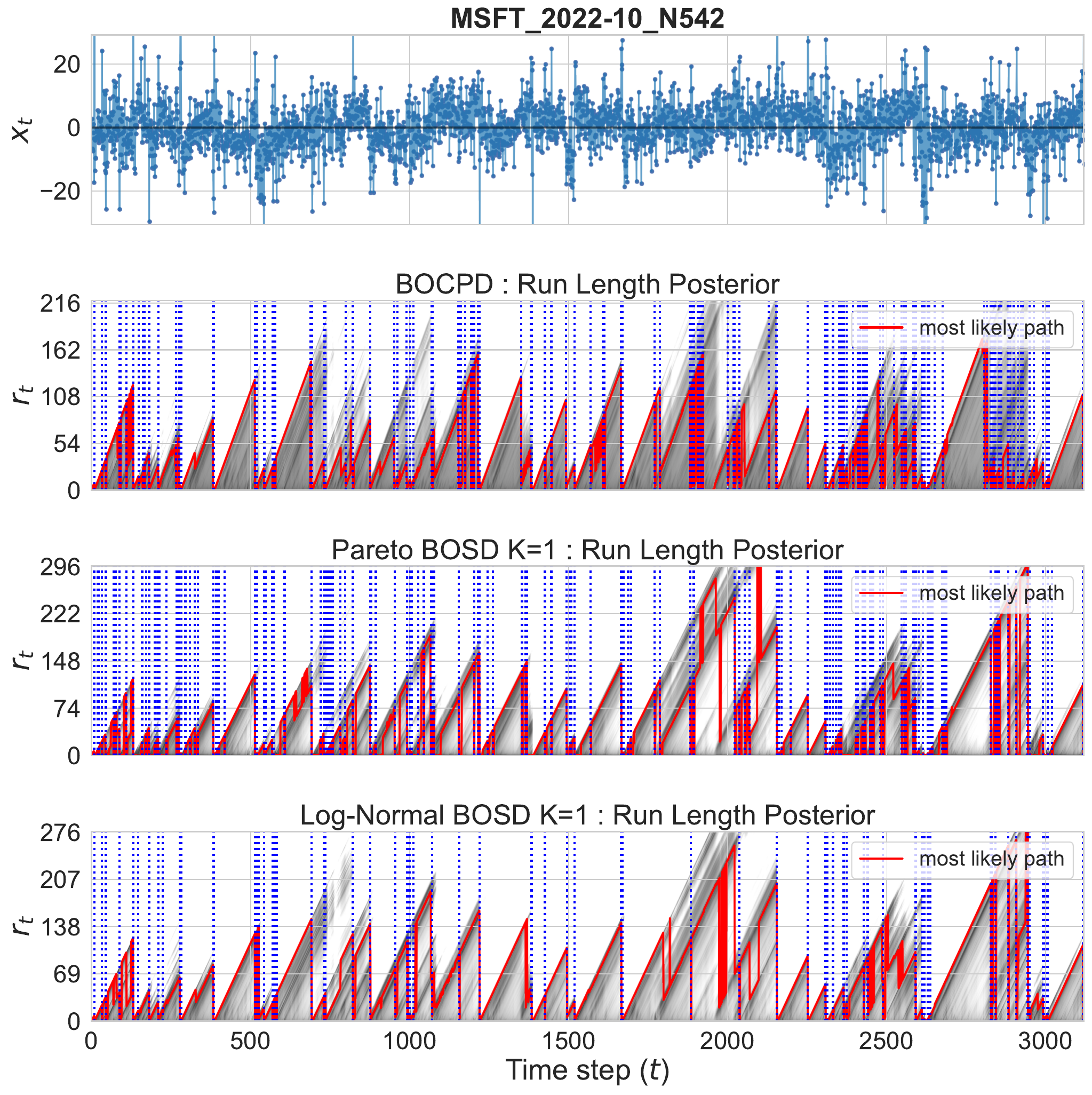}
        \caption{MSFT --- Log-Lik-calibrated}
    \end{subfigure}

    \vspace{4pt}
    \begin{subfigure}[b]{0.48\linewidth}
        \includegraphics[width=\linewidth]{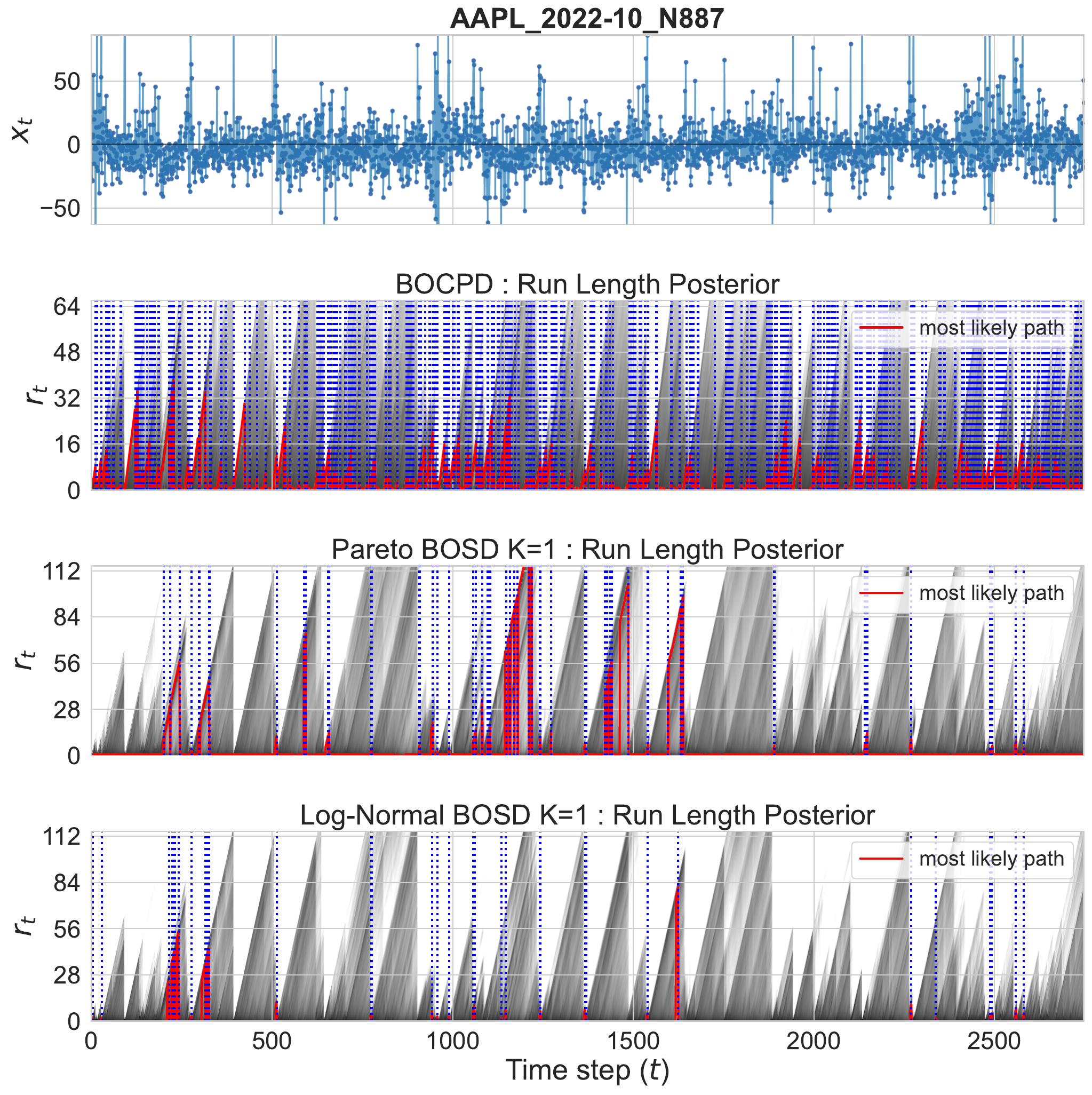}
        \caption{AAPL --- MSE-calibrated}
    \end{subfigure}\hfill
    \begin{subfigure}[b]{0.48\linewidth}
        \includegraphics[width=\linewidth]{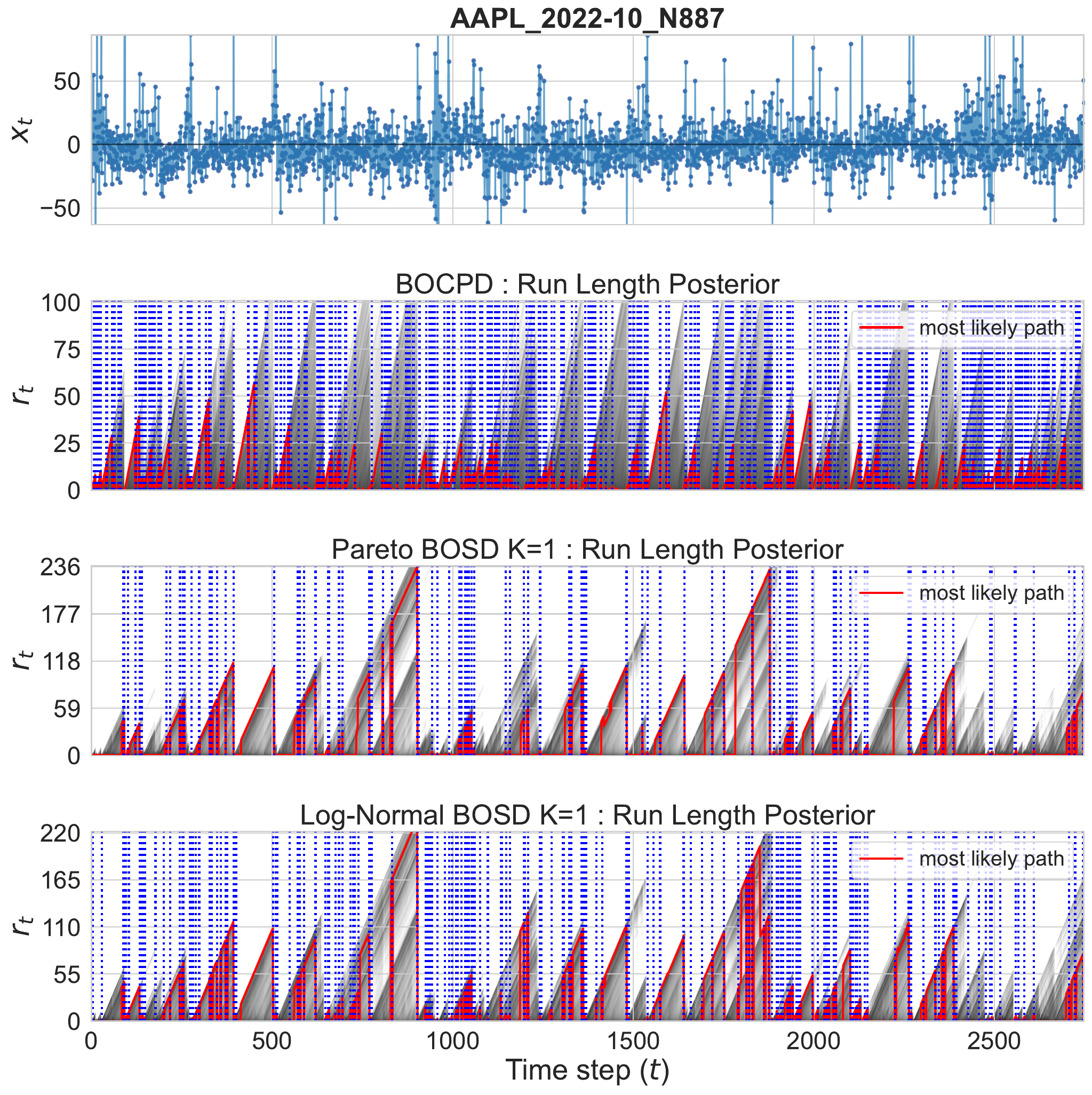}
        \caption{AAPL --- Log-Lik-calibrated}
    \end{subfigure}
    \caption{Run Length Posteriors on the 2022-10 datasets, under both calibration criteria.}
    \label{fig:rlpost_2022_10}
\end{figure}

\begin{figure}[H]
    \centering
    \begin{subfigure}[b]{0.48\linewidth}
        \includegraphics[width=\linewidth]{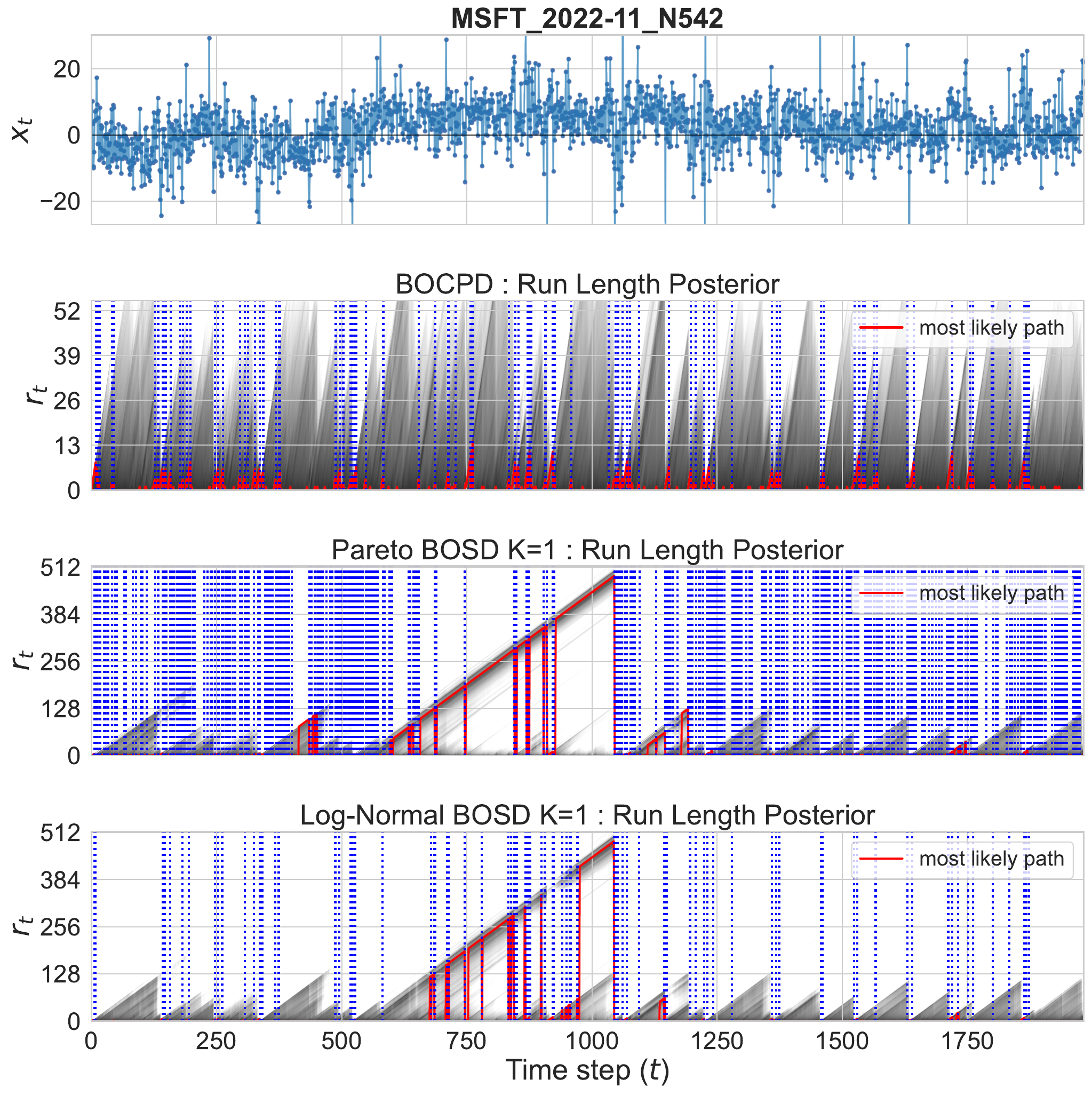}
        \caption{MSFT --- MSE-calibrated}
    \end{subfigure}\hfill
    \begin{subfigure}[b]{0.48\linewidth}
        \includegraphics[width=\linewidth]{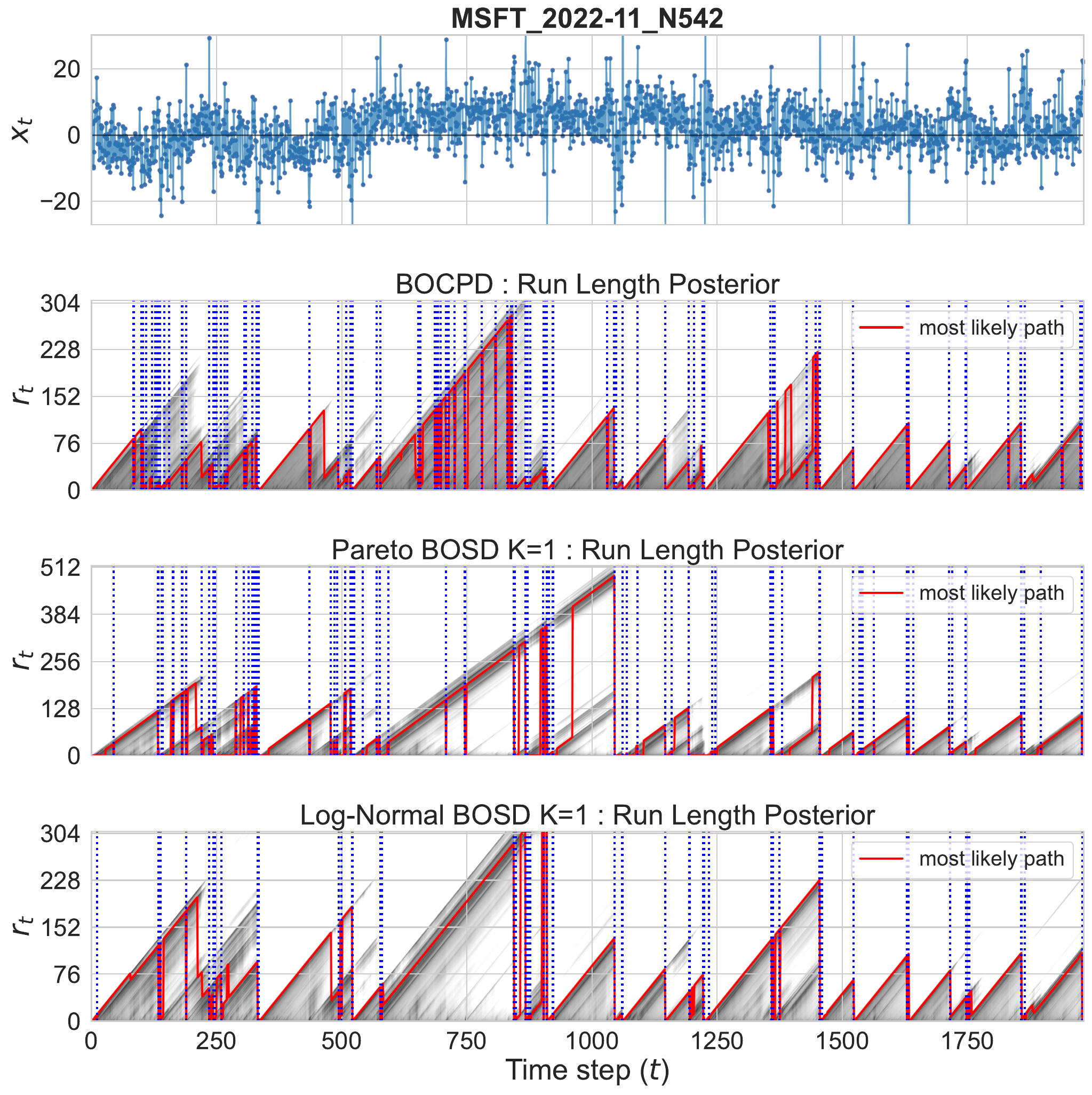}
        \caption{MSFT --- Log-Lik-calibrated}
    \end{subfigure}

    \vspace{4pt}
    \begin{subfigure}[b]{0.48\linewidth}
        \includegraphics[width=\linewidth]{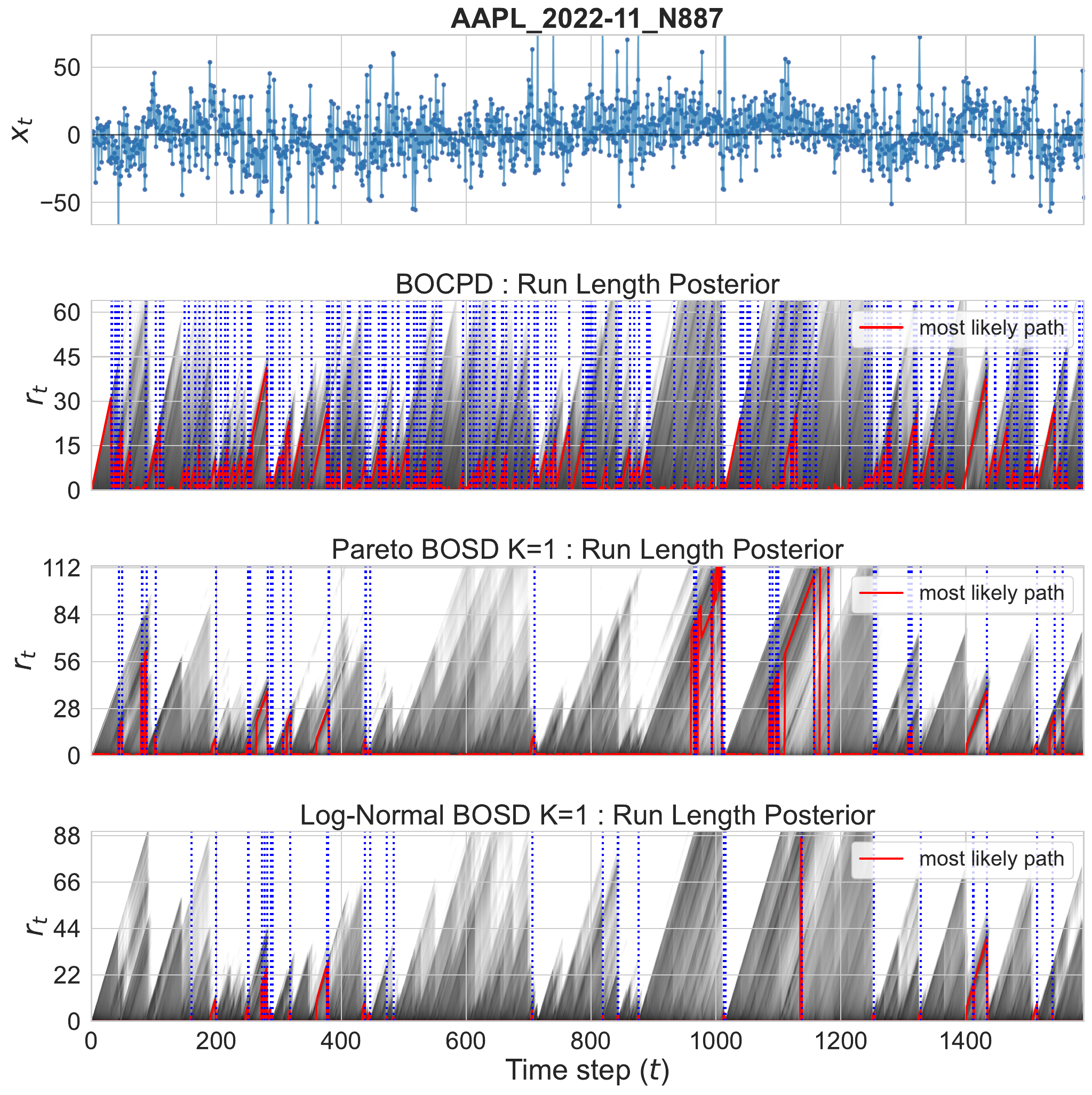}
        \caption{AAPL --- MSE-calibrated}
    \end{subfigure}\hfill
    \begin{subfigure}[b]{0.48\linewidth}
        \includegraphics[width=\linewidth]{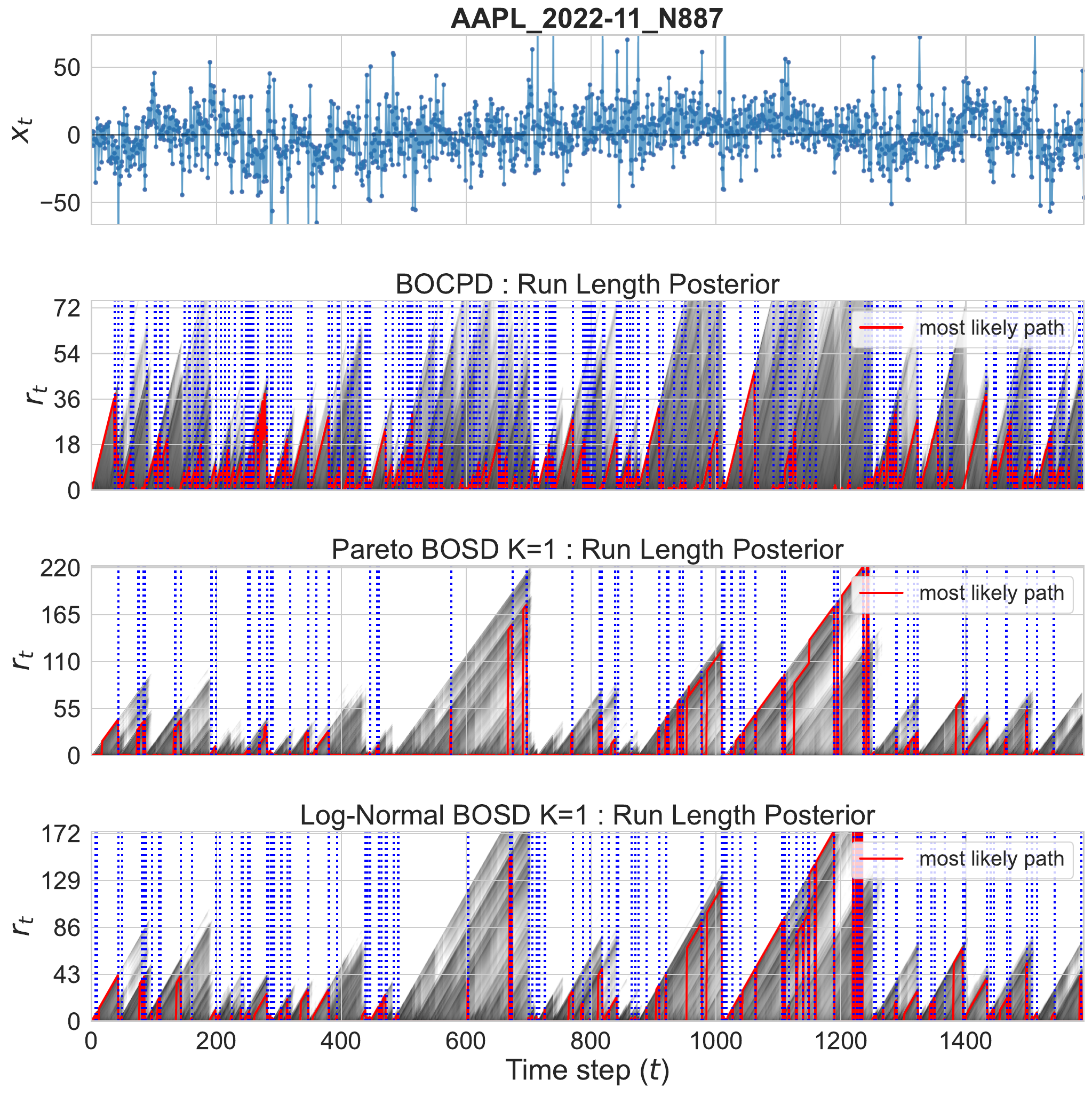}
        \caption{AAPL --- Log-Lik-calibrated}
    \end{subfigure}
    \caption{Run Length Posteriors on the 2022-11 datasets, under both calibration criteria.}
    \label{fig:rlpost_2022_11}
\end{figure}

\newpage
\section{BOCPD regime duration distributions}\label{app:valid_bocpd}

A theoretical property of the standard BOCPD framework is that if the hazard function is a constant $H = 1/h$, the distribution of regime durations should be a Geometric distribution.

\subsection{Data Generating Process and Experimental Setup}

We perform a loop of 1000 runs with long synthetic datasets of $T = 1{,}000$ data points to ensure a statistically significant number of regimes. Within each regime, the data follows $x_t \sim \mathcal{N}(\mu, \sigma^2)$ where the variance is fixed at $\sigma^2 = 10^2$. The data hazard rate is fixed at $H_{\text{data}} = 1/h_{\text{data}} = 1/100$. When a changepoint is triggered, the new regime mean $\mu$ is redrawn from a Gaussian prior distribution: $\mu \sim \mathcal{N}(\mu_0 = 0, \ \sigma_0^2)$.

The standard BOCPD algorithm is run using an identical hazard rate parameter $h_{\text{algo}} = 100$, perfectly matching the hazard rate of the data-generating process.

\subsection{Theoretical Distribution vs. Empirical Histogram}

Mathematically, given a constant hazard rate $H$, the probability that a regime survives for exactly $n-1$ steps and terminates at the $n$-th step is governed by the geometric probability distribution. Under this distribution, the theoretical expected length of a regime is given by $\mathbb{E}[d] = 1/H = h_{\text{data}}$.

By extracting the durations of all the regimes successfully detected by the BOCPD algorithm, we construct an empirical frequency histogram (Figure~\ref{fig:bocpd_geom_synth}).

\begin{figure}[H]
    \centering
    \begin{subfigure}[b]{0.42\linewidth}
        \includegraphics[width=\linewidth]{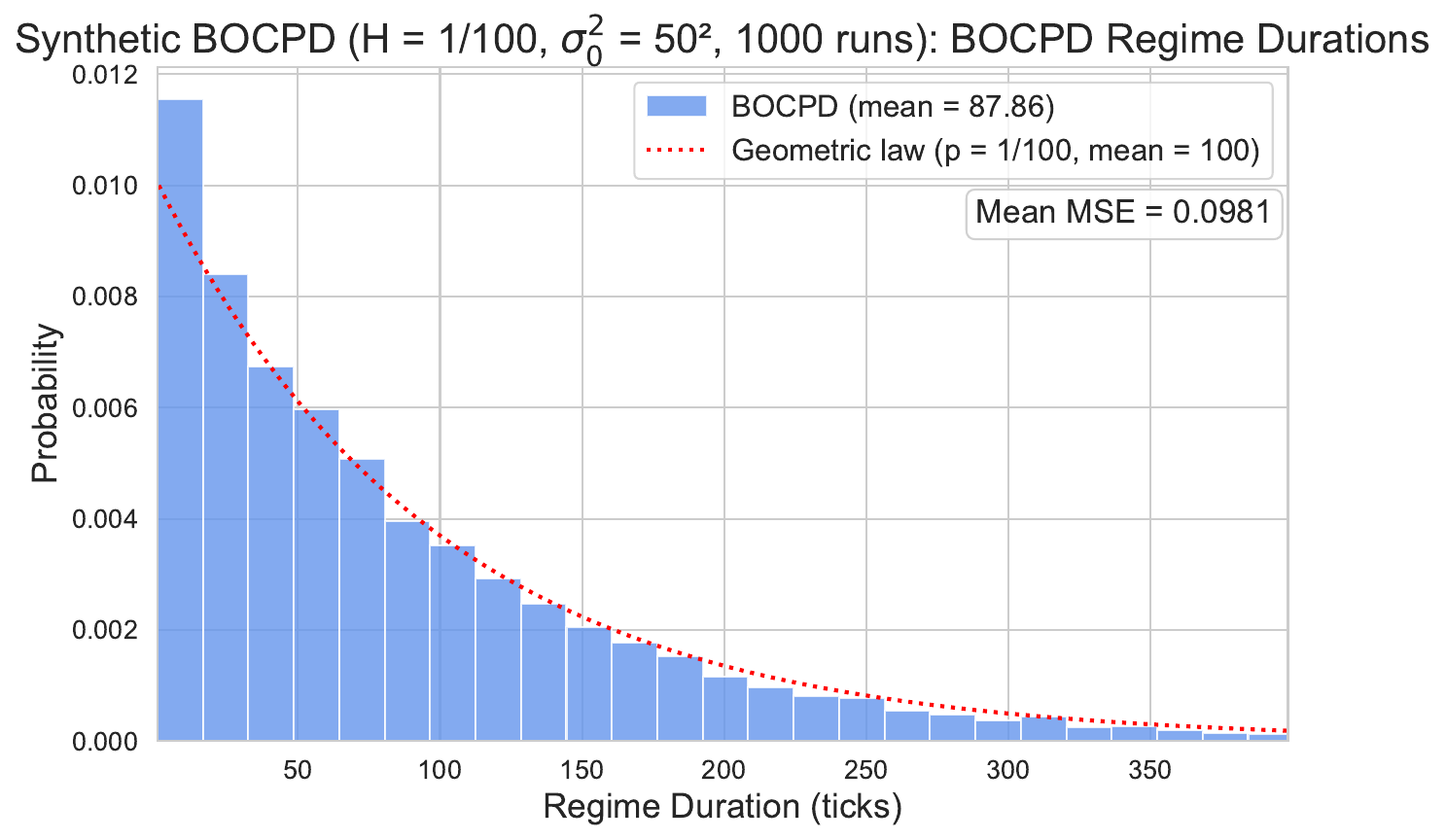}
        \caption{$\sigma_0 = 50$}
    \end{subfigure}\hfill
    \begin{subfigure}[b]{0.42\linewidth}
        \includegraphics[width=\linewidth]{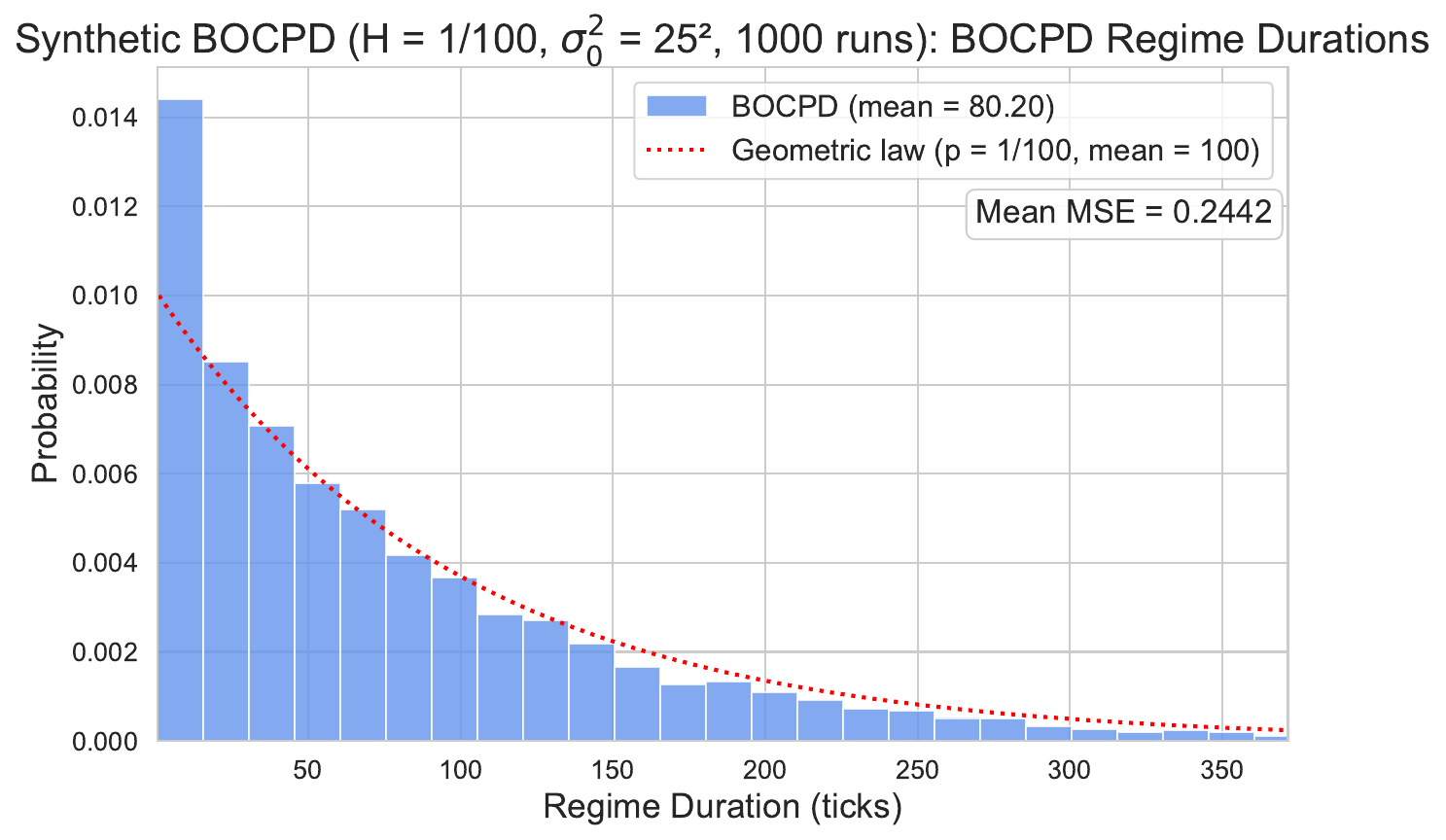}
        \caption{$\sigma_0 = 25$}
    \end{subfigure}

    \vspace{2pt}
    \begin{subfigure}[b]{0.42\linewidth}
        \includegraphics[width=\linewidth]{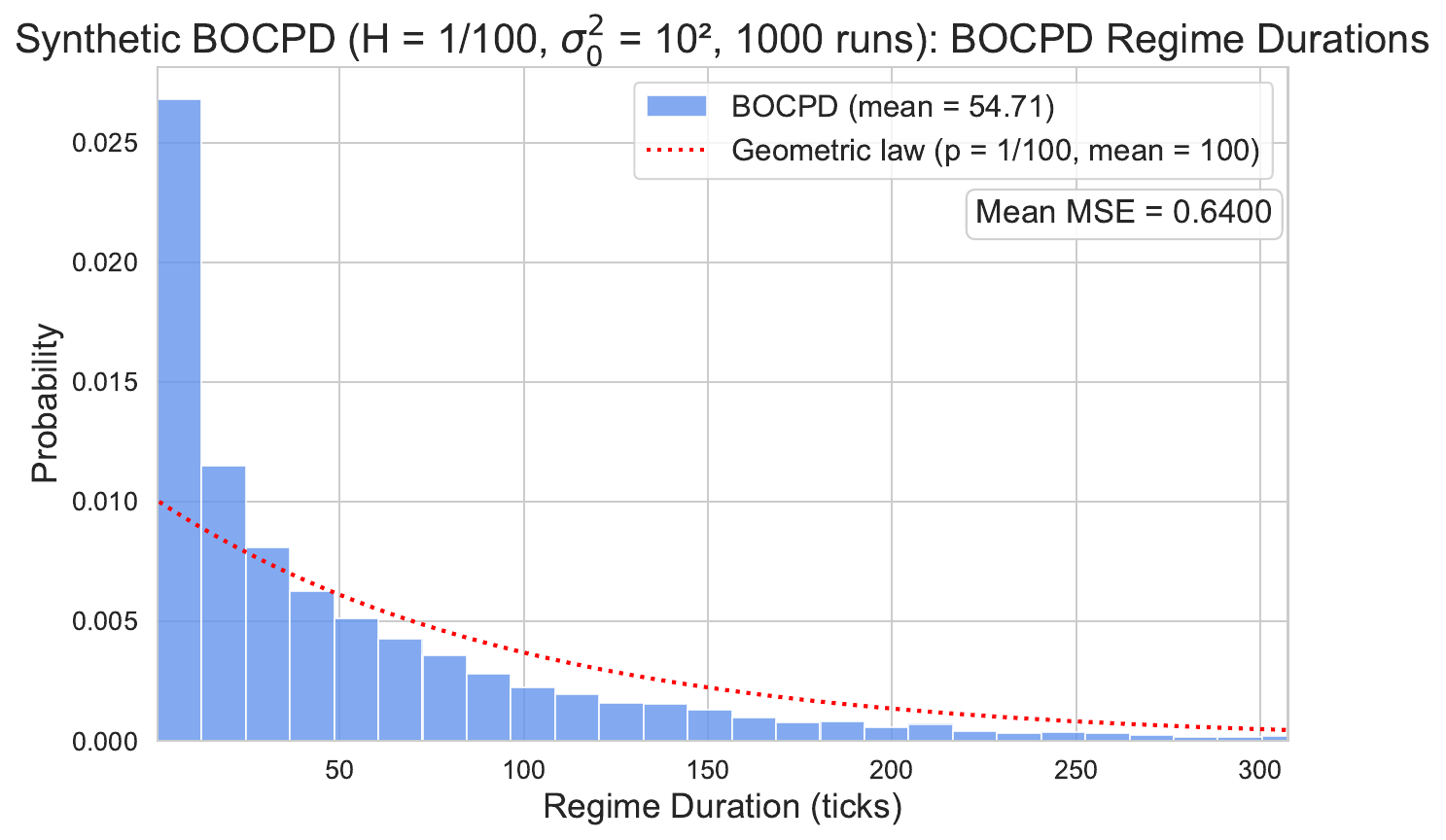}
        \caption{$\sigma_0 = 10$}
    \end{subfigure}\hfill
    \begin{subfigure}[b]{0.42\linewidth}
        \includegraphics[width=\linewidth]{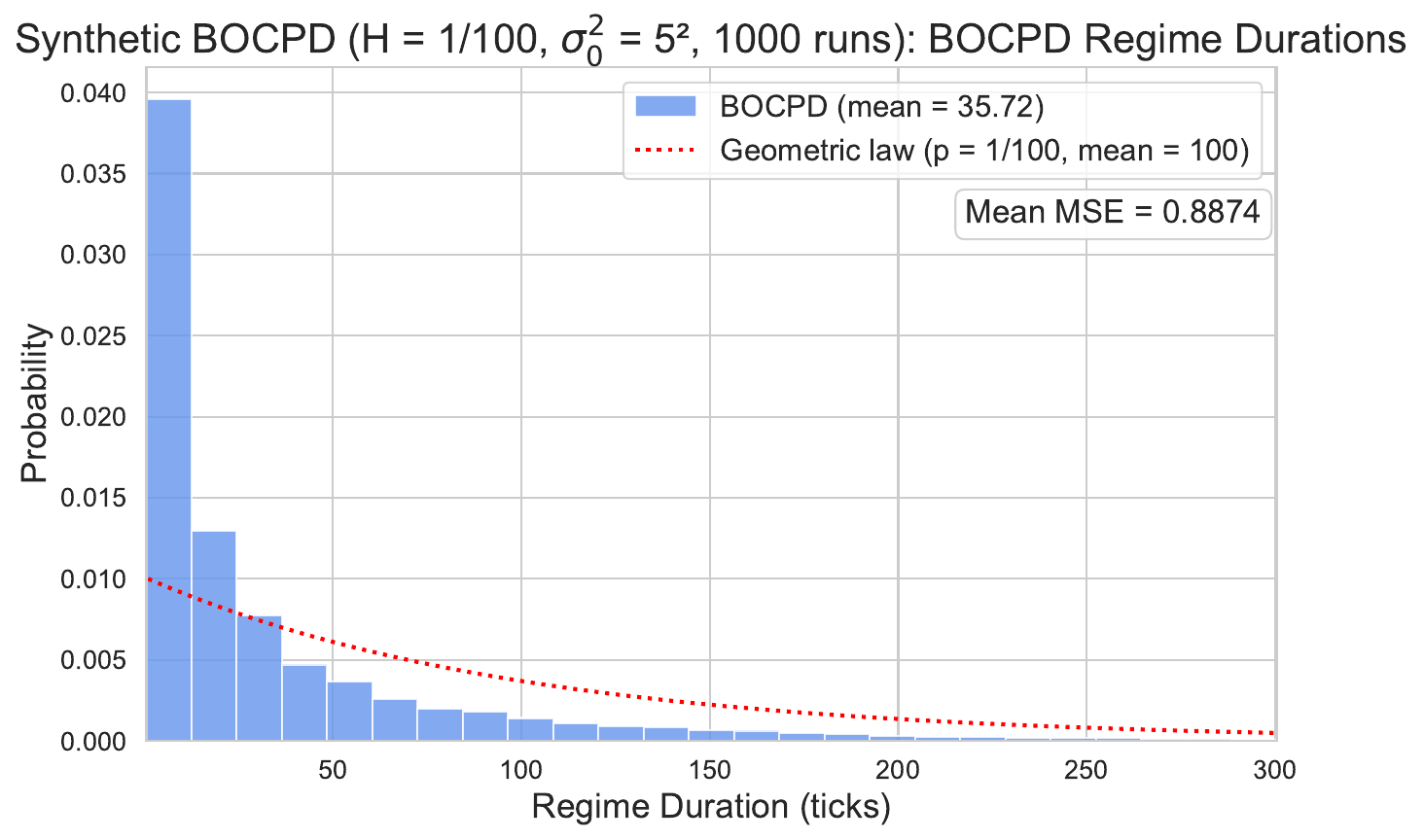}
        \caption{$\sigma_0 = 5$}
    \end{subfigure}
    \caption{Empirical histograms of detected regime lengths against the theoretical Geometric distribution curve. Averaged over 1000 runs on synthetic data.}
    \label{fig:bocpd_geom_synth}
\end{figure}

\newpage
\subsection{Empirical Validation on real data}

We now perform an empirical validation to verify this property on real data. For each asset, we pool the detected regime durations across the four available 2022 months and build a single averaged empirical histogram, compared against the theoretical Geometric distribution implied by the calibrated hazard rate (Figure~\ref{fig:geom_real}).

\begin{figure}[H]
    \centering
    \begin{subfigure}[b]{0.42\linewidth}
        \includegraphics[width=\linewidth]{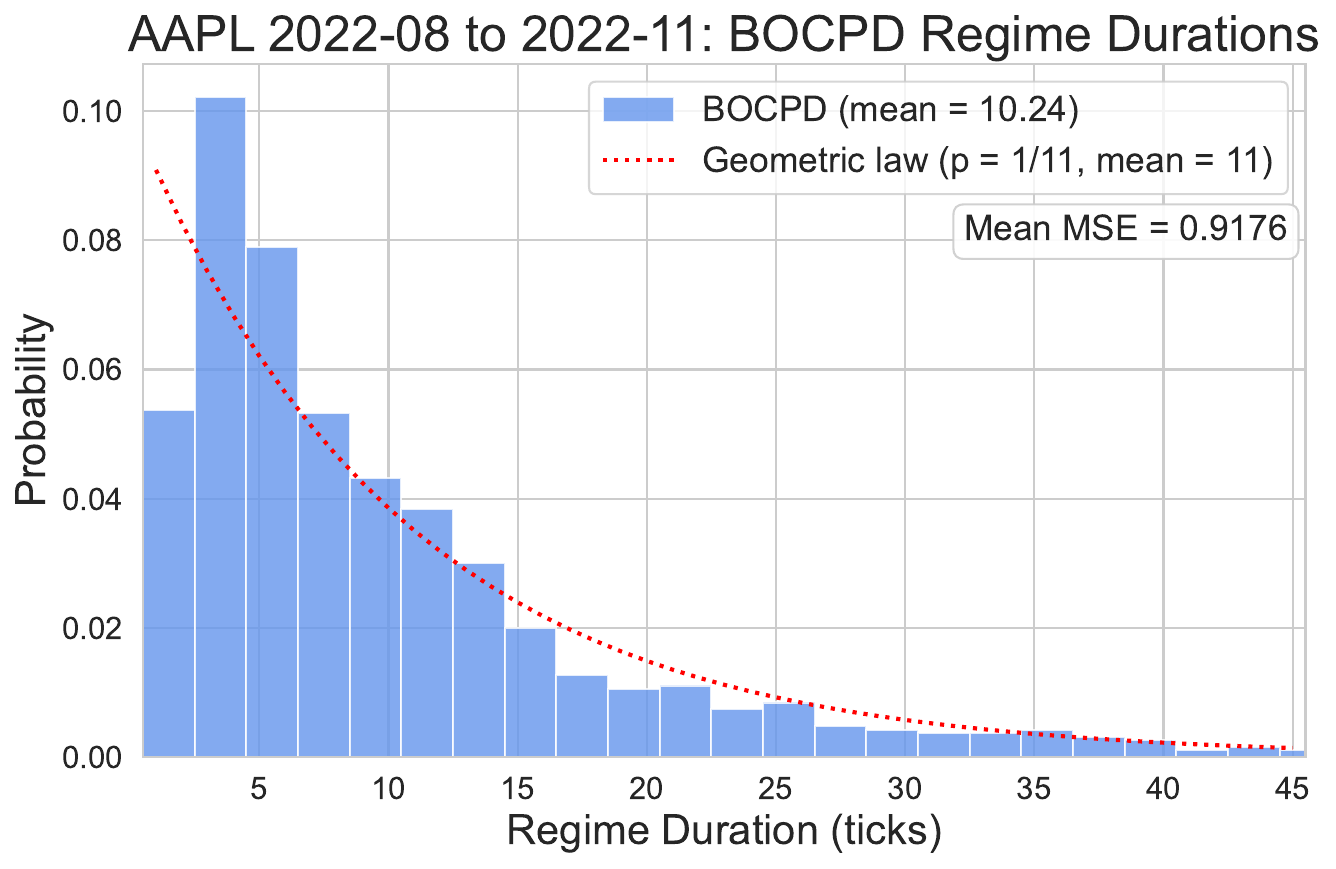}
        \caption{AAPL --- MSE-calibrated}
    \end{subfigure}\hfill
    \begin{subfigure}[b]{0.42\linewidth}
        \includegraphics[width=\linewidth]{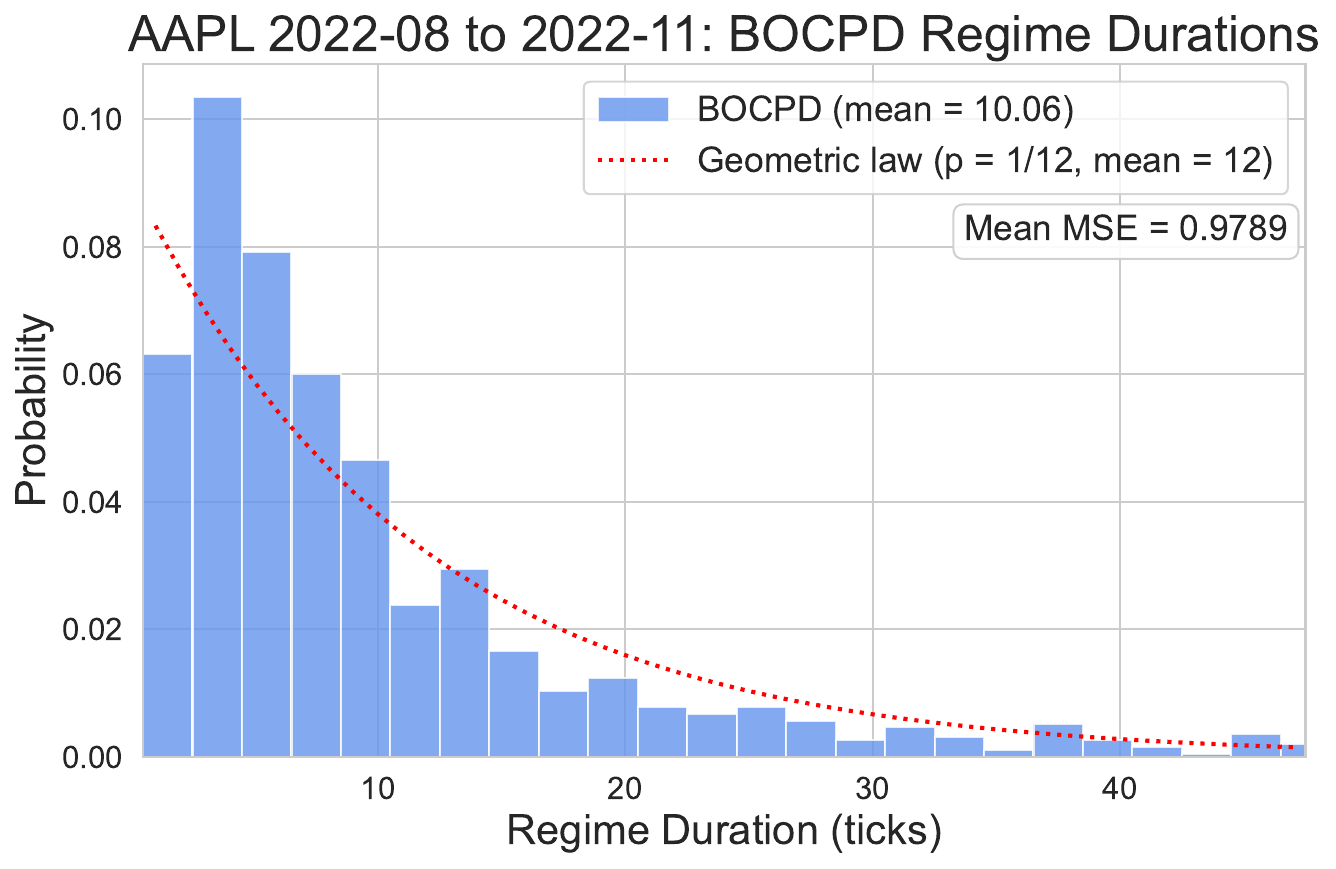}
        \caption{AAPL --- Log-Lik-calibrated}
    \end{subfigure}

    \vspace{2pt}
    \begin{subfigure}[b]{0.42\linewidth}
        \includegraphics[width=\linewidth]{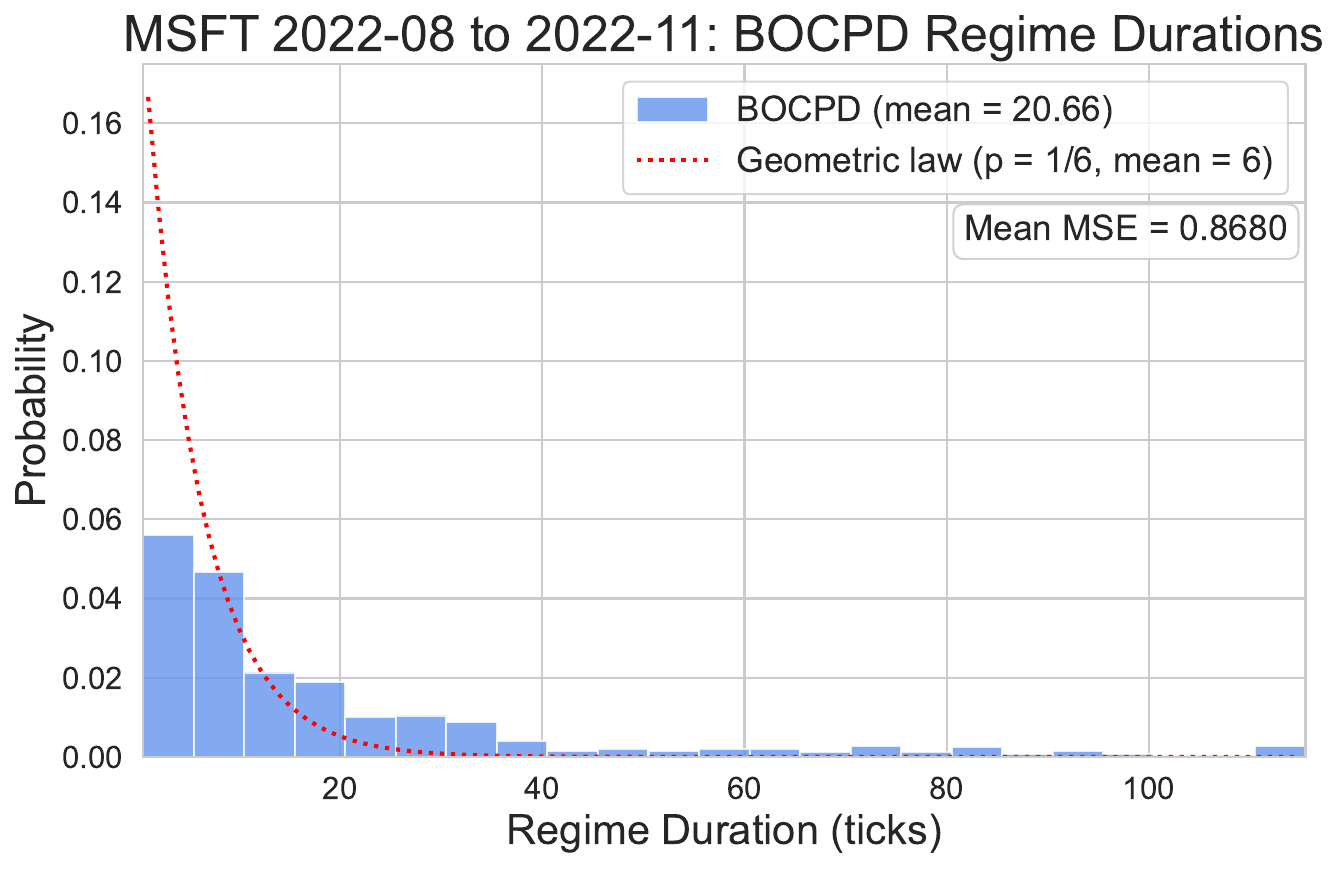}
        \caption{MSFT --- MSE-calibrated}
    \end{subfigure}\hfill
    \begin{subfigure}[b]{0.42\linewidth}
        \includegraphics[width=\linewidth]{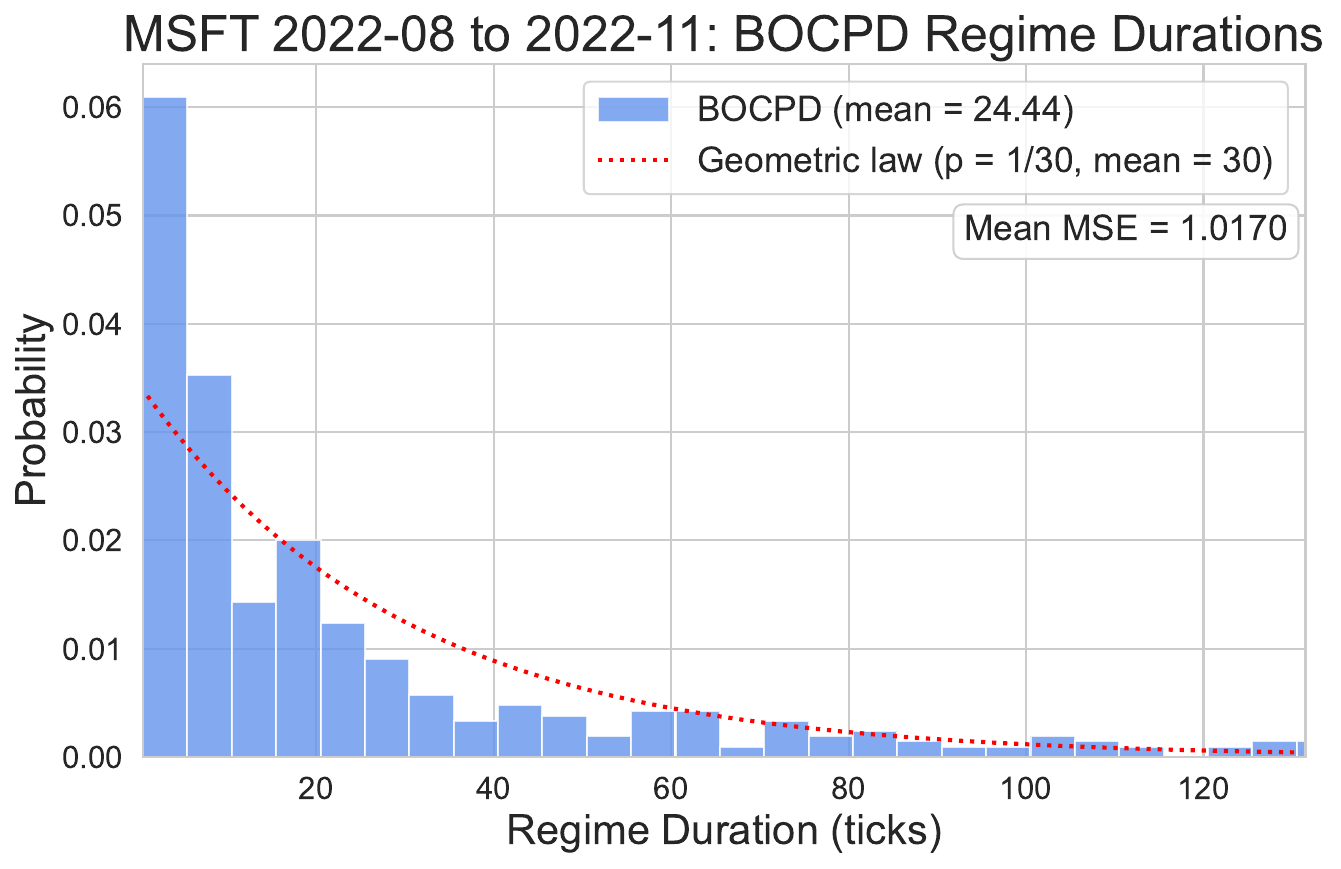}
        \caption{MSFT --- Log-Lik-calibrated}
    \end{subfigure}
    \caption{Detected regime durations pooled over 2022-08 to 2022-11, against the theoretical Geometric law implied by the calibrated hazard rate.}
    \label{fig:geom_real}
\end{figure}

\newpage
\section{BOSD ($K=1$) regime duration distributions}\label{app:valid_bosd}

The BOSD framework was introduced precisely to lift the constraint of Geometric distribution for the regime durations. We verify that when the data is generated by an HSMM whose regime durations follow a chosen distribution, and the matching hazard vector is supplied to the algorithm, the durations of the detected regimes recover that same distribution.

\subsection{Data Generating Process and Experimental Setup}

We reproduce the Monte Carlo protocol of the BOCPD validation, replacing the memoryless Bernoulli generator by the HSMM-consistent generative process. The total duration of each regime is now drawn directly from an explicit duration distribution. The BOSD ($K=1$) algorithm is then run with the hazard vector $H(r)$ derived from the same duration distribution, so that the inference model is perfectly matched to the generative process. This is the exact analogue of setting $h_{\text{algo}} = h_{\text{data}}$ in the BOCPD experiment.

\subsection{Theoretical Distributions vs. Empirical Histograms}

\textbf{Pareto duration distribution.} The Pareto distribution with shape $\alpha$ and minimum $d_{\text{min}}$ yields a continuous hazard, which decreases: the longer a regime has already survived, the less likely it is to terminate at the next step. For $\alpha = 1.5$ and $d_{\text{min}} = 2$, the mean duration is $\mathbb{E}[d] = \alpha d_{\min}/(\alpha-1) = 6$, while the variance is infinite (a finite variance would require $\alpha > 2$).

\textbf{Log-Normal duration distribution.} The Log-Normal distribution with shape $s$ and scale $e^{\mu}$ has the mean $\mathbb{E}[d] = \exp(\mu + s^2/2)$. For $s = 1$ and scale $= 5$ (i.e. $\mu = \ln 5$), this yields $\mathbb{E}[d] = 5 e^{1/2} \approx 8.24$. The Log-Normal is lighter-tailed than the Pareto, and its hazard first rises and then decays.

\begin{figure}[H]
    \centering
    \begin{subfigure}[b]{0.42\linewidth}
        \includegraphics[width=\linewidth]{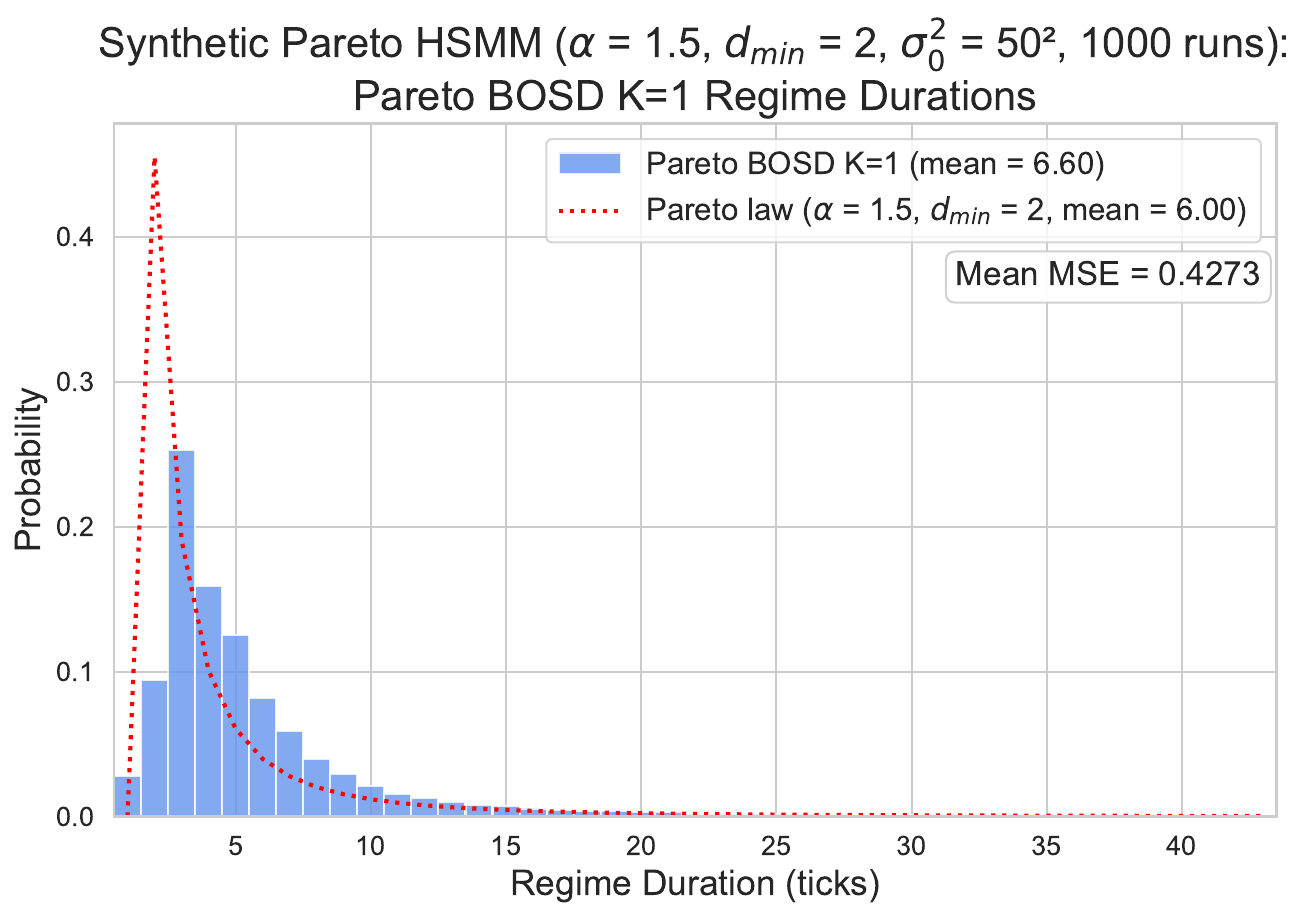}
        \caption{$\sigma_0 = 50$}
    \end{subfigure}\hfill
    \begin{subfigure}[b]{0.42\linewidth}
        \includegraphics[width=\linewidth]{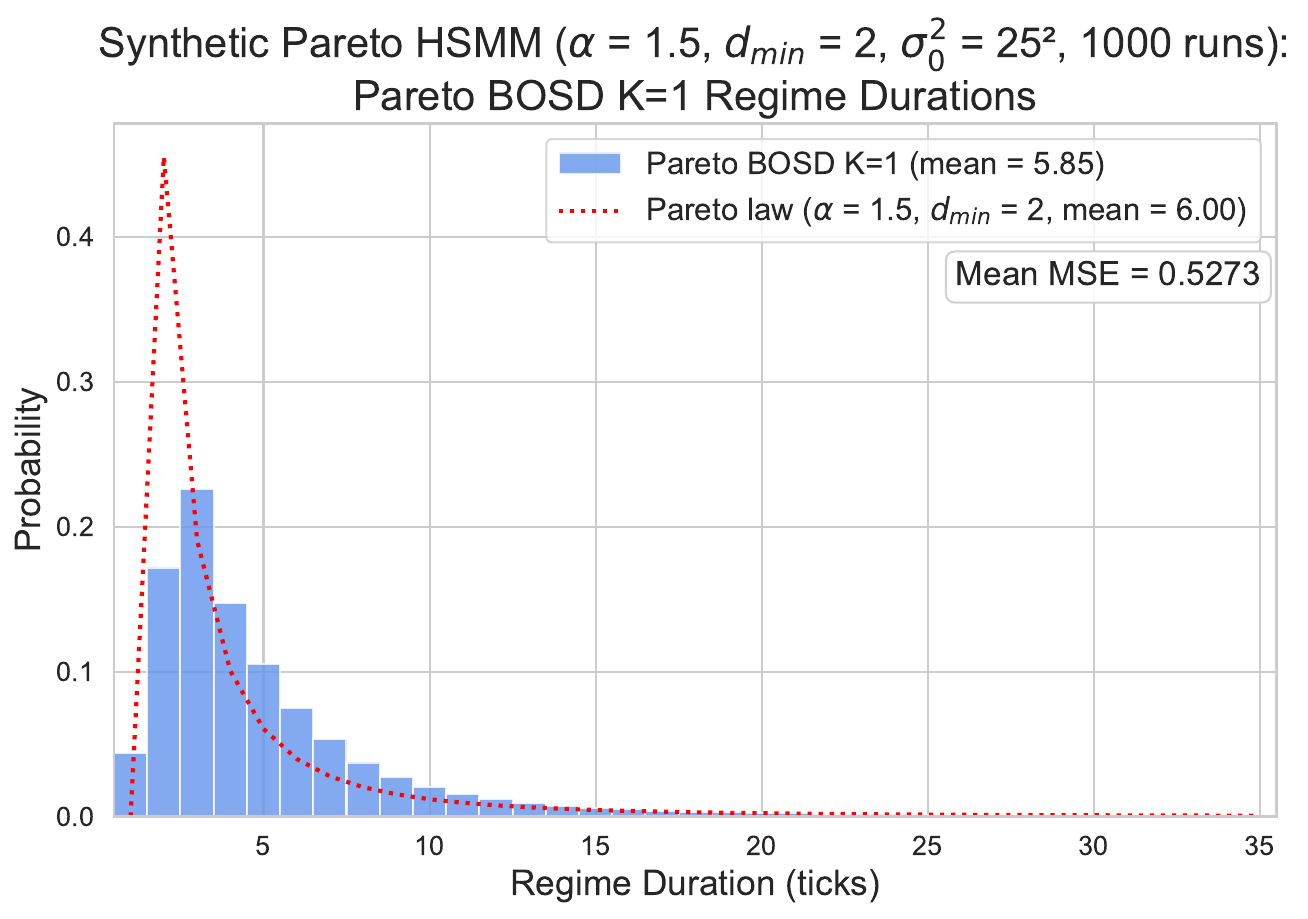}
        \caption{$\sigma_0 = 25$}
    \end{subfigure}

    \vspace{2pt}
    \begin{subfigure}[b]{0.42\linewidth}
        \includegraphics[width=\linewidth]{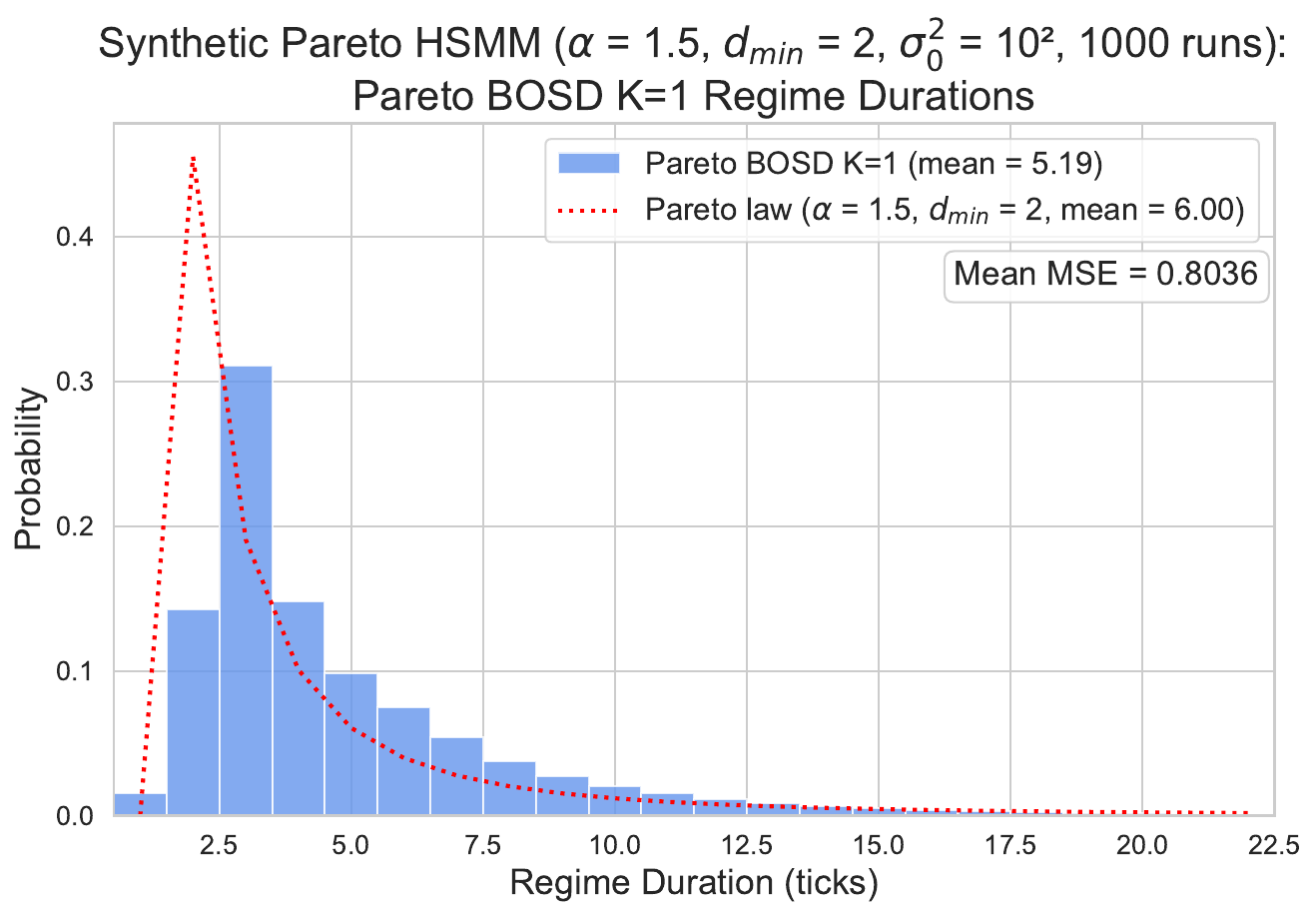}
        \caption{$\sigma_0 = 10$}
    \end{subfigure}\hfill
    \begin{subfigure}[b]{0.42\linewidth}
        \includegraphics[width=\linewidth]{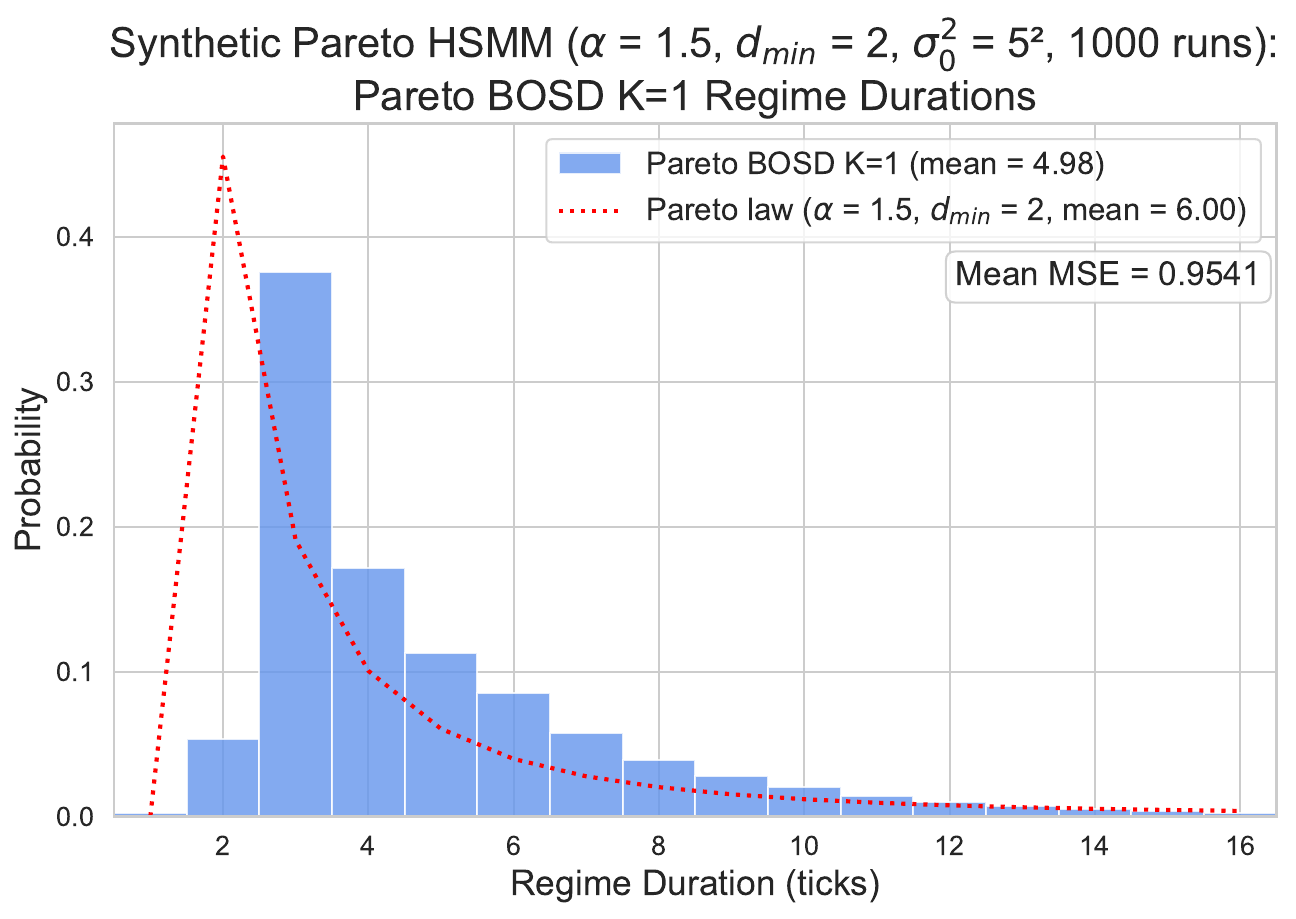}
        \caption{$\sigma_0 = 5$}
    \end{subfigure}
    \caption{Empirical histograms of BOSD-detected regime lengths against the theoretical Pareto distribution ($\alpha = 1.5$, $d_{\text{min}} = 2$). Averaged over $N = 1000$ runs.}
    \label{fig:pareto_durations_vs_theory}
\end{figure}

\begin{figure}[H]
    \centering
    \begin{subfigure}[b]{0.42\linewidth}
        \includegraphics[width=\linewidth]{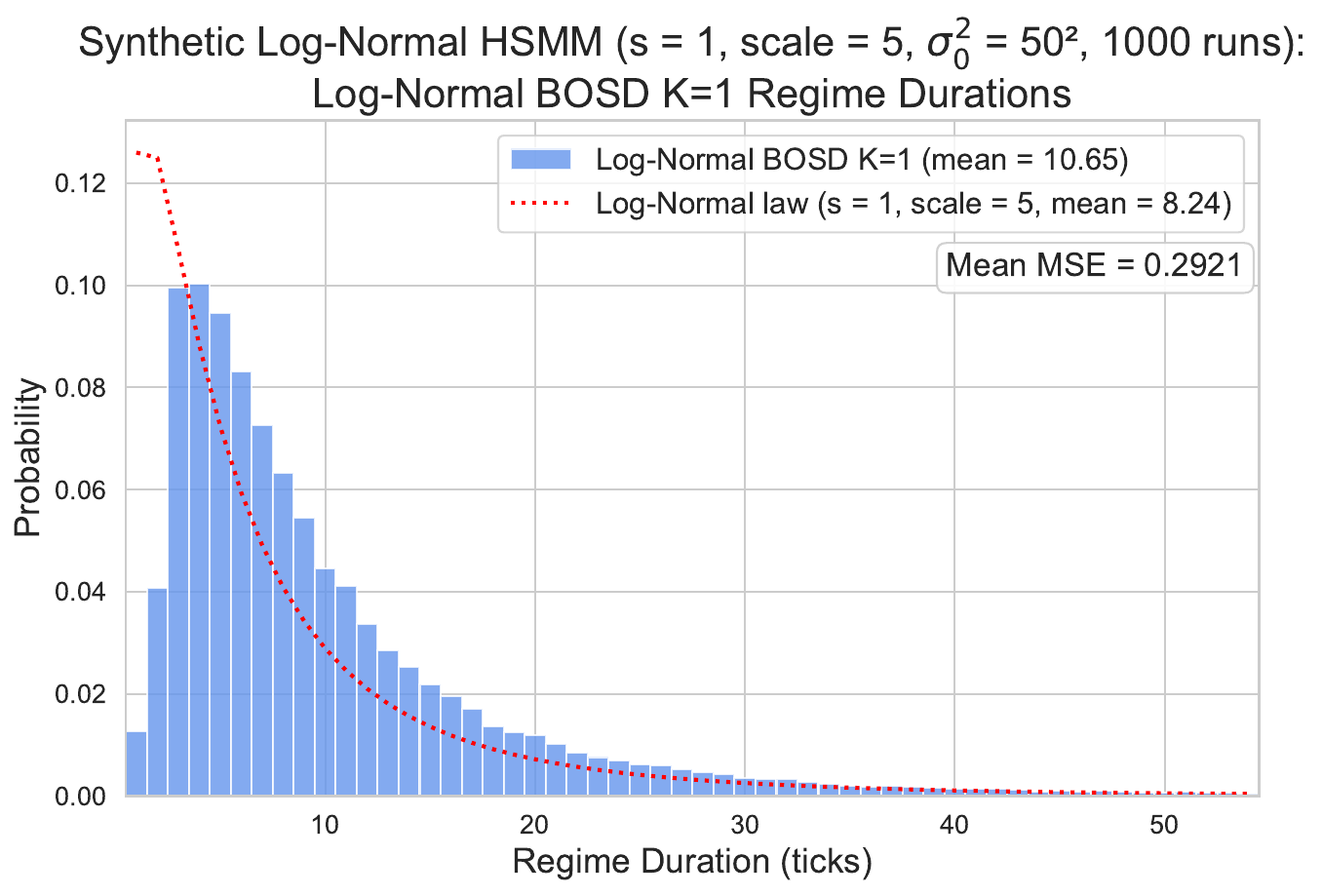}
        \caption{$\sigma_0 = 50$}
    \end{subfigure}\hfill
    \begin{subfigure}[b]{0.42\linewidth}
        \includegraphics[width=\linewidth]{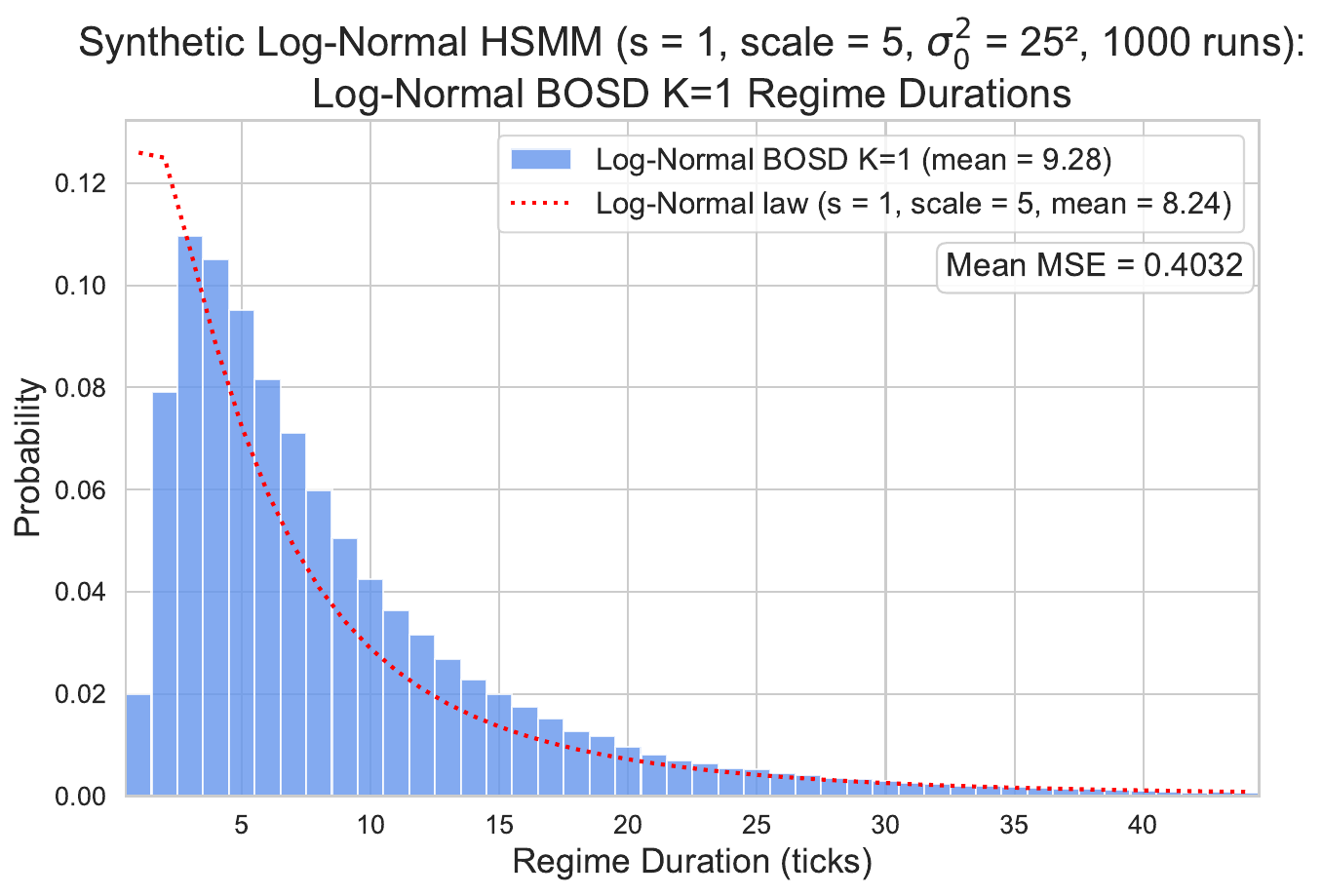}
        \caption{$\sigma_0 = 25$}
    \end{subfigure}

    \vspace{2pt}
    \begin{subfigure}[b]{0.42\linewidth}
        \includegraphics[width=\linewidth]{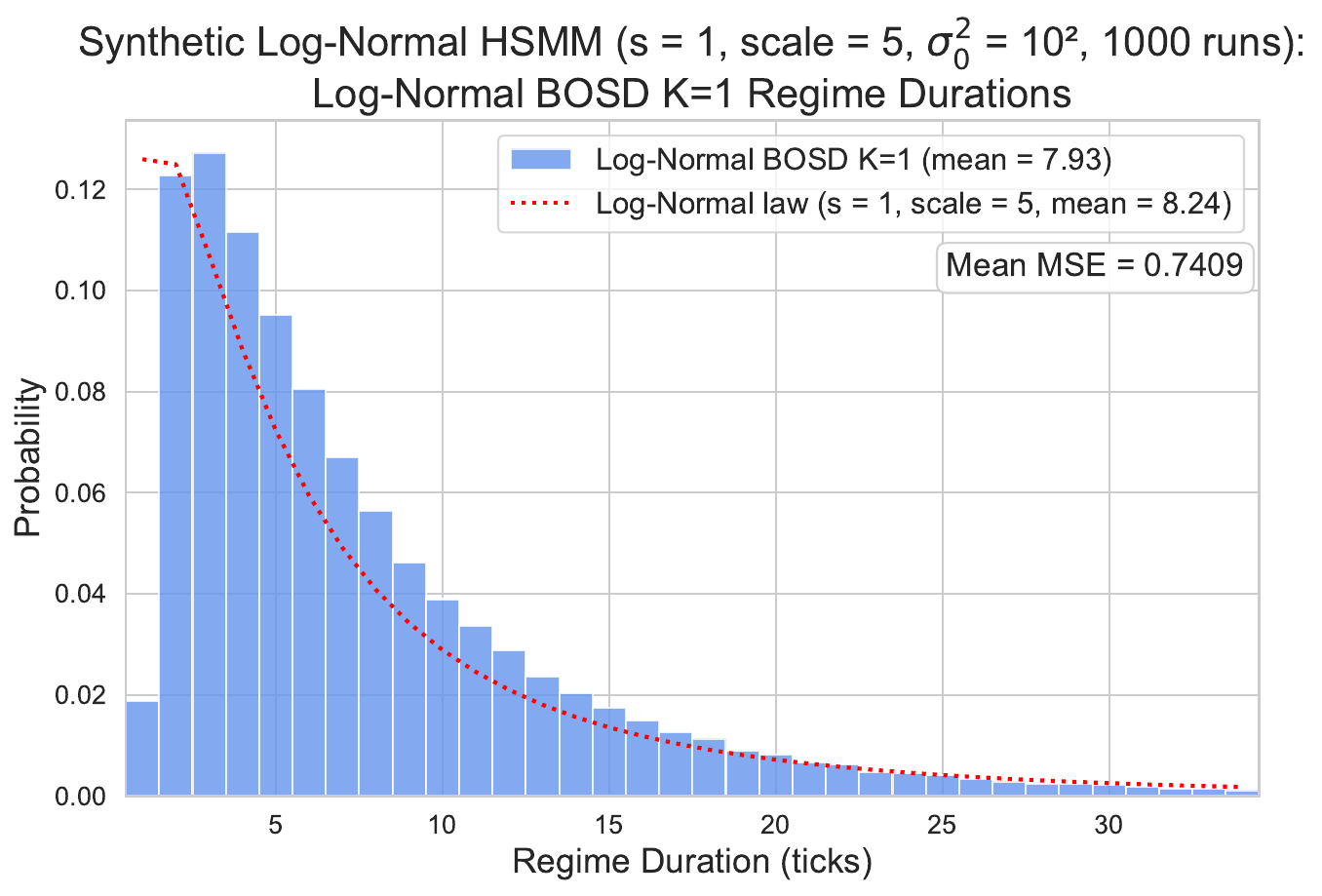}
        \caption{$\sigma_0 = 10$}
    \end{subfigure}\hfill
    \begin{subfigure}[b]{0.42\linewidth}
        \includegraphics[width=\linewidth]{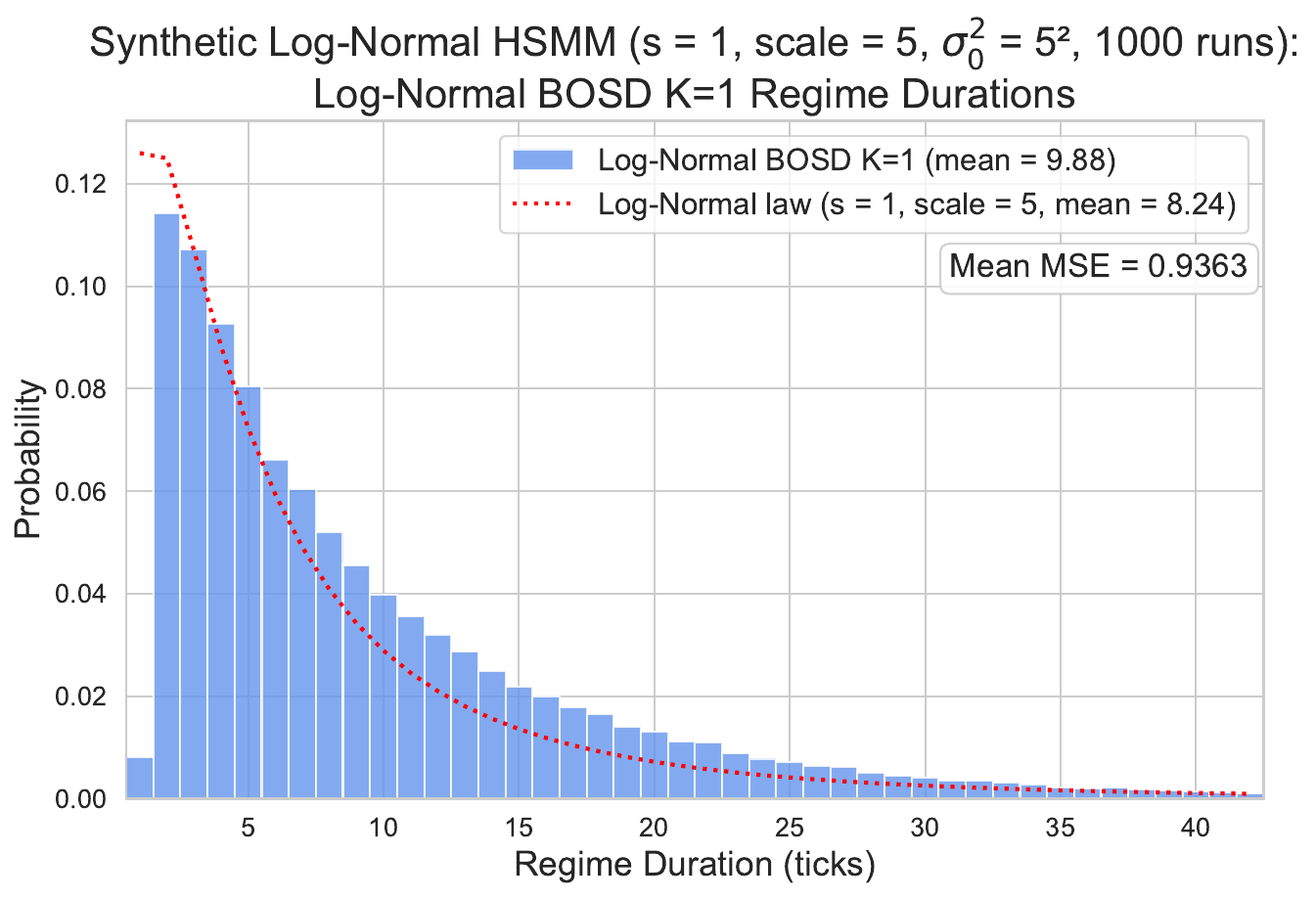}
        \caption{$\sigma_0 = 5$}
    \end{subfigure}
    \caption{Empirical histograms of BOSD-detected regime lengths against the theoretical (discretized) Log-Normal distribution ($s = 1$, scale $= 5$). Averaged over $N = 1000$ runs.}
    \label{fig:lognormal_durations_vs_theory}
\end{figure}

\newpage
\section{Posterior for a general known SPD noise matrix}\label{app:spd}

\textit{This appendix gives in full the derivation sketched in Section~\ref{sec:omega}, where we
argue that the Normal-Inverse-Gamma conjugacy of the BVAR holds for any known
symmetric positive-definite noise matrix $\boldsymbol{\Omega}$, and not only for a
diagonal one.}

\subsection*{Setting and notation}

Over a regime of run length $r$, stack the responses and regressors as in Section~\ref{sec:paramlearning}: $\mathbf{Y} \in \R^{(r+1)\times S}$, $\mathbf{X} \in \R^{(r+1)\times k}$, with $n := r+1$ observations. The model is
\begin{align*}
    \mathbf{Y} \mid \mathbf{c}, \sigma^2 &\sim \mathcal{MN}\big(\mathbf{X}\mathbf{c},\; \sigma^2\mathbf{I}_n,\; \boldsymbol{\Omega}\big), \\
    \mathbf{c} \mid \sigma^2 &\sim \mathcal{MN}\big(\mathbf{0},\; \sigma^2\mathbf{V}_c,\; \boldsymbol{\Omega}\big), \qquad
    \sigma^2 \sim \mathrm{InvGamma}(a_0, b_0),
\end{align*}
where $\mathcal{MN}(\mathbf{M}, \mathbf{U}, \mathbf{V})$ denotes the matrix-normal law with mean $\mathbf{M}$, row covariance $\mathbf{U}$ and column covariance $\mathbf{V}$. In vectorised form the coefficient prior reads $\mathrm{vec}(\mathbf{c}) \mid \sigma^2 \sim \mathcal{N}(\mathbf{0}, \sigma^2(\boldsymbol{\Omega}\otimes\mathbf{V}_c))$: this is the \textbf{coupled} prior, and it is what makes $\boldsymbol{\Omega}$ factor out below.

\subsection*{Posterior of the coefficients}

The joint density of $(\mathbf{Y}, \mathbf{c})$ given $\sigma^2$ is proportional to
$$\exp\Big\{-\tfrac{1}{2\sigma^2}\tr\big[\boldsymbol{\Omega}^{-1}(\mathbf{Y}-\mathbf{X}\mathbf{c})^T(\mathbf{Y}-\mathbf{X}\mathbf{c})\big]\Big\}
\exp\Big\{-\tfrac{1}{2\sigma^2}\tr\big[\boldsymbol{\Omega}^{-1}\mathbf{c}^T\mathbf{V}_c^{-1}\mathbf{c}\big]\Big\}.$$
Both exponents carry the same factor $\boldsymbol{\Omega}^{-1}$ inside the trace, so it can be collected. Expanding and completing the square in $\mathbf{c}$ gives
$$(\mathbf{c}-\hat{\mathbf{c}})^T\mathbf{P}(\mathbf{c}-\hat{\mathbf{c}}) + \mathbf{Y}^T\mathbf{Y} - \hat{\mathbf{c}}^T\mathbf{P}\hat{\mathbf{c}},
\qquad \mathbf{P} := \mathbf{V}_c^{-1} + \mathbf{X}^T\mathbf{X}, \quad \hat{\mathbf{c}} := \mathbf{P}^{-1}\mathbf{X}^T\mathbf{Y}.$$

\begin{proposition}[Posterior of $\mathbf{c}$, general SPD $\boldsymbol{\Omega}$]
Under the coupled prior, the conditional posterior of the coefficient matrix is matrix-normal, $\mathbf{c} \mid \mathbf{Y}, \sigma^2 \sim \mathcal{MN}(\hat{\mathbf{c}},\, \sigma^2\mathbf{P}^{-1},\, \boldsymbol{\Omega})$. Neither the posterior precision $\mathbf{P}$ nor the MAP estimate $\hat{\mathbf{c}}$ depends on $\boldsymbol{\Omega}$: it appears only as the column covariance.
\end{proposition}

With the \emph{plain} prior the two exponents would carry $\boldsymbol{\Omega}^{-1}$ and $\mathbf{I}_S$ respectively, the trace would not collapse, and $\hat{\mathbf{c}}$ would depend on $\boldsymbol{\Omega}$ through an unwieldy Kronecker expression --- still conjugate, but no longer reducible to a single $k\times k$ precision matrix.

\subsection*{Posterior of the variance and predictive}

Marginalising $\mathbf{c}$ out leaves an Inverse-Gamma kernel in $\sigma^2$ with
$$a_r = a_0 + \frac{S(r+1)}{2}, \qquad b_r = b_0 + \tfrac{1}{2}\Big(\mathrm{YY}_r - \tr\big[\boldsymbol{\Omega}^{-1}\mathbf{W}^T\mathbf{P}^{-1}\mathbf{W}\big]\Big),$$
with $\mathrm{YY}_r := \tr[\boldsymbol{\Omega}^{-1}\mathbf{Y}^T\mathbf{Y}]$, which is exactly the expression used in Section~\ref{sec:paramlearning}. The shape increment $S/2$ per observation is the one discussed in Section~\ref{sec:hyperlearning} in connection with the online learning of the hyperparameters. Marginalising both $\mathbf{c}$ and $\sigma^2$ against the new regressor $\boldsymbol{\phi}_t$, with $h_r = \boldsymbol{\phi}_t^T\mathbf{P}^{-1}\boldsymbol{\phi}_t$:

\begin{proposition}[UPM, general SPD $\boldsymbol{\Omega}$]
$$f(\mathbf{x}_t \mid \mathbf{x}_{1:(t-1)}, r_t) = \mathrm{Student}\text{-}t_{2a_r}\!\Big(\mathbf{x}_t \;\Big|\; \hat{\mathbf{c}}_r^T\boldsymbol{\phi}_t,\; \tfrac{b_r}{a_r}(1+h_r)\,\boldsymbol{\Omega}\Big),$$
identical in form to the Student-$t$ UPM \eqref{eq:student_t_upm}, with $\boldsymbol{\Omega}$ now arbitrary SPD.
\end{proposition}

\textbf{Consequence for the results on real data.} The location of the predictive is $\boldsymbol{\Omega}$-free, so the point forecast --- and hence the MSE --- is unaffected by the choice of $\boldsymbol{\Omega}$ except through the sufficient statistics; only the predictive \emph{density}, and therefore the log-likelihood and the learned $b_0$, are reweighted. This is precisely why the three specifications of $\boldsymbol{\Omega}$ tested on the bivariate order-flow data (Table~\ref{tab:exp23}) produce MSE values so close to one another.

\end{document}